\newif\ifsubmission\submissionfalse
\newif\ifcomments\commentsfalse
\newif\ifproofsinbody\proofsinbodytrue
\newif\iffullversion\fullversionfalse   
\newif\ificalpfullversion\icalpfullversiontrue 
\newif\ificalpshortversion\icalpshortversionfalse 
\newif\ificalpcorrectedversion\icalpcorrectedversionfalse

\ifsubmission
  \documentclass[USenglish]{lipics-v2016}
  \RequirePackage{etex}
  \usepackage{microtype}
\else
	\documentclass[11pt]{article}
  \usepackage[margin=1in]{geometry}
  \usepackage{amsfonts}
  \usepackage{amsmath, amssymb}
  \RequirePackage{etex}
  \usepackage{microtype}
\fi

\usepackage[dvipsnames]{xcolor}
\usepackage{stmaryrd}
\usepackage{url}
\usepackage{hyperref}
\usepackage{tikz}
\usepackage{algorithm}
\usepackage{pgfplots}
\usepackage{subcaption}
\usepackage{enumerate}
\usepackage{environ,refcount}
\usepackage{wrapfig}

\usepackage{color}

\usepackage{enumitem}
\setlist[itemize]{noitemsep, topsep=2pt, leftmargin=.25in}
\setlist[enumerate]{noitemsep, topsep=2pt, leftmargin=.25in}
\setlist[description]{leftmargin=\parindent,labelindent=0pt}

  \globtoksblk\savedenvtoks{100}  
  \newcounter{savedenvcount}

  \NewEnviron{savedenv}[1][]{%
    \refstepcounter{savedenvcount}%
    \if\relax\detokenize{#1}\relax
    \else
      \label{#1}%
    \fi
    \global\toks\numexpr\savedenvtoks+\value{savedenvcount}\relax=\expandafter{\BODY}%
  }
  \toks\savedenvtoks={??}

  \newcommand{\printsaved}[1]{%
    \the\toks\numexpr\savedenvtoks+\getrefnumber{#1}\relax
  }

  \makeatletter
  \newcommand{\printallsaved}{%
    \@tempcnta=\z@
    \loop
      \ifnum\@tempcnta<\value{savedenvcount}
      \advance\@tempcnta\@ne
      \the\toks\numexpr\savedenvtoks+\@tempcnta\relax\par
    \repeat
  }
  \makeatother

\newcommand{\subh}[1]{\par \vspace{4pt} \noindent \textbf{\textsf{#1}}}

\ifsubmission
  \theoremstyle{plain}
  \newtheorem{conjecture}[theorem]{Conjecture}
	\newtheorem{property}[theorem]{Property}
  \newtheorem{gconstruction}[theorem]{Graph Construction}
  \theoremstyle{definition}
  \newtheorem*{openquestion}{Open Question}
  \renewcommand{\paragraph}[1]{\subh{#1}}
\else
	\usepackage{amsthm}
	\theoremstyle{plain}
	\newtheorem{theorem}{Theorem}[section]
	
	\newtheorem{lemma}[theorem]{Lemma}
	\newtheorem{corollary}[theorem]{Corollary}
	\newtheorem{definition}[theorem]{Definition}

	\newtheorem{gconstruction}[theorem]{Graph Construction}
  \newtheorem*{openquestion}{Open Question}

	\theoremstyle{definition}
	\newtheorem*{remark}{Remark}
\fi

\newenvironment{informalthm}
  {\medskip\noindent{Theorem (informal).}}
  {\medskip}

\newenvironment{informalthmwithnum}[1]
  {\medskip\noindent\textbf{\textsf{Informal version of Theorem #1.}}}
  {\medskip}

\newenvironment{informaldef}
  {\medskip\noindent{\bf Definition (informal).}}
  {\medskip}

\usepackage{xspace}

\usepackage{mathtools}

\DeclarePairedDelimiter\floor{\lfloor}{\rfloor}

\renewcommand{\sec}{\kappa}
\renewcommand{\deg}{\delta}

\newcommand{\pcost}{\mathsf{p{\mh}cc}_{\alpha}}
\newcommand{\pcostm}{\mathsf{p{\mh}cc}_{\alpha}^M}
\newcommand{\sspace}{\mathbf{P}_{\rm ss}}
\newcommand{\sspacem}{\mathbf{P}_{\rm ss}}
\newcommand{\stime}{\mathbf{P}_{\rm opt{\mh}ss}}
\newcommand{\stimem}{\mathbf{P}_{\rm opt{\mh}ss}}
\newcommand{\scost}{\mathbf{P}_{\rm s}}
\newcommand{\scostm}{{\mathbf{P}_{\rm s}}}

\newcommand{\p}{\scost}

\newcommand{\ti}{\mathsf{Time}}

\newcommand{\s}{\mathcal{P}}
\newcommand{\ps}{\mathcal{P}^{||}}

\newcommand{\g}{G}

\newcommand{\glong}{G_{n,\delta}=(V_{n},E_{n})}

\newcommand{\sg}[1]{G}
\newcommand{\gf}{\mathbb{G}}
\newcommand{\sgf}{\mathbb{G}}
\newcommand{\gfsec}{\mathbb{G}_{\delta}}

\newcommand{\cc}{\text{CC}}
\newcommand{\cca}{\text{CC}^{\alpha}}
\newcommand{\pyr}{\Pi_{h}^C}

\makeatletter
\def\moverlay{\mathpalette\mov@rlay}
\def\mov@rlay#1#2{\leavevmode\vtop{%
   \baselineskip\z@skip \lineskiplimit-\maxdimen
   \ialign{\hfil$\m@th#1##$\hfil\cr#2\crcr}}}
\newcommand{\charfusion}[3][\mathord]{
    #1{\ifx#1\mathop\vphantom{#2}\fi
        \mathpalette\mov@rlay{#2\cr#3}
      }
    \ifx#1\mathop\expandafter\displaylimits\fi}
\makeatother

\newcommand{\bigcupdot}{\charfusion[\mathop]{\bigcup}{\cdot}}

\usepackage[
	n,
	operators,
	advantage,
	sets,
	adversary,
	landau,
	probability,
	notions,	
	logic,
	ff,
	mm,
	primitives,
	events,
	complexity,
	asymptotics,
	keys]{cryptocode}

\newcommand{\cA}{{\cal A}}
\newcommand{\cB}{{\cal B}}

\newcommand{\cE}{{\cal E}}
\newcommand{\cF}{{\cal F}}

\newcommand{\cH}{{\cal H}}

\newcommand{\cP}{{\cal P}}

\newcommand{\ff}{\mathfrak{f}}
\newcommand{\bq}{\mathbf{q}}

\let\epsilon=\varepsilon

\renewcommand{\emptyset}{\varnothing}
\renewcommand{\phi}{\varphi}
\renewcommand{\theta}{\vartheta}
\newcommand{\eps}{\varepsilon}

\newcommand{\Seek}{\mathsf{Seek}}

\newcommand{\zo}{\{0,1\}}
\newcommand{\mh}{\textrm{-}}

\newcommand{\defeq}{=}
\newcommand{\RO}{\mathcal{O}}
\newcommand{\OSet}{\mathbb{O}}

\newcommand{\LabelSet}{\mathfrak{L}}
\newcommand{\lab}{\mathsf{label}}
\renewcommand{\state}{\sigma}
\newcommand{\cost}{\mathsf{c{\mh}mem}}

\newcommand{\satscore}{\mathsf{mem}}

\newcommand{\sattime}{\mathsf{s{\mh}mem}}

\newcommand{\comp}{\mathsf{comp}}
\newcommand{\acomp}{\mathsf{a{\mh}comp}}
\newcommand{\pred}{\mathsf{pred}}
\newcommand{\indeg}{\mathsf{indeg}}
\newcommand{\sink}{\mathsf{sink}}
\newcommand{\prelab}{\mathsf{pre{\mh}lab}}

\newcommand{\epfm}{\mathsf{epf{\mh}magic}}
\newcommand{\wsize}[1]{\talloblong#1\talloblong}
\newcommand{\Mod}[1]{\ \mathrm{mod}\ #1}
\newcommand{\mbound}{\mathfrak{M}}

\newcommand{\crypt}{{\tt crypt}\xspace}

\ifcomments
\newcommand{\sunoo}[1]{{\color{WildStrawberry}/* Sunoo: #1 */}}
\renewcommand{\qq}[1]{{\color{RoyalBlue}/* Quanquan: #1 */}}
\newcommand{\tadge}[1]{{\color{OliveGreen}/* Tadge: #1 */}}
\else
\newcommand{\sunoo}[1]{\ignorespaces}
\renewcommand{\qq}[1]{\ignorespaces}
\newcommand{\tadge}[1]{\ignorespaces}
\fi

\newcounter{section-preserve}
\newcounter{theorem-preserve}
\newcommand{\blank}[1]{}
\newtoks\magicAppendix
\magicAppendix={}
\newtoks\magictoks
\newif\iflater
\laterfalse
\long\def\later#1{\magictoks={#1}%
  \edef\magictodo{\noexpand\magicAppendix={\the\magicAppendix \par
    \the\magictoks%
  }}
  \magictodo}
\long\def\both#1{\magictoks={#1}%
  \edef\magictodo{\noexpand\magicAppendix={\the\magicAppendix \par
    \noexpand\setcounter{theorem-preserve}{\noexpand\arabic{theorem}}%
    \noexpand\setcounter{theorem}{\arabic{theorem}}%
    \noexpand\setcounter{section-preserve}{\noexpand\arabic{section}}%
    \noexpand\setcounter{section}{\arabic{section}}%
    \noexpand\let\noexpand\oldsection=\noexpand\thesection
    \noexpand\def\noexpand\thesection{\thesection}
    \noexpand\let\noexpand\oldlabel=\noexpand\label
    \noexpand\let\noexpand\label=\noexpand\blank
    \the\magictoks%
    \noexpand\setcounter{theorem}{\noexpand\arabic{theorem-preserve}}%
    \noexpand\setcounter{section}{\noexpand\arabic{section-preserve}}%
    \noexpand\let\noexpand\thesection=\noexpand\oldsection
    \noexpand\let\noexpand\label=\noexpand\oldlabel
  }}
  \magictodo
  \the\magictoks}
\def\magicappendix{\latertrue \the\magicAppendix}

\ifproofsinbody\else
\let\pf\proof
\let\endpf\endproof
\NewEnviron{mypf}[1][]{
  
  \ifproofsinbody\pf[#1]\BODY\endpf\fi
}
    \renewenvironment{proof}[1][]
      {\mypf[#1]}
      {\endmypf}
\fi

\let\ft\footnote
\renewcommand{\footnote}[1]{\ificalpfullversion\ft{#1}\else\ignorespaces\fi}

\makeatletter
\newcommand{\removelatexerror}{\let\@latex@error\@gobble}
\makeatother

\begin{document}

\ificalpfullversion
  \title{Static-Memory-Hard Functions and Nonlinear Space-Time Tradeoffs via Pebbling}
\else
  \title{Static-Memory-Hard Functions and Nonlinear Space-Time Tradeoffs via Pebbling}
\fi
\ifsubmission
\else
  \author{Thaddeus Dryja \and Quanquan C. Liu \and Sunoo Park}
  \date{MIT}
\fi

\maketitle

\begin{abstract}
Pebble games were originally formulated to study time-space tradeoffs in computation, modeled by games played on directed acyclic graphs (DAGs). Close connections between pebbling and cryptography have been known for decades. A series of recent research starting with (Alwen and Serbinenko, STOC 2015) has deepened our understanding of the notion of \emph{memory-hardness} in cryptography --- a useful property of hash functions for deterring large-scale password-cracking attacks --- and has shown memory-hardness to have intricate connections with the theory of graph pebbling. 
Definitions of memory-hardness are not yet unified in this somewhat nascent field, however, and the guarantees proven are with respect to a range of proposed definitions.

In this work, we improve upon two main limitations of existing models of memory-hardness.
First, existing measures of memory-hardness only account for \emph{dynamic} (i.e., runtime) memory usage, and do not consider \emph{static} memory usage. We propose a new definition of \emph{static-memory-hard} function (SHF) which takes into account static memory usage and allows the formalization of larger memory requirements for efficient functions, than in the dynamic setting (where memory usage is inherently bounded by runtime). We then give two SHF constructions based on pebbling; to prove static-memory-hardness, we define a new pebble game (``\emph{black-magic pebble game}''), and new graph constructions with optimal complexity under our proposed measure.
Secondly, existing memory-hardness models implicitly consider \emph{linear} tradeoffs between the costs of time and space. We propose a new model to capture \emph{nonlinear} time-space trade-offs and prove that nonlinear tradeoffs can in fact cause adversaries to employ different strategies from linear tradeoffs.

Finally, as an additional contribution of independent interest, we present an asymptotically tight graph construction that achieves the best possible space complexity up to $\log{\log{n}}$-factors for an existing memory-hardness measure called \emph{cumulative complexity} in the sequential pebbling model.\qq{I changed the last sentence here.}

\end{abstract}

\ificalpfullversion
  \newpage
  \tableofcontents
  \newpage
\fi

\section{Introduction}\label{sec:intro}

Pebble games were originally formulated to model time-space tradeoffs by a game played on DAGs. Generally, a DAG can be thought to represent a computation graph where each node is associated with some computation and a pebble placed on a node represents saving the result of its computation in memory. Thus, the number of pebbles represents the amount of memory necessary to perform some set of computations. The natural complexity measures to optimize in this game is the minimum number of pebbles used, as well as the minimum amount of time it takes to finish pebbling all the nodes; these goals correspond with minimizing the amount of memory and time of computation. 

\ificalpfullversion
Pebble games were first introduced to study programming languages and compiler construction~\cite{PH70} but have since then been used to study a much broader range of tasks such as register allocation~\cite{Sethi75}, proof complexity~\cite{ARNV17, Nor12}, time-space tradeoffs in Turing machine computation~\cite{Cook73,HPV77}, reversible computation~\cite{Bennett89}, circuit complexity~\cite{Potechin17}, and time-space tradeoffs in various algorithms such as FFT~\cite{Tompa81}, linear recursion~\cite{Chandra73,SS79}, matrix multiplication~\cite{Tompa81}, and integer multiplication~\cite{SS79b} in the RAM as well as the external memory model~\cite{HK81}. To see a more comprehensive survey of the results in pebbling up to the last couple of years, see~\cite{Pippenger82} up to the 1980s and~\cite{Nor15} up to 2015.
\fi

The relationship between pebbling and cryptography has been a subject of research interest for decades, which has enjoyed renewed activity in the last few years.
A series of recent works \cite{AB16,ABH17,ABP17,ABP17sustained,AS15,AT17,ACPRT16,AGKKOPRRR16,BZ16,BZ17} has deepened our understanding of the notion of \emph{memory-hardness} in cryptography, and has shown memory-hardness to have intricate connections with the theory of graph pebbling. 

\emph{Memory-hard functions (MHFs)} have garnered substantial recent interest as a security measure against adversaries trying to perform attacks at scale, particularly in the ubiquitous context of password hashing. Consider the following scenario: hashes of user passwords are stored in a database,\footnote{In practice, the password should first be concatenated with a random user-specific string called a \emph{salt}, and then hashed. The salt is stored in the database alongside the hash to deter \emph{dictionary attacks}.} and when a user enters a password $p$ to log in, her computer sends $H(p)$ to the database server, and the server compares the received hash to its stored hash for that user's account. For a normal user, it would be no problem if hash evaluation were to take, say, one second. An attacker trying to guess the password by brute-force search, on the other hand, would want to try orders of magnitude more passwords, so a one-second hash evaluation could be prohibitively expensive for the attacker.

The evolution of password hashing functions has been something of an arms race for decades, starting with the ability to increase the number of rounds in the DES-based unix \crypt function to increase its computation time---a feature that was used for exactly the above purpose of deterring large-scale password-cracking.
Attackers responded by building special-purpose circuits for more efficient evaluation of \crypt, resulting in a gap between the evaluation cost for an attacker and the cost for an honest user.\footnote{E.g., \cite{CB02} discusses FPGA-based attacks on DES.}

A promising approach to mitigating this asymmetry in cost between hash evaluation on general- and special-purpose hardware is to increase the use of \emph{memory} in the password hashing function.  Memory is implemented in standardized ways which have been highly optimized, and memory chips are widely regarded to be an interchangeable commodity.  Commonly used forms of memory --- whether on-die SRAM cache, DRAM, or hard disks --- are already optimized for the purpose of data I/O operations; and while there is active research in improving memory access times and costs, progress is and has been relatively incremental.  This state of affairs sets up a relatively ``even playing field,'' as the normal user and the attacker are likely to be using memory chips of similar memory access speed. While an attacker may choose to buy more memory, the cost of doing so scales linearly with the amount purchased.

The designs of several MHFs proposed to date (e.g., \cite{Per09,AS15,AB16,ACPRT16,ABP17}) have proven memory-hardness guarantees by basing their hash function constructions on DAGs, and using space complexity bounds from graph pebbling. Definitions of memory-hardness are not yet unified in this somewhat nascent field, however --- the first MHF candidate was proposed only in 2009 \cite{Per09} --- and the guarantees proven are with respect to a range of definitions. The ``cumulative complexity''-based definitions of \cite{AS15} have enjoyed notable popularity, but some of their shortcomings have been pointed out by subsequent work proposing alternative more expressive measures, in particular, \cite{ABP17sustained,AT17}.

\subh{Our Contribution}
In this work, we improve upon two main limitations of existing models of memory-hardness as described in (1) and (2) below. 
We also provide an additional contribution of separate interest, described in (3).
\begin{enumerate}
	\item Existing measures of memory-hardness only account for \emph{dynamic} (i.e., runtime) memory usage, and do not consider \emph{static} memory usage. Among other things, this means that the amount of memory usage is inherently upper-bounded by runtime; in contrast, counting static memory would potentially allow quantification of much larger memory requirements. To address this, we introduce \emph{static-memory-hard functions} (SHFs) (Definition~\ref{def:hardness}). To prove such properties of our functions, we formulate a new type of pebble game called the \emph{black-magic pebble game} (Definition~\ref{def:blackmagic}) and prove properties of the space complexity of this game for new graphs (Graph Constructions~\ref{def:wrappyramid} and \ref{def:longer-time}). The black-magic pebble game may additionally be of independent interest for the pebbling literature.

	Based on our new graph constructions, we construct SHFs with provable guarantees on sustained memory usage, as follows.
	Graph Construction~\ref{def:longer-time} gives a better asymptotic guarantee, but Graph Construction~\ref{def:wrappyramid} has the advantage of simplicity in practice, and indeed our full version describes an implementation of the latter construction.

	\begin{informalthmwithnum}{\ref{thm:wraparound2}}
		The ``cylinder graph'' (Graph Construction~\ref{def:wrappyramid})
		can be used to construct an SHF with static memory requirement $\Lambda\in\Theta(\sqrt{n})$ where $n$ is the number of nodes in the graph, such that any adversary using non-trivially less \emph{static} memory than $\Lambda$ must incur at least $\Lambda$ \emph{dynamic} memory usage for at least $\Theta(\sqrt{n})$ steps.
	\end{informalthmwithnum}

	\begin{informalthmwithnum}{\ref{thm:highest-memory2}}
		Graph Construction~\ref{def:longer-time}
		can be used to construct an SHF with static memory requirement $\Lambda\in\Theta(\sqrt{n})$ where $n$ is the number of nodes in the graph, such that any adversary using non-trivially less \emph{static} memory than $\Lambda$ must incur at least $\Lambda$ \emph{dynamic} memory usage for at least $\Theta(n)$ steps.
	\end{informalthmwithnum}

	\item Existing measures of memory-hardness implicitly assume a linear trade-off between the costs of space and time. This model precludes situations where the relative costs of space and time might be more unbalanced (e.g., quadratic or cubic). We demonstrate that this modeling limitation is significant, as follows.

	\begin{informalthmwithnum}{\ref{thm:different-strategies}}
		 There exist graphs for which an adversary facing a linear space-time cost trade-off would in fact employ a \emph{different pebbling strategy} from one facing a cubic trade-off. 
	\end{informalthmwithnum}

	To remedy this shortcoming, we define \emph{graph-optimal} variants (defined in Section~\ref{sec:pebbling-defs}) 
	that \emph{explicitly} model the relative cost of space and time. These can be seen as extending the main memory-hardness measures in the literature (namely, \emph{cumulative complexity} and \emph{sustained memory complexity}).

	\begin{informalthmwithnum}{\ref{lem:cc-bound}}
		Given any graph construction $G = (V, E)$, there exists a pebbling strategy that is less expensive asymptotically than any strategy using a number of pebbles asymptotically equal to the number of nodes in the graph for any time-space tradeoff.
	\end{informalthmwithnum}
	\item We give the first graph construction that is tight, up to $\log{\log{n}}$-factors, to the optimal cumulative complexity that can be achieved for any graph (upper bound due to \cite{ABP17,ABP17sustained}).
	
	\begin{informalthmwithnum}{\ref{thm:cca-lb}}
		There exists a family of graphs where the cumulative complexity of any constant in-degree graph with $n$ nodes in the family is $\Theta\left(\frac{n^2 \log{\log{n}}}{\log{n}}\right)$ which is asymptotically tight to the upper bound of $\Theta\left(\frac{n^2 \log{\log{n}}}{\log{n}}\right)$ given in~\cite{ABP17,ABP17sustained} in the sequential pebbling model.
	\end{informalthmwithnum}
\end{enumerate}

\ificalpfullversion
	\subsection{Background on graph pebbling}

	The standard \emph{black pebble game} is parametrized by a directed acyclic graph (DAG) and a special subset of its nodes (called the \emph{target set}). In the game, an unlimited supply of ``pebbles'' is made available to a player, who must place and remove pebbles on the nodes of the DAG in a sequence of moves according to the following two rules.
	\begin{enumerate}
		\item A pebble may be placed or moved onto a node only if all of its predecessors have already been pebbled. (In particular, pebbles may be placed on source nodes at any time.)
		\item Any pebble can be removed from the graph at any time.
	\end{enumerate}
	The goal of the game is to arrive at a state where every pebble in the target set is covered by a pebble. Often, the target set is the set of the sink nodes.

	The pebbling literature, starting with \cite{PH70,Sethi75,Cook73,HPV77}, has established a number of complexity measures describing the complexity of pebbling: e.g., measuring the minimum number of pebbles that must be used to achieve a complete pebbling, or the minimum number of moves needed. In the literature, there are several variants of the game, including sequential and parallel (depending on whether many pebbles can be placed in a single move), and versions where other different types of pebbles are used (such as the red-blue pebble game~\cite{HK81} and the black-white pebble game~\cite{CS74}). In this work, our results are stated and proven in the context of constant in-degree graphs for simplicity; however, most of our results extend straightforwardly to non-constant in-degree graphs.

	\subh{Graph pebbling and memory-hardness}
	Graph pebbling algorithms can be used to construct hash functions in the (parallel) random oracle model. This paradigm has been used by prior constructions of memory-hard hashing \cite{AS15} as well as other prior works \cite{DKW11}. 

	Informally, the idea to ``convert'' a graph into a hash function is to associate with each node $v$ a string called a \emph{label}, which is defined to be $\RO(v,\pred(v))$ where $\RO$ is a random oracle and $\pred(v)$ is the list of labels of predecessors of $v$. For source nodes, the label is instead defined to be $\RO(v,\zeta)$ for a string $\zeta$ which is an input to the hash function. The output of the hash function is defined to be the list of labels of target nodes. Intuitively, since the label of a node cannot be computed without the ``random'' labels of all its predecessors, any algorithm computing this hash function must move through the nodes of the graph according to rules very similar to those prescribed by the pebbling game; and therefore, the memory requirement of computing the hash function roughly corresponds to the pebble requirement of the graph. Thus, proving lower bounds on the pebbling complexity of graph families has useful implications for constructing provably memory-hard functions.

	In our setting, in contrast to previous work, we employ a variant of the above technique: the string $\zeta$ is a fixed parameter of our hash function,
	and the input to the hash function instead specifies the indices of the target nodes whose labels are to be outputted.
\fi

\ificalpshortversion
	\subh{Discussion and related work}
	The original paper proposing memory-hard functions \cite{Per09} suggested a very simple measure: the minimum amount of memory necessary to compute the hash function. It was subsequently observed that a major drawback of this measure is that it does not distinguish between functions $f$ and $g$ with the same peak memory usage, even if the peak memory lasts a long time in evaluating $f$ and is just fleeting in evaluating $g$ (see the figure in the attached full version). This is significant as the latter type of function is much better for a password-cracking adversary.

	Towards addressing this, \cite{AS15} put forward the notion of \emph{cumulative complexity} (CC), a complexity measure on graphs. CC was adopted by several subsequent works as a canonical measure of memory-hardness. CC measures the \emph{cumulative} memory usage of a graph pebbling function evaluation: that is, the sum of memory usage over all time-steps of computation. In other words, this is the area under a graph of memory usage against time. 

	Yet more recently, Alwen, Blocki, and Pietrzak~\cite{ABP17} proposed a new measure of memory complexity, which captures not only the cumulative memory usage over time (as does CC), but goes further and captures the amount of time for which a particular level of memory usage is sustained. Our SHF definition also captures \emph{sustained} memory usage: we propose a definition of capturing the duration for which a given amount of memory is required, designed to capture static as well as dynamic memory requirements.

	More detailed discussion of notions from related work including \emph{depth-robust graphs} \cite{AB16,ABP17}, \emph{sustained memory complexity} \cite{ABP17sustained}, \emph{core-area memory ratio} \cite{BK15,AB16,RD17}, and \cite{AT17}'s more general theory of moderately hard functions is given in the attached full version due to space constraints.

	\ificalpfullversion
		\subh{Outline}
	\fi
	\ificalpshortversion

		\medskip
		
	\fi
	All proofs are deferred to the full version due to space constraints. The full version is attached to this submission as the attached full version.
\fi

\ificalpfullversion
	\subsection{Discussion and Related Work}

	The original paper proposing memory-hard functions \cite{Per09} suggested a very simple measure: the minimum amount of memory necessary to compute the hash function. It was subsequently observed that a major drawback of this measure is that it does not distinguish between functions $f$ and $g$ with the same peak memory usage, even if the peak memory lasts a long time in evaluating $f$ and is just fleeting in evaluating $g$ (Figure \ref{fig:mem1}). This is significant as the latter type of function is much better for a password-cracking adversary. In particular, pipelining the evaluation of the latter type of function would allow reuse of the same memory for many function evaluations at once, effectively reducing the adversary's amortized memory requirement by a factor of the number of concurrent executions (Figure \ref{fig:mem2}).

	\pgfmathdeclarefunction{gauss}{2}{%
	\pgfmathparse{1/(#2*sqrt(2*pi))*exp(-((x-#1)^2)/(2*#2^2))}%
	}
	\pgfplotsset{ticks=none}
	\begin{center}
		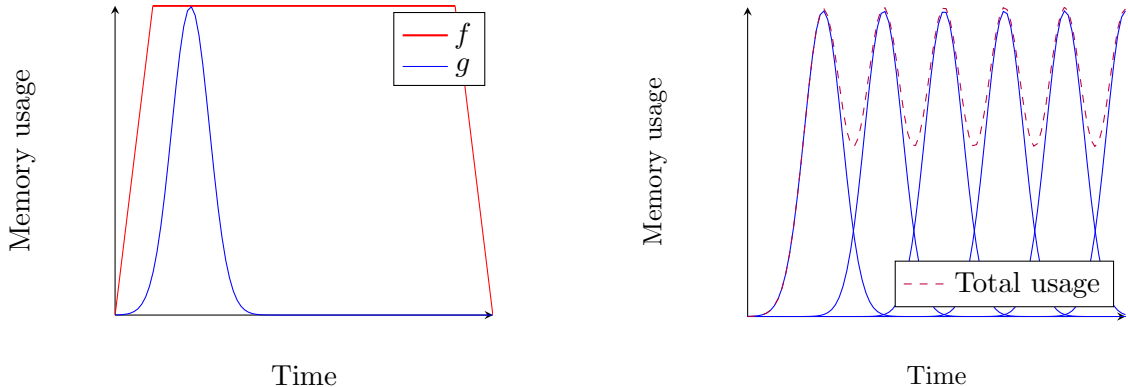
\begin{figure}[ht!]
			\begin{subfigure}[t]{0.5\textwidth}
				\begin{tikzpicture}
				\begin{axis}[
					width=0.8\textwidth,
					axis lines = left,
					xlabel = {Time},
					ylabel = {Memory usage},
				]
				\addplot [
					domain=0:1, 
					samples=80, 
					color=red,
					forget plot,
				]
				{0.8*x};
				\addplot [
					domain=9:10, 
					samples=80, 
					color=red,
					forget plot,
				]
				{0.8-0.8*(x-9)};
				\addplot[const plot, no marks, thick, color=red] coordinates {(1,0.8) (9,0.8)} node[above,pos=.57,black] {$F_x$};
				\addlegendentry{$f$}
				\addplot [
					domain=0:10, 
					samples=100, 
					color=blue,
					]
					{gauss(2,0.5)};
				\addlegendentry{$g$}
				\end{axis}
				\end{tikzpicture}
				\caption{Functions with the same peak memory usage}
				\label{fig:mem1}
			\end{subfigure}%
			~
			\begin{subfigure}[t]{0.5\textwidth}
				\begin{tikzpicture}
				\begin{axis}[
					width=0.8\textwidth,
					axis lines = left,
					xlabel = {\small Time},
					ylabel = {\small Memory usage},
					legend pos = south east,
				]
				\addplot [
					domain=0:10, 
					samples=100, 
					color=blue,
					forget plot,
					]
					{gauss(2,0.5)};
				\addplot [
					domain=0:10, 
					samples=100, 
					color=blue,
					forget plot,
					]
					{gauss(3.6,0.5)};
				\addplot [
					domain=0:10, 
					samples=100, 
					color=blue,
					forget plot,
					]
					{gauss(5.2,0.5)};
				\addplot [
					domain=0:10, 
					samples=100, 
					color=blue,
					forget plot,
					]
					{gauss(6.8,0.5)};
				\addplot [
					domain=0:10, 
					samples=100, 
					color=blue,
					forget plot,
					]
					{gauss(8.4,0.5)};
				\addplot [
					domain=0:10, 
					samples=100, 
					color=blue,
					forget plot,
					]
					{gauss(10,0.5)};
				\addplot [
					domain=0:10, 
					samples=100, 
					color=purple,
					dashed,
					]
					{gauss(2,0.5)+gauss(3.6,0.5)+gauss(5.2,0.5)+gauss(6.8,0.5)+gauss(8.4,0.5)+gauss(10,0.5)}; 
				\addlegendentry{Total usage}
				\end{axis}
				\end{tikzpicture}
				\caption{Pipelined evaluations of $g$ (reusing memory)}
				\label{fig:mem2}
			\end{subfigure}
			\caption{Limitations of peak memory usage as a memory-hardness measure}
			\label{fig:memory}
		\end{figure}
	\end{center}

	\paragraph{Cumulative complexity}
	\cite{AS15} put forward the notion of \emph{cumulative complexity} (CC), a complexity measure on graphs. CC was adopted by several subsequent works as a canonical measure of memory-hardness. CC measures the \emph{cumulative} memory usage of a graph pebbling function evaluation: that is, the sum of memory usage over all time-steps of computation. In other words, this is the area under a graph of memory usage against time. 
	CC is designed to be very robust against amortization, and in particular, scales linearly when computing many copies of a function on different inputs. This is a great advantage compared to the simpler measure of \cite{Per09}, which does not account well for an amortizing adversary (as shown in Figure \ref{fig:memory}).

	\paragraph{Depth-robust graphs}
	More recently, \cite{AB16,ABP17} proved bounds on optimal CC of certain graph families. They showed that a particular graph property called \emph{depth-robustness} suffices to attain optimal CC (the CC of any graph with bounded in-degree is upper bounded by $O\left(\frac{n^2 \log{\log{n}}}{\log{n}}\right)$~\cite{AB16,ABP17sustained}). 
	An $(r,s)$-depth-robust graph is one where there exists a path of length $s$ even when any $r$ vertices are removed. Intuitively, this captures the notion that storing any $r$ vertices of the graph will not shortcut the pebbling in a significant way.
	It turns out that depth-robustness will again be a useful property in our new model of memory-hardness with preprocessing.

	\paragraph{Sustained memory complexity}
	Very recently, Alwen, Blocki, and Pietrzak~\cite{ABP17} proposed a new measure of memory complexity, which captures not only the cumulative memory usage over time (as does CC), but goes further and captures the amount of time for which a particular level of memory usage is sustained. Our SHF definition also captures \emph{sustained} memory usage: we propose a definition of capturing the duration for which a given amount of memory is required, designed to capture static as well as dynamic memory requirements. By the nature of static memory, it is especially appropriate in our setting to consider (and maximize) the amount of time for which a static memory requirement is \emph{sustained}. 

	\paragraph{Core-area memory ratio}
	Previous works have considered certain hardware-dependent non-linearities in the ratio between the cost of memory and computation~\cite{BK15,AB16,RD17}. Such phenomena may incur a multiplicative factor increase in the memory cost that is dependent, in a possibly non-linear way, on specific hardware features. Note that the non-linearity here is in the hardware-dependence, rather than the space-time tradeoff itself.
	In contrast, our new models are more expressive, in that they
	make configurable the asymptotic tradeoff between space and time (by a parameter $\alpha$ which is in the exponent, as detailed in Definition~\ref{def:cca}) in an application-dependent way. This versatility of configuration targets applications where the trade-off may realistically depend on arbitrary and possibly exogenous space/time costs, and thus contrasts with metrics tailored for a specific hardware feature, such as core-memory ratio.

	\paragraph{Towards a general theory of moderately hard functions}
	Most recently, Alwen and Tackmann \cite{AT17} proposed a more general (though not comprehensive) framework for defining desirable guarantees of ``moderately hard functions,'' i.e., functions that are efficient to compute but somewhat hard to invert. Their work points out a number of drawbacks of prior measures such as those described above. Notably, many of the prior measures characterized the hardness of \emph{computing} the function with an implicit assumption that this hardness would translate to the hardness of \emph{inverting} the function (as it would indeed in the case of a brute-force approach to inversion). In other words, these measures implicitly assume that the hash function in question ``behaves like a random oracle'' in the sense that brute-force inversion is the optimal approach.
\fi

\section{Pebbling definitions}\label{sec:pebbling-defs}

\ificalpfullversion
	A \emph{pebbling game} is a one-player game played on a DAG where the goal of the player is to place pebbles on a set of one or more \emph{target nodes} in the DAG. 
\fi
\ificalpshortversion
	We introduce the \emph{black-magic pebble game} below
	(assuming familiarity with the standard pebble game--details of the standard black pebble game is given in the  attached full version). 
\fi
\ificalpfullversion
	
	In Section \ref{sec:standard-defs}, we formally define two variations of the sequential and parallel pebble games: the \emph{standard (black) pebble game} and the \emph{black-magic pebble game}, the latter of which we introduce in this work. 
	We also give the definitions of valid strategies and moves in these games.
	Then in Section \ref{sec:cost-of-peb}, we define measures for evaluating the sequential and parallel pebbling complexity on families of graphs. 
\fi

\ificalpfullversion
	\subsection{Standard and magic pebbling definitions}
	\label{sec:standard-defs}
\fi
\label{sec:standard-defs}
\ificalpfullversion
	\begin{definition}[Standard (black) pebble game]\label{def:standard-peb-game}

	~

	\begin{itemize}
		\item \textbf{Input:} A DAG, $G = (V, E)$, and a \emph{target set} $T \subseteq V$. Define $\pred{(v)} = \left\{u \in V: (u, v) \in E \right\}$, and let $S \subseteq V$ be the set of sources of $G$.
		\item \textbf{Rules at move $i$:} At the start of the game, no node of $G$ contains a pebble. The player has access to a supply of \emph{black pebbles}. Game-play proceeds in discrete moves, and $P_i$ (called a ``pebble configuration'') is defined as the set of nodes containing pebbles after the $i$th move. $P_0=\emptyset$ represents the initial configuration where no pebbles have been placed.
	Each move may consist of multiple actions adhering to the following rules.\footnote{Multiple applications of rules \ref{itm:peb-source}, \ref{itm:peb-remove}, and \ref{itm:peb-place} can occur in a single move. E.g., multiple sources can be pebbled in a single move. Rule \ref{itm:peb-slide} can also be applied multiple times in a single move \emph{for different pebbles}, but cannot be applied more than once to the same pebble (since, naturally, a single pebble cannot move to multiple locations).}
		\begin{enumerate}
			\item\label{itm:peb-source} 
				A pebble can be placed on any source, $s \in S$.
			\item\label{itm:peb-remove} 
				A pebble can be removed from any vertex.
			\item\label{itm:peb-place}
				A pebble can be placed on a non-source vertex, $v$, if and only if its direct predecessors were pebbled at time $i-1$ (i.e., $\pred{(v)} \in P_{i-1}$).
			\item\label{itm:peb-slide}
				A pebble can be moved from vertex $v$ to vertex $w$ if and only if $(v, w) \in E$ and $\pred{(w)} \in P_{i - 1}$.
		\end{enumerate}
		\item \textbf{Goal:} Pebble all nodes in $T$ at least once (i.e., $T \subseteq \bigcup_{i=0}^t P_i$).\footnote{This goal statement corresponds to the notion of a \emph{visiting pebbling} as defined in \cite{Nor15}. Our paper will use this \emph{visiting pebbling} notion throughout; however, we remark that an alternative notion of pebbling exists in the literature, called \emph{persistent pebbling}, which requires that all the nodes in $T$ be pebbled in the final configuration (i.e., $T\subseteq P_t$).}
	\end{itemize}
	\end{definition}
\fi

\ificalpfullversion
\begin{remark}
	At first glance, it may seem that rule \ref{itm:peb-slide} in Definition \ref{def:standard-peb-game} is redundant as a similar effect can be achieved by a combination of the other rules. However, the application of rule \ref{itm:peb-slide} can allow the usage of fewer pebbles. For example, a simple two-layer binary tree (with three nodes) could be pebbled with two pebbles using rule \ref{itm:peb-slide}, but would require three pebbles otherwise. Nordstrom \cite{Nor15} showed that in sequential strategies, it is always possible to use one fewer pebble by using rule \ref{itm:peb-slide}. 
	\iffullversion
		In fact, in parallel strategies, the difference made by rule \ref{itm:peb-slide} can be up to a factor of two, as we show in Appendix \ref{sec:pebble-sliding-rule}.
	\fi

	We note for completeness that while rule \ref{itm:peb-slide} is standard in the pebbling literature, not all the papers in the MHF literature include rule \ref{itm:peb-slide}.
	\end{remark}
\fi

\ificalpfullversion
	Next, we define the \emph{black-magic pebble game} which we will use to prove security properties of our static-memory-hard functions.
\fi

\begin{definition}[Black-magic pebble game]\label{def:blackmagic}

	~

	\begin{itemize}
		\item \textbf{Input:} A DAG $G = (V, E)$, a target set $T\subseteq V$, and \emph{magic pebble bound} $\mbound\in\NN\cup\{\infty\}$.
		\item \textbf{Rules:} At the start of the game, no node of $G$ contains a pebble. The player has access to two types of pebbles: \emph{black pebbles} and up to $\mbound$ \emph{magic pebbles}. Game-play proceeds in discrete moves, and $P_i=(M_i,B_i)$ is the pebble configuration after the $i$th move, where $M_i,B_i$ are the sets of nodes containing magic and black pebbles after the $i$th move, respectively. $P_0=(\emptyset,\emptyset)$ represents the initial configuration where no black pebbles or magic pebbles have been placed. 
		Each move may consist of multiple actions adhering to the following rules.
			\begin{enumerate}
				\item Black pebbles can be placed and removed according to the rules of the standard pebble game which are defined in the full version.\footnote{The rules of the standard pebble game are a standard definition in the pebbling literature. In the black-magic game, a predecessor node counts as ``pebbled'' if it contains either a black or a magic pebble. Where Definition~\ref{def:standard-peb-game} treats $P_i$ as a set of nodes, Definition~\ref{def:blackmagic} treats $P_i$ as equal to $M_i\cup B_i$.}
				\item A magic pebble can be placed on and removed from any node, subject to the constraint that at most $\mbound$ magic pebbles are used throughout the game.
				\item Each magic pebble can be placed at most once: after a magic pebble is removed from a node, it disappears and can never be used again.
			\end{enumerate}
		\item \textbf{Goal:} Pebble all nodes in $T$ at least once (i.e., $T \subseteq \bigcup_{i=0}^t \left(M_i \cup B_i\right)$).
	\end{itemize}
\end{definition}

\begin{remark}
	In the black-magic pebble game, unlike in the standard pebble game, there is always the simple strategy of placing magic pebbles directly on all the target nodes. At first glance, this may seem to trivialize the black-magic game. When optimizing for space usage, however, this simple strategy may not be favorable for the player: by employing a different strategy, the player might be able to use much fewer than $T$ pebbles overall.
\end{remark}

\ificalpfullversion
	Next, we define valid sequential and parallel strategies in these games.
\fi

\ificalpfullversion
	\begin{definition}[Pebbling strategy]\label{def:peb-strategy}
	Let $G$ be a graph and $T$ be a target set.
	A \emph{standard (resp., black-magic) pebbling strategy for $(G,T)$} is defined as a sequence of pebble configurations, $\s = \left\{P_0, \dots, P_t\right\}$, satisfying conditions \ref{itm:start-empty} and \ref{itm:legal-pebbling} below. $\s$ is moreover \emph{valid} if it satisfies condition \ref{itm:valid-pebbling}, and \emph{sequential} if it satisfies condition \ref{itm:sequential-pebbling}.
	\begin{enumerate}
	\item\label{itm:start-empty} $P_0=\emptyset$. 
	\item\label{itm:legal-pebbling} For each $i\in[t]$, $P_i$ can be obtained from $P_{i-1}$ by a legal move in the standard (resp., black-magic) pebble game.
	\item\label{itm:valid-pebbling} $\s$ successfully pebbles all targets, i.e., $T \subseteq \bigcup\limits_{i=0}^{t} P_i$. 
	\item\label{itm:sequential-pebbling} For each $i\in[t]$,
	$P_i$ contains at most one vertex not contained in $P_{i-1}$ 
	(i.e., $|P_i\setminus P_{i-1}|\leq 1$).
	\end{enumerate}
	A black-magic pebbling strategy must satisfy one additional condition to be considered \emph{valid}:
	\begin{enumerate}
		\setcounter{enumi}{4}
		\item At most $\mbound$ magic pebbles are used throughout the strategy, i.e., $|\bigcup_{i\in[t]} M_i|\leq\mbound$ where $M_i$ is the $i$th configuration of magic pebbles.
	\end{enumerate}
	\end{definition}
\fi
\ificalpshortversion
	A \emph{pebbling strategy} $\cP=(P_1,\dots,P_t)$ is a sequence of configurations, describing a series of moves on a DAG. A strategy is \emph{valid} if it follows the rules of a particular pebbling game.
\fi

\ificalpfullversion
	\subsection{Cost of pebbling}
	\label{sec:cost-of-peb}

	In this subsection, we give definitions of several cost measures of graph pebbling, applicable to the standard and black-magic pebbling games.
	While these definitions assume parallel strategies, we note that the sequential versions of the definitions are entirely analogous.

	\subsubsection{Space complexity in standard pebbling}

	We give a brief informal summary of the definitions in this subsection, before proceeding to the formal definitions.
\fi

\subh{Pebbling complexity measures}
We informally overview the pebbling complexity definitions, some of which are new to this work.

\ificalpshortversion
We refer to the full version for formal definitions.
\fi

The \emph{time complexity} of a pebbling strategy $\s$ is the number of steps, i.e., $\ti\left(\s\right) = |\s|$. The \emph{time complexity} of a graph $G = (V, E)$ given that at most $S$ pebbles can be used is $\ti(G, S) = \min_{\s \in \mathbb{P}_{G, T, S}} \left(\ti\left(\s\right)\right)$. Next, we overview variants of space complexity.

\begin{enumerate}
\item \textbf{Space complexity} of a \emph{pebbling strategy} $\s$ on a graph $G$, denoted by $\scost(\s)$, is the minimum number of pebbles required to execute $\s$. Space complexity of the \emph{graph} $G$ with target set $T$, written $\scost(G,T)$, is the minimum space complexity of any valid pebbling strategy for $G$.
\item \textbf{$\Lambda$-sustained space complexity} \cite{ABP17}\footnote{We note that our notation diverges from that of \cite{ABP17}, but our Definition \ref{def:lambda-sattime} is equivalent to their definition of ``$s$-sustained space complexity.'' (E.g., they write $\Pi_{ss}(\s,\Lambda)$ instead of $\sspace(G,\s,\Lambda)$.) \ificalpfullversion We gave this decision some consideration as inconsistent notation can add confusing clutter to a literature; we decided on our notation (1) in order to keep consistency with the pebbling literature, where the pyramid graphs that will be used in our SHF construction are traditionally denoted by $\Pi$; and (2) because our notation makes the graph $G$ explicit where sometimes it is implicit in \cite{ABP17}, and this is important for the new ``graph-optimal sustained complexity'' notion we introduce.\fi} of a \emph{pebbling strategy} $\s$ on a graph $G$, denoted by $\sspace(\s,\Lambda)$, is the number of time-steps during the execution of $\s$, in which at least $\Lambda$ pebbles are used. $\Lambda$-sustained space complexity of the \emph{graph} $G$ with target set $T$, written $\sspace(G,\Lambda,T)$ is the minimum $\Lambda$-sustained space complexity of all valid pebbling strategies for $G$.
\item \textbf{Graph-optimal sustained complexity} of a \emph{pebbling strategy} $\s$, denoted by $\stime(\s)$, is the number of time-steps during the execution of $\s$, in which the number of pebbles in use is equal to the space complexity of $G$. Graph-optimal sustained complexity of the \emph{graph} $G$ with target set $T$, written $\stime(G,T)$ is the minimum graph-optimal sustained complexity of all valid pebbling strategies for $G$.
\item \textbf{$\Delta$-suboptimal sustained complexity} of a \emph{pebbling strategy} $\s$ is the number of time-steps, during the execution of $\s$, in which the number of pebbles in use is at least the space complexity of $G$ \emph{minus $\Delta$}. $\Delta$-suboptimal sustained complexity of the \emph{graph} $G$ is the minimum $\Delta$-suboptimal sustained complexity of all valid pebbling strategies for $G$.
\end{enumerate}

\ificalpshortversion
	The third and fourth definitions are new to this paper. They can be seen as special variants of $\Lambda$-sustained space complexity, i.e., with a special setting of $\Lambda$ dependent on a graph family. They are useful to define in their own right, as they express complexity relative to the best possible value of $\Lambda$ at which sustained space usage could be hoped for, for a \emph{given graph family}. 
	Hereafter, we prove guarantees on \emph{graph-optimal sustained complexity} of our constructions. We also define \emph{$\Delta$-suboptimal sustained complexity} above for completeness, since it is more general \footnote{More specifically, graph-optimal sustained complexity is $\Delta$-suboptimal sustained complexity for $\Delta=0$.} and preferable to graph-optimal complexity when evaluating graph families where the maximal space usage may not be sustained for very long.
\fi
\ificalpfullversion
	A couple of remarks are in order.

	\begin{remark}
	The third and fourth definitions are new to this paper. They can be seen as special variants of $\Lambda$-sustained space complexity, i.e., with a special setting of $\Lambda$ dependent on the specific graph family in question. They are useful to define in their own right, as unlike plain $\Lambda$-sustained space complexity, these measures express complexity for a given graph family relative to the best possible value of $\Lambda$ at which sustained space usage could be hoped for. In the rest of this paper, we prove guarantees on \emph{graph-optimal sustained complexity} of our constructions, which have high sustained space usage at the optimal $\Lambda$-value. However, we also define \emph{$\Delta$-suboptimal sustained complexity} here for completeness, since it is more general\footnote{More specifically, graph-optimal sustained complexity is $\Delta$-suboptimal sustained complexity for $\Delta=0$.} and preferable to graph-optimal complexity when evaluating graph families where the maximal space usage may not be sustained for very long.
	\end{remark}

	\begin{remark}
	We have found the term ``$\Lambda$-sustained space complexity'' can be slightly confusing, in that it measures a number of time-steps rather than an amount of space. We retain the original terminology as it was introduced, but include this remark to clarify this point.
	\end{remark}
\fi

\ificalpfullversion
	We now present the formal definitions of the complexity measures for the standard pebbling game. In all of the below definitions, $G = (V, E)$ is a graph, $T\subseteq V$ is a target set, $\cP=(P_1,\dots,P_t)$ is a standard pebbling strategy on $(G,T)$, and $\mathbb{P}_{G,T}$ denotes the set of all valid standard pebbling strategies on $(G,T)$.

	\begin{definition}
	\label{def:satcost}\label{def:min-satcost}
	The \emph{space complexity of pebbling strategy $\s$} is: 
	$\scost(\s) \defeq \max_{P_i \in \s} \left(\left|P_i\right|\right)$. 
	The \emph{space complexity of $G$}
	is the minimal space complexity of any valid pebbling strategy that pebbles the target set $T \subset V$:
	$\scost(\g, T) \defeq \min_{\s' \in \mathbb{P}_{G,T}}\left(\scost\left(\s'\right)\right)$.
	\end{definition}

	\begin{definition}
	\label{def:lambda-sattime}
	The \emph{$\Lambda$-sustained space complexity of $\s$} is: 
	$\sspace(\s, \Lambda) \defeq \left|\left\{P_i: |P_i| \geq \Lambda\right\}\right|$.
	The \emph{$\Lambda$-sustained space complexity of $G$} 
	is the minimal $\Lambda$-sustained space complexity of any valid pebbling strategy that pebbles the target set $T \subseteq V$:
	$\sspace(\sg{2},\Lambda, T) \defeq \min_{\s' \in \mathbb{P}_{G,T}}\left(\sspace\left(\s', \Lambda\right)\right)$.
	\end{definition}

	\begin{definition}
	\label{def:sattime}
	The \emph{graph-optimal sustained complexity of $\s$} is: 

	$\stime(\s)=\sspace(\s,\scost(G, T))$.
	The \emph{graph-optimal sustained complexity of $G$}
	is the minimal graph-optimal sustained complexity of any valid pebbling strategy that pebbles the target set $T \subseteq V$:
	$\stime(\sg{2}, T) \defeq \min_{\s' \in \mathbb{P}_{G,T}}\left(\stime\left(\s'\right)\right)$.
	\end{definition}

	\begin{definition}
	The \emph{$\Delta$-suboptimal sustained complexity of $\s$} is: 
	$$\stime(\s,\Delta)=\sspace(\s,\scost(G, T)-\Delta).$$
	The \emph{$\Delta$-suboptimal sustained complexity of $G$}
	is the minimal graph-optimal sustained complexity of any valid pebbling strategy that pebbles the target set $T \subseteq V$:
	$\stime(\sg{2},\Delta, T) \defeq \min_{\s' \in \mathbb{P}_{G,T}}\left(\stime\left(\s',\Delta\right)\right)$.
	\end{definition}

	\subsubsection{Time complexity in standard pebbling}

	We present the following formal definitions for measuring the time complexity of strategies in the standard pebble game.  In all the below definitions, $G = (V, E)$ is a graph, $T \subseteq V$ is a target set, $\s = (P_1, \dots, P_t)$ is a standard pebbling strategy on $(G, T)$ where $\mathbb{P}_{G, T, S}$ denotes the set of all valid pebbling strategies on $(\sg{2}, T)$ that use at most $S$ pebbles. 

	\begin{definition}\label{def:strategy-time}
	The \emph{time complexity} of a pebbling strategy $\s$ is $\ti\left(\s\right) = |\s|$. The \emph{time complexity} of a graph $G = (V, E)$ given that at most $S$ pebbles can be used is $\ti(G, S) = \min_{\s \in \mathbb{P}_{G, T, S}} \left(\ti\left(\s\right)\right)$.
	\end{definition}

	\subsubsection{Space complexity in black-magic pebbling}

	Next, we define the corresponding complexity notions for the black-magic pebbling game. As above, $G = (V, E)$ is a graph, $T\subseteq V$ is a target set, and $\mbound$ is a magic pebble bound. In this subsection, $\cP=(P_1,\dots,P_t)=((M_1,B_1),\dots,(M_t,B_t))$ denotes a black-magic pebbling strategy on $(G,T)$. Moreover, $\mathbb{M}_{G,T,\mbound}$ denotes the set of all valid magic pebbling strategies on $(G,T)$, and $m(\s)$ denotes the total number of magic pebbles used in the execution of $\s$.

	\newcommand{\magicstrategies}{\mathbb{P}_{G,T,\mbound}}

	\begin{definition}\label{def:magic-satcost}
	The \emph{(magic) space complexity} of $\cP$ is:
	$\scostm(\s) \defeq \max\left(m({\s}), \max_{P_i \in \s} \left(\left|P_i\right|\right)\right)$. 
	The \emph{(magic) space complexity of $G$ w.r.t. $\mbound$}
	is the minimal space complexity of any valid magic pebbling strategy that pebbles the target set $T \subseteq V$:
	$\scostm(\g,\mbound, T) \defeq \min_{\s \in \magicstrategies}\left(\scostm\left(\s\right)\right)$.
	\end{definition}

	\begin{remark}
	We briefly provide some intuition for the complexity measure defined above in Def.~\ref{def:magic-satcost}. If we consider all magic pebbles to be static memory objects that were saved from a previous evaluation of the hash function, then the total number of magic pebbles is the amount of memory that was used to save the results of a previous evaluation of the hash function. Because of this, it is natural to take the maximum of the memory used to store results from a previous evaluation of the function and the current memory that is used by our current pebbling strategy since that would represent how much memory was used to compute the results of hash function during the current evaluation. 
	\end{remark}

	\begin{definition}\label{def:magic-lambda-sattime}
	The \emph{(magic) $\Lambda$-sustained space complexity} of $\cP$ is:
	$\sspacem(\s,\Lambda) \defeq \left|\left\{P_i: \left|P_i\right| \geq \Lambda\right\}\right|$.
	The \emph{$\Lambda$-sustained space complexity of $G$ w.r.t. $\mbound$ and $T \subseteq V$} is:
	$\sspacem(\sg{2},\Lambda,\mbound, T) \defeq \min_{\s \in \magicstrategies}\left(\stimem\left(\s,\Lambda\right)\right)$.
	\end{definition}

	\begin{definition}\label{def:magic-sattime}
	The \emph{(magic) graph-optimal sustained complexity} of $\cP$ is:
	$\stimem(\s)=\sspacem(\s,\scostm(G, T))$.
	The \emph{graph-optimal sustained complexity of $G$ w.r.t. $\mbound$ and $T \subseteq V$} is:
	$\stimem(\sg{2},\mbound, T) \defeq \min_{\s \in \magicstrategies}\left(\stimem\left(\s\right)\right)$.
	\end{definition}

	\begin{definition}\label{def:magic-suboptimal}
	The \emph{(magic) $\Delta$-suboptimal sustained complexity} of $\cP$ is:
	$\stimem(\s,\Delta)=\sspacem(\s,\scostm(G, T)-\Delta)$.
	The \emph{$\Delta$-suboptimal sustained complexity of $G$ w.r.t. $\mbound$ and $T \subseteq V$} is:
	$$\stimem(\sg{2},\Delta,\mbound, T) \defeq \min_{\s \in \magicstrategies}\left(\stimem\left(\s,\Delta\right)\right).$$
	\end{definition}
\fi

\ificalpfullversion
	\subsection{Incrementally hard graphs}

\fi
\ificalpshortversion
	\subh{Incrementally hard graphs}
\fi
We introduce the following definition for our notion of graphs which require $|T|$ pebbles to pebble regardless of the number of targets that are asked, given a constraint on the number of magic pebbles that can be used. This concept has not been previously analyzed in the pebbling literature; traditional pebbling complexity usually treats graphs with fixed target sets.

\begin{definition}[Incremental Hardness]\label{def:incrementally-hard}
Given at most $\mbound$ magic pebbles, for any subset of targets $C \subseteq T$ where $|C| > \mbound$, the number of pebbles (magic and black pebbles) necessary in the black-magic pebble game to pebble $C$ is at least $|T|$ where the number of magic pebbles used in this game is upper bounded by $\mbound$: $\scostm(\sg{2}, |C| - 1, C) \geq |T|$.
\end{definition}

\subsubsection{$\alpha$-tradeoff cumulative complexity}

\emph{$\alpha$-tradeoff cumulative complexity}, or $\cca$, is a new measure introduced in this paper, which accounts for situations where space and time do not trade off linearly. 
\ificalpfullversion
	Here, we see the usefulness of defining sustained complexities in terms of the minimum required space (as opposed to being parametrized by $\Lambda$) since we can always obtain an upper bound on $\cca$, for \emph{any} $\alpha$, of a graph directly from our proofs of the space complexity and sustained time complexity of a DAG.
\fi

\begin{definition}[Standard pebbling $\alpha$-space cumulative complexity]\label{def:pcost}
Given a valid parallel standard pebbling strategy, $\s$, for pebbling a graph $\sg{2} = (V, E)$, the \emph{standard pebbling $\alpha$-space cumulative complexity} is the following:
\ificalpfullversion
	\begin{align*}
	\pcost(\sg{2}, \s) \defeq \sum_{P_i \in \s} |P_i|^{\alpha}\ .
	\end{align*}
\fi
\ificalpshortversion
	$\pcost(\sg{2}, \s) \defeq \sum_{P_i \in \s} |P_i|^{\alpha}$.
\fi
\end{definition}

\begin{definition}[Black-magic pebbling $\alpha$-space cumulative complexity]\label{def:pcost}
Given a valid parallel black-magic pebbling strategy, $\s$, for pebbling a graph $\sg{2} = (V, E)$, the \emph{black-magic pebbling $\alpha$-space cumulative complexity} is the following:
\begin{align*}
\pcostm(\sg{2}, \s) \defeq \max\left(m(\s)^{\alpha}, \sum_{P_i \in \s} |P_i|^{\alpha}\right) = \max\left(m(\s)^{\alpha}, \sum_{P_i \in \s}|B_i \cup M_i|^{\alpha}\right)\ .
\end{align*}
\end{definition}

\ificalpfullversion
	The following definition, $\cca$, is an analogous definition to $CC$ as defined by~\cite{AS15} (specifically, $\cca$ when $\alpha = 1$ is equivalent to $\cc$) to account for varying costs of memory usage vs. time. 
\fi

\begin{definition}[$\cca$]\label{def:cca}
Given a graph, $\g \in \gf$, and a valid standard/magic pebbling strategy, $\s$, we define the $\cca(\g)$ to be

\begin{align*}
	\cca(\s) \defeq \left(\pcost\left(\sg{2}, \s\right)\right).
\end{align*}

Given a graph, $\g \in \gf$, and a family of valid standard pebbling strategies, $\mathbb{P}$, we define the $\cca(\g)$ to be
\ificalpfullversion
	\begin{align*}
	    \cca(\g) \defeq \min_{\s \in \mathbb{P}}\left(\pcost\left(\g, \s \right)\right)\ ,
	\end{align*}
\fi
\ificalpshortversion
	$\cca(\g) \defeq \min_{\s \in \mathbb{P}}\left(\pcost\left(\g, \s \right)\right)$.
\fi
and, given a family $\mathbb{P}^M$ of valid black-magic pebbling strategies, we define $\cca(\g)$ to be
\ificalpfullversion
	\begin{align*}
	\cca(\g) \defeq \min_{\s^M \in \mathbb{P}^M}\left(\pcostm\left(\g, \s^M\right)\right)\ .
	\end{align*}
\fi
\ificalpshortversion
	$\cca(\g) \defeq \min_{\s^M \in \mathbb{P}^M}\left(\pcostm\left(\g, \s^M\right)\right)$.
\fi
\end{definition}

\section{Parallel random oracle model (PROM)}\label{sec:prom}

\ificalpfullversion
    In this paper, we consider two broad categories of computations: \emph{pebbling strategies} and \emph{PROM algorithms}. Specifically, we discussed above the pebbling models and pebble games we use to construct our static memory-hard functions. Now, we define our PROM algorithms. 
\fi

\ificalpfullversion
    Prior work has observed the close connections between these two types of computations as applied to DAGs, and our work brings out yet more connections between the two models. In this section, we give an overview of how PROM computations work and define the complexity measures that we apply to PROM algorithms. Some of the complexity measures were introduced by prior work, and others are new in this work.
\fi

\ificalpfullversion
    \subsection{Overview of PROM computation}
\fi
\label{sec:prom-overview}
\ificalpfullversion
    The random oracle model was introduced by \cite{BR93}.
    When we say random oracle, we always mean a \emph{parallel} random oracle unless otherwise specified.

\fi
An \emph{algorithm} in the PROM is a probabilistic algorithm $\cB$ which has parallel access to a stateless oracle $\RO$: that is, $\cB$ may submit many queries in parallel to $\RO$.
We assume $\RO$ is sampled uniformly from an oracle set $\OSet$ and that $\cB$ may depend on $\OSet$ but not $\RO$.

The algorithm proceeds in discrete time-steps called \emph{iterations}, and may be thought to consist of a series of algorithms $(\cB_i)_{i\in\NN}$,
indexed by the iteration $i$, where each $\cB_i$ passes a \emph{state} $\state_i\in\zo^*$ to its successor $\cB_{i+1}$.
$\state_0$ is defined to contain the input to the algorithm.
We write $|\state_i|$ to denote the size, in bits, of $\state_i$. We write $\wsize{\state_i}$ to denote $\frac{|\state_i|}{w}$, where $w$ is the output length of the oracle $\RO$. In other words, $\wsize{\state_i}$ is the size of $\state_i$ when counting in words of size $w$. 
In each iteration, the algorithm $\cB_i$ may make a \emph{batch} $\bq_i=(q_{i,1},\dots,q_{i,|\bq_i|})$ of queries, consisting of $|\bq_i|$ individual queries to $\RO$, and instantly receive back from the oracle the evaluations of $\RO$ on the individual queries, i.e.,
$(\RO(q_{i,1}),\dots,\RO(q_{i,|\bq_i|}))$.

At the end of any iteration, $\cB$ can append values to a special output register, and it can end the computation by appending a special terminate symbol $\bot$ on that register.
When this happens, the contents $y$ of the output register, excluding the trailing $\bot$, is considered to be the output of the computation.
To denote the process of sampling an output, $y$, provided input $x$, we write $y\gets\cB^{\RO}(x)$.

\begin{definition}[Oracle functions]
An \emph{oracle function} is a collection $\ff=\{f^\RO:D\to R\}_{\RO\in\OSet}$ of functions with domain $D$ and outputs in $R$ indexed by oracles $\RO\in\OSet$.

A \emph{family of oracle functions} is a set $\cF=\{\ff_\sec:D_\sec\to R_\sec\}_{\sec\in\NN}$
where each $\ff_\sec$ is indexed by oracles from an oracle set $\OSet_\sec:\zo^\sec\to\zo^\sec$ indexed by a security parameter $\kappa$.\footnote{For simplicity, we have the input and output domains of the oracles equal to $\zo^\sec$, but this is not a necessary restriction: the sizes could be any polynomials in $\sec$.}
\end{definition}

\begin{definition}[Memory complexity of PROM algorithms]
The \emph{memory complexity of $\cB(x;\rho)$} 
\ificalpfullversion
    (i.e., the memory complexity of $\cB$ on input $x$ and randomness $\rho$) 
\fi
is defined as:
\ificalpfullversion
    \begin{equation}\label{eqn:satscore-def}
        \satscore_\OSet(\cB,x,\rho)\defeq\max_{i\in\NN}\left\{\wsize{\sigma_i}\right\}\ .
    \end{equation}
\fi
\ificalpshortversion
    $\satscore_\OSet(\cB,x,\rho)\defeq\max_{i\in\NN}\left\{\wsize{\sigma_i}\right\}$.
\fi
\end{definition}

\begin{definition}[$\Lambda$-sustained memory complexity of PROM algorithms]
The \emph{$\Lambda$-sustained memory complexity of $\cB(x;\rho)$} is defined as: 
\ificalpfullversion
    \begin{equation}\label{eqn:lambda-sus-def}
        \sattime_\OSet(\Lambda,\cB,x,\rho)\defeq\left|\left\{i\in\NN: |\sigma_i|\geq\Lambda\right\}\right|\ .
    \end{equation}
\fi
\ificalpshortversion
    $\sattime_\OSet(\Lambda,\cB,x,\rho)\defeq\left|\left\{i\in\NN: |\sigma_i|\geq\Lambda\right\}\right|$.
\fi
\end{definition}

\ificalpfullversion
    Note that \eqref{eqn:satscore-def} and \eqref{eqn:lambda-sus-def} are distributions over the choice of $\RO\gets\OSet$.
\fi

\ificalpfullversion
	\subsection{Functions defined by DAGs}

\fi
\ificalpshortversion
	\subh{Functions defined by DAGs}
\fi
\ificalpfullversion
	We now describe how to translate a graph construction into a function family, whose evaluation involves a series of oracle calls in the PROM. 
\fi
Any family of DAGs induces a family of \emph{oracle functions} in the PROM, whose complexity is related to the pebbling complexity of the DAG.
We first define the syntax of \emph{labeling} of DAG nodes, then define a \emph{graph function family}.

\begin{definition}[Labeling]\label{def:labeling}
Let $G=(V,E)$ be a DAG with maximum in-degree $\deg$, let
$\LabelSet$ be an arbitrary ``label set,'' and define
$\OSet(\deg,\LabelSet)=\left(V\times \bigcup_{\deg'=1}^{\deg}\LabelSet^{\deg'}\to \LabelSet\right)$.
For any function $\RO\in\OSet(\deg,\LabelSet)$ and any label $\zeta\in\LabelSet$,
the $(\RO,\zeta)$-labeling of $G$ is a mapping 
$\lab_{\RO,\zeta}:V\to\LabelSet$ defined recursively as 
\ificalpshortversion
	described in the full version of the paper.
\fi
\ificalpfullversion
	follows.\footnote{%
	We abuse notation slightly and also invoke $\lab_{\RO,\zeta}$ on \emph{sets} of vertices, in which case the output is defined to be a tuple containing the labels of all the input vertices, arranged in lexicographic order of vertices.}
	$$\lab_{\RO,\zeta}(v)=\begin{cases}
	\RO(v,\zeta) & \mbox{ if } \indeg(v)=0 \\
	\RO(v,\lab_{\RO,\zeta}(\pred(v))) & \mbox{ if } \indeg(v)>0 
	\end{cases}\ .$$
\fi
\end{definition}

\begin{definition}[Graph function family]\label{def:gff}
Let $n=n(\sec)$ and let $\gfsec=\{\glong\}_{\sec\in\NN}$ be a graph family. We write $\OSet_{\delta,\sec}$ to denote the set $\OSet(\deg,\zo^\sec)$ as defined in Definition \ref{def:labeling}.
The \emph{graph function family} of $\gf$ is the family of oracle functions 
$\cF_{\gf}=\{\ff_{\g}\}_{\sec\in\NN}$ where
$\ff_{\g}=\{f^{\RO}_{\g}:\zo^\sec\to(\zo^\sec)^z\}_{\RO\in\OSet_{\delta,\sec}}$ and $z=z(\sec)$ is the number of sink nodes in $\g$. The output of $f^{\RO}_{\g}$ on input label $\zeta\in\zo^\sec$ is defined to be
\ificalpfullversion
	$$f^{\RO}_{\g}(\zeta)\defeq\lab_{\RO,\zeta}(\sink(\g))\ ,$$
	where $\sink(\g)$ is the set of sink nodes of $\g$.
\fi
\ificalpshortversion
	$f^{\RO}_{\g}(\zeta)\defeq\lab_{\RO,\zeta}(\sink(\g))$,
	where $\sink(\g)$ is the set of sink nodes of $\g$.
\fi
\end{definition}

\ificalpfullversion
	\subsection{Relating complexity of PROM algorithms and pebbling strategies}

\fi
\ificalpshortversion
	\subh{Relating complexity of PROM algorithms and pebbling}
\fi
Any PROM algorithm $\cB$ and input $x$ induce a black-magic pebbling strategy, $\epfm_\zeta(\cB,\RO,x,\$)$, called an \emph{ex-post-facto black-magic pebbling strategy}. The way in which this strategy is induced is similar to \emph{ex-post-facto pebbling} as originally defined by \cite{AS15} in the context of the standard pebble game. We adapt their technique for the black-magic game. 
\ificalpfullversion
	\begin{definition}[Ex-post-facto black-magic pebbling]\label{def:expostfacto-magic}
	Let $n=n(\sec)$ and let $\gfsec=\{\glong\}_{\sec\in\NN}$ be a graph family. Let $\zeta=\zeta(\sec)\in\zo^\sec$ be an arbitrary input label for the graph function family $\cF_{\gf}$. For any $v\in V_n$, define \begin{equation*}
		\prelab_{\RO,\zeta}(v)\defeq(v,\lab_{\RO,\zeta}(\pred(v)))\ .
	\end{equation*}
	
	Let $\cB$ be a non-uniform PROM algorithm.
	Fix an implicit security parameter $\sec$.
	Let $x$ be an input to $\cB$.
	We now define a magic pebbling strategy induced by any given execution of $\cB^\RO(x;\$)$, where $\$$ denotes the random coins of $\cB$. Such an execution makes a sequence of batches of random oracle calls (as defined in Section \ref{sec:prom-overview}), which we denote by
	\begin{equation*}
		\bq(\cB,\RO,x,\$)\defeq(\bq_1,\dots,\bq_t)\ .
	\end{equation*}
	
	The induced black-magic pebbling strategy,
	\begin{equation}\label{eqn:epf_def-magic}
	\epfm_\zeta(\cB,\RO,x,\$)=((B_0,M_0),\dots,(B_t,M_t)) \ ,
	\end{equation}
	is called an \emph{ex-post-facto black-magic pebbling}, and is defined by the following procedure.
	
	\begin{enumerate}
	\item $B_0=M_0=\emptyset$.
	\item For $i=1,\dots,t$:
		\begin{enumerate}
		\item $B_i=B_{i-1}$. 
		\item $M_i=M_{i-1}$.
		\item For each individual query $q\in\bq_i$, if there is some $v\in V_n$ such that $q=\prelab_{\RO,\zeta}(v)$ and $v\notin P_i$, then ``pebble $v$'' by performing the following steps:
			   \begin{enumerate}
				   \item If $\pred(v)\subseteq M_i\cup B_i$:
					   \begin{itemize}
						   \item $B_i=B_i\cup\{v\}$.
					   \end{itemize} 
				   \item Else:
					   \begin{itemize}
						   \item $V=\{v\}$.
						   \item Let $V^*$ be the transitive closure of $V$ under the following operation:\\
						   $V=V\cup\left(\bigcup_{v'\in V}\pred(v')\cap(M_i\cup B_i)\right)$.
						   \item $M_i=M_i\cup V^*$.
					   \end{itemize}
			   \end{enumerate}
		\end{enumerate}
	\item For $i=1,\dots,t$:
		\begin{enumerate}
			\item A node $v\in M_i\cup B_i$ is said to be \emph{necessary at time $i$} if
				\begin{align*}
					\exists j\in[t], q\in\bq_j, v'\in V_n \mbox{ s.t. } \quad & j>i \wedge v\in\pred(v') \wedge q=\prelab_{\RO,\zeta}(v') \\ & \wedge \Big(\not\exists k\in[t], q'\in\bq_k \mbox{ s.t. } i<k<j \wedge q'=\prelab_{\RO,\zeta}(v)\Big) \ .
				\end{align*}
				In other words, a node is necessary if its label will be required in a future oracle call,
				but its label will not be obtained by any oracle query between now and that future oracle call.
		
				Remove from $B_i$ and $M_i$ all nodes that are not necessary at time $i$.
		\end{enumerate}
	\end{enumerate}
\end{definition}
\fi
\ificalpshortversion
	The formal definition 
	is given in the full version
	due to space constraints.
\fi

\ificalpfullversion
	\subsection{Legality and space usage of ex-post-facto black-magic pebbling}
\fi

The following theorems establish that the space usage of PROM algorithms is closely related to the space usage of the induced pebbling.

\ificalpfullversion
	We will use the following supporting lemma, also used in prior work such as \cite{AS15,DKW11} (see, e.g., \cite{DKW10} for a proof).

	\begin{lemma}\label{lemma1}
	Let $B=b_1,\dots,b_u$ be a sequence of random bits and let $\mathbb{H}$ be a set. Let $\cP$ be a randomized procedure that gets a hint $h\in\mathbb{H}$, and can adaptively query any of the bits of $B$ by submitting an index $i$ and receiving $b_i$ as a response. At the end of its execution, $\cP$ outputs a subset $S\subseteq\{1,\dots,u\}$ of $|S|=\phi$ indices which were not previously queried, along with guesses for the values of the bits $\{b_i:i\in S\}$. Then the probability (over the choice of $B$ and the randomness of $\cP$) that there exists some $h\in\mathbb{H}$ such that $\cP(h)$ outputs \emph{all} correct guesses is at most $|\mathbb{H}|/2^\phi$.
	\end{lemma}
\fi

\begin{lemma}[Legality and magic pebble usage of ex-post-facto black-magic pebbling]\label{lem:epfm-legal}
Let $n=n(\sec)$ and let $\gfsec=\{\glong\}_{\sec\in\NN}$ be a graph family.
Let $\zeta\in\zo^\sec$ be an arbitrary input label for $\gfsec$.
Fix any PROM algorithm $\cB$ and input $x$.
With overwhelming probability over 
\ificalpfullversion the choice of random oracle  \fi
$\RO\gets\OSet$ and the random coins $\$$ of $\cB$, it holds that
the ex-post-facto magic pebbling $\epfm_\zeta(\cB,\RO,x,\$)$ consists of valid magic-pebbling moves, 
and uses no more than $\chi=\floor{|x|/\sec}$ magic pebbles (i.e., for all $i$, $|M_i|\leq\chi$).
\end{lemma}
\begin{proof}
Fix an algorithm $\cB$ and, for the sake of contradiction, suppose that there is an input $x$ such that with non-negligible probability over $\RO$ and $\$$, the induced pebbling $\epfm_\zeta(\cB,\RO,x,\$)$ uses more than $\chi$ magic pebbles or contains an invalid move. By definition, this means that the following event $\cE$ occurs with non-negligible probability: on more than $\chi$ occasions, a (magic) pebble is placed on a node $v$ although its parents were not all pebbled in the previous step. In turn, this means that a correct random-oracle query for the label of $v$ is made by $\cB$; and the correct query contains the label of some predecessor node $v'$ which was not contained in the output of any previous oracle call.

Let us suppose that event $\cE$ occurs with probability more than $p=(q(\chi+1)2^{|x|})/2^{\sec(\chi+1)}$, where $q$ is the number of oracle queries made by $\cB$. Note that this probability is negligible, since $\sec(\chi+1)=\sec(\floor{|x|/\sec}+1)\geq |x|+\sec-1\gg|x|+\log(q(\chi+1))$. Based on this assumption, we construct a predictor that predicts $\chi+1$ output values of the random oracle with impossibly high probability (specifically, violating Lemma~\ref{lemma1}) as follows. The predictor $\cP$ depends on input $x$ and can query the random oracle on inputs of its choice, before outputting its prediction. Let $\hat{r}$ be an upper bound on the number of random bits used by $\cB(x)$. The predictor also has access to a sequence $\hat{R}$ of $\hat{r}$ random bits, that it can use to simulate the random coins of $\cB$.
\begin{itemize}
	\item \emph{Hint:} The predictor $\cP$ receives as its hint\footnote{Note that the hint may depend both on the choice of random oracle, and on the randomness $\hat{R}$.} \emph{either} $\bot$ if the induced pebbling $\epfm_\zeta(\cB,\RO,x,\$)$ is valid and uses no more than $\chi$ magic pebbles, \emph{or} the following information otherwise:
		\begin{itemize}
			\item the index $i^*\in[q]$ of the first oracle call causing the illegal event (inducing the $(\chi+1)$th placement of a magic pebble on some node $v$) to happen;
			\item the indices $I\subset[i^*]$ of all oracle calls preceding the $i^*$th oracle call, that induce the placement of a magic pebble or pebbles; and
			\item $\cB$'s input $x$.
		\end{itemize}
		The size of this hint is at most $(\chi+1)\log(q)+|x|$ bits.
	\item \emph{Execution:} If the hint is $\bot$, then $\cP$ halts and outputs nothing. Otherwise, $\cP$ runs $\cB(x;\hat{R})$, forwarding all oracle calls to the random oracle, until the $i^*$th query. By construction, for each $i'\in I\cup\{i^*\}$, the $i'$th query contains the labels of the parents of the node $v_{i'}$ whose pebbling is induced by the $i'$th query, and at least one of these labels (say, label $\ell_{w_{i'}}$ for parent node $w_{i'}$) was not the output of any previous query to the random oracle. For each $i'\in I\cup\{i\}$, our predictor recomputes the value $\tilde{w}_{i'}=\prelab_{\RO,\zeta}(w_{i'})$ which is the preimage under $\RO$ of $\ell_{w_{i'}}$. Note that by definition of $\prelab$, $\tilde{w}_{i'}$ can be computed without ever querying $\RO$ on input $\tilde{w}_{i'}$. Finally, $\cP$ outputs the following pairs:
	$$\left\{(\tilde{w}_{i'},\ell_{w_{i'}})\right\}_{i'\in I\cup\{i^*\}}\ .$$
	Since by construction, each query in $i'\in I\cup\{i^*\}$ induced the placement of a magic pebble, it follows that each pair $(\tilde{w}_{i'},\ell_{w_{i'}})$ is a valid input-output pair of $\RO$. Moreover, $\cP$ never queried $\RO$ on any $\tilde{w}_{i'}$.
\end{itemize}
By our assumption about the probability $p$ of the event $\cE$, the predictor's hint is $\bot$ with probability at most $\cE$, and the predictor succeeds whenever the hint is not $\bot$. That is, the predictor must succeed with probability greater than $p=(q(\chi+1)2^{|x|})/2^{\sec(\chi+1)}$. By construction, the size of the predictor's hint set is at most $q(\chi+1)2^{|x|}$, and the predictor's output is $\sec(\chi+1)$ bits long. Thus Lemma \ref{lemma1} implies that the probability (over the choice of $\RO$ and the randomness of $\cP$) that there is some hint such that $\cP$ outputs all correct guesses is at most $(q(\chi+1)2^{|x|})/2^{\sec(\chi+1)}$. (This is equal to $p$.) We have a contradiction, and the lemma follows.
\end{proof}

\begin{lemma}[Space usage of ex-post-facto black-magic pebbling]\label{lem:epfm-space}
Let $n,\gfsec,\zeta$ be as in Lemma~\ref{lem:epfm-legal}.
Fix any PROM algorithm $\cB$ and input $x$.
Fix any $i\in[t]$, $\lambda\geq 0$, 
and define 
\ificalpfullversion
	\begin{equation*}
		\epfm_\zeta(\cB,\RO,x,\$)=(P^\RO_1,\dots,P^\RO_t)=((B^\RO_1,M^\RO_1),\dots,(B^\RO_t,M^\RO_t)) 
	\end{equation*}
\fi
\ificalpshortversion
	$\epfm_\zeta(\cB,\RO,x,\$)=(P^\RO_1,\dots,P^\RO_t)=((B^\RO_1,M^\RO_1),\dots,(B^\RO_t,M^\RO_t))$
\fi
for oracle $\RO$. 
\ificalpfullversion
	We may omit the superscript $\RO$ for notational simplicity.
\fi
It holds for all large enough $\sec$ that the following probability is overwhelming:
\ificalpfullversion
	\begin{equation*}
		\Pr\left[\forall i\in[t], ~|P_i|\leq\wsize{\state_i}+\lambda\right]\ ,
	\end{equation*}
\fi
\ificalpshortversion
	$\Pr\left[\forall i\in[t], ~|P_i|\leq\wsize{\state_i}+\lambda\right]$,
\fi
where the probability is taken over $\RO\gets\OSet$ and the coins of $\cB$. 
\end{lemma}
\begin{proof}
Fix any $i\in[t]$ and $\lambda\geq 0$.
We design a predictor $\cP$ to predict the labels of all nodes in $B_i^\RO$, 
as well as the labels of all magic pebbles placed, as follows. 
We refer to the oracle call that causes a node $v\in B_i^\RO$ to be deemed necessary a \emph{critical call}. 
The number $c$ of critical calls is at most $|B_i^\RO|$, i.e., $c\leq|B_i^\RO|$. 
As in the proof of Lemma \ref{lem:epfm-legal}, $\cP$ depends on $x$, $\RO$, 
and a long enough sequence $\hat{R}$ of random bits used to simulate the coins of $\cB$.

\begin{itemize}
	\item \emph{Hint:} The predictor $\cP$ receives as its hint:
		\begin{itemize}
			\item the indices $J=\{j_1,\dots,j_c\}\in[q]^c$ of the critical calls made by $\cB$, and
			\item the state $\sigma_i$ outputted by $\cB$ at the end of iteration $i$, and
			\item the set of indices $I$ of all oracle calls that induce the placement of a magic pebble or pebbles, and
			\item the number of oracle calls that were made up to and including the $i$th iteration, and
			\item $\cB$'s input $x$.
		\end{itemize}
		In the above, $q$ is the number of queries made by $\cB$.
		The size of this hint is $c\log(q)\cdot2^{|\sigma_i|}\cdot q|I|\cdot\log(q)\cdot 2^{|x|}$.
		Recall that by Lemma \ref{lem:epfm-legal}, the set of indices $I$ of all oracle calls that induce the placement of a magic pebble or pebbles has size at most $|I|\leq\chi$ where $\chi=|x|/\sec$.
	\item \emph{Execution:}
		$\cP$ runs $\cB$ on input $(x,z,\sigma_i)$, recording the labels of all input-nodes of the critical calls. To answer any oracle call $Q$ with output-node $v$, the predictor does the following:
		\begin{itemize}
			\item Determines if the call is correct. A call is correct iff it is a critical call or for each parent $w_{i'}$ of $v$, a correct call for $w_{i'}$ has already been made and $Q$ matches the results of those calls. In particular, $Q=\prelab_{\RO,\zeta'}(w_{i'})$ and no new oracle calls need be made by the predictor to check this. 
			\item If the call is correct and the label of $v$ has already been recorded then output the label. Otherwise query $\RO$ to answer the call.
		\end{itemize}
		Finally, $\cP$ outputs predictions of all of the labels of the magic pebbles and all the labels associated with $B_i$, as follows.
		\begin{itemize} 
			\item The labels of the magic pebbles are determined as described in the proof of Lemma \ref{lem:epfm-legal}.
			\item When $\cB$ terminates, $\cP$ checks the transcript to determine the set $B_i$. 
			It is easy to verify that their labels were never queried to $\RO$ by $\cP$. 
			Then, for all $v\in B_i$ the predictor computes $\tilde{v}=\prelab_{\RO,\zeta'}(v)$ 
			and outputs the pair $(\tilde{v},\ell_v)$ where $\ell_v$ is the label of $v$ 
			(as specified in the input of the oracle call for associated critical call).
		\end{itemize}
\end{itemize}

Assume for contradiction that there is some $\lambda\geq 0$ such that with non-negligible probability for some $i\in[t]$ it holds that $|P_i|>\wsize{\sigma_i}+k/\sec+\lambda/\sec$.
Let $\cE$ denote the event that the predictor succeeds (i.e., outputs all correct guesses).
By construction, $\Pr[\cE]$ is non-negligible.
The predictor's output is $\sec(|I|+|P_i|)$ bits long.
From Lemma \ref{lemma1}, it follows that the probability (over the choice of $\RO$ and the randomness of $\cP$) that there is some hint such that $\cP$ outputs all correct guesses is at most
\begin{equation}\label{eqn:epfm-space-prob}
\frac{c\log(q)\cdot 2^{|\sigma_i|}\cdot q|I|\cdot \log(q)\cdot 2^{|x|}}{2^{\sec(|I|+|P_i|)}}\ .
\end{equation}
\eqref{eqn:epfm-space-prob} is negligible 
if $|P_i|>\wsize{\sigma_i}+\lambda$.
By assumption, this inequality holds with non-negligible probability. 
Thus, \eqref{eqn:epfm-space-prob} contradicts our earlier observation that by construction, $\Pr[\cE]$ is non-negligible. The lemma follows.
\end{proof}

\section{Static-memory-hard functions}\label{sec:memory-hardness}

We now define \emph{static-memory-hard functions}.
As mentioned above, prior notions of memory-hardness consider only dynamic memory usage.
To model static memory usage, we consider a hash function with two parts $(\cH_1,\cH_2)$ where $\cH_2(x)$ computes the output of the hash function $h(x)$ given oracle access to the output of $\cH_1$.
This design can be seen to reduce honest party computation time by limiting the hard work to one-off preprocessing phase, while maintaining a large space requirement for password-cracking adversaries. 
Informally, our guarantee says that unless the adversary stores a specified amount of \emph{static} memory, he must use an equivalent amount of \emph{dynamic} memory to compute $h$ correctly on many outputs.
Definition~\ref{def:preprocessing_mhf} is syntactic and Definition~\ref{def:hardness} formalizes the memory-hardness guarantee.

\paragraph{Notation}
PPT stands for ``probabilistic polynomial time.''
For $\vec{b}\in\{0,1\}^*$, define $\Seek_{\vec{b}}:\{1,\dots,|\vec{b}|\}\to\{0,1\}$ to be an oracle that on input $\iota$ returns the $\iota$th bit of $\vec{b}$.

\begin{definition}[Static-memory hash function family (SHF)]\label{def:preprocessing_mhf}
A \emph{static-memory-hard hash function family} $\cH^\RO=\{h^\RO_\sec : \{0,1\}^{w'}\to\{0,1\}^w\}_{\sec\in\NN}$ mapping $w'=w'(\sec)$ bits to $w=w(\sec)$ bits
is described by a pair of deterministic oracle algorithms $(\cH_1,\cH_2)$ such that for all $\sec\in\NN$ and $x\in\{0,1\}^n$,
\ificalpfullversion
	\begin{equation*}
	\cH_2^{\Seek_{\hat{R}}}(1^\sec,x)=h_\sec(x), \mbox{ where }
	R=\cH_1(1^\sec)\ .
	\end{equation*}
\fi
\ificalpshortversion
	$\cH_2^{\Seek_{\hat{R}}}(1^\sec,x)=h_\sec(x), \mbox{ where }
	R=\cH_1(1^\sec)$.
\fi
(The superscript $\RO$ is left implicit.)
\end{definition}

\begin{definition}[$(\Lambda,\tau,q)$-hardness of SHF]\label{def:hardness}
	Let $\cH=\{h_\sec\}_{\sec\in\NN}$ be a static-memory hash function family described by algorithms $(\cH_1,\cH_2)$, mapping $w'$ to $w$ bits.
	$\cH$ is \emph{$(\Lambda,\tau)$-hard} if for any large enough $\sec\in\NN$,
	any $\Delta=\Delta(\sec)\in\omega(\log(\sec))$,
	any string $R\in\{0,1\}^{\Lambda-\Delta}$,
	and any PPT algorithm $\cA$, 
	for any set $X=\{x_1,\dots,x_q\}\subseteq\{0,1\}^{w'}$,
	there is a negligible $\eps$ such that
	\begin{gather*}\label{eqn:cheatonstring}
	\Pr_{\RO,\rho}\Big[
	\big\{(x_1,h_\sec(x_1)),\dots,(x_q,h_\sec(x_q))\big\}=\cA(1^\sec,R;\rho)
	\wedge
	\sattime_{\OSet}(\Lambda,\cA,R,\rho)\geq \tau\Big]<\eps\ .
	\end{gather*}

\end{definition}

For simplicity, we henceforth assume $w'=w=\sec$ (i.e., the oracle's input and output sizes are equal to the security parameter) unless otherwise stated.

\subh{The role of $q$.}
The parameter $q$ in Definition~\ref{def:hardness} serves to capture
the intuitive idea that an adversary that uses a certain amount of space
could always use that space to directly store output values of $h_\sec$.
Clearly, an adversary with an arbitrary input $R$ 
could very easily output up to $\wsize{|R|}$
correct output values. Our goal is to lower bound the amount of space
needed by an adversary who outputs nontrivially more correct values
than that --- and $q$, which is a function of $|R|$, captures how many more.


\section{SHF constructions}\label{sec:constructions}

\ificalpshortversion
Recall from Definition~\ref{def:incrementally-hard} that the property we want in our SHF constructions (more specifically, our $\cH_1$ constructions) is this ``locally hard to access'' notion, meaning that if an adversarial party chooses to not store the static part of our hash function which they obtain from performing the ``preprocessing'' computation associated with $\cH_1$, then they must use the same memory and sustained time to recompute the function when our static-memory-hard function is called on \emph{any subset of inputs} larger than the memory used to store the preprocessed computation.
We achieve this desired property in our $\cH_1$ functions using two novel DAG constructions, one of which is optimal for a specific graph class and the other we conjecture to be optimal for all general graph classes.
\fi

\ificalpfullversion
\paragraph{A first attempt}
What if we pebble a hard-to-pebble graph, and then let $R_{k,i}=H(P(k), i)$ where $P(k)$ is the entire pebbling of the graph (on input $k$ and iteration $i$ is the $i$-th call to the hash function $H$)? 
This would in fact work in the random oracle model where the random oracle takes arbitrary-length input.
However, in practice, hash functions do not take arbitrary-length input.
While constructions like Merkle-Damg{\aa}rd~\cite{Mer79} and sponge~\cite{BDPA08} can transform a fixed-input-length hash function into one that takes arbitrary-length inputs, the resulting function does \emph{not} behave like a random oracle even if the fixed-length hash function does.\footnote{For example, both the constructions mentioned process the input sequentially in chunks. Evaluating the hash function on inputs that differ only in the final chunk will yield outputs that differ in a known way; this provides a way to distinguish these constructions from a random oracle even if the underlying fixed-length hash function is a random oracle.} Moreover, the computation graphs of known length-expanding transformations such as Merkle-Damg{\aa}rd and sponge functions require very little space to compute. For instance, the computation graph of the Merkle-Damg{\aa}rd construction is a binary tree and the computation graph of the sponge function is a caterpillar graph both of which take logarithmic and constant space, respectively, to compute. Thus, we have to use special constructions to achieve the local-hardness properties we need.

Recall from Definition~\ref{def:incrementally-hard} that the property we want is this ``locally hard to access'' notion, meaning that if an adversarial party chooses to not store the static part of our hash function which they obtain from performing the ``preprocessing'' computation associated with $\cH_1$, then they must use the same memory and sustained time to recompute the function when our static-memory-hard function is called on \emph{any subset of inputs} larger than the memory used to store the preprocessed computation.
We achieve this desired property in our $\cH_1$ functions using two novel DAG constructions, one of which is optimal for a specific graph class and the other we conjecture to be optimal for all general graph classes.
\fi

\subsection{$\cH_1$ constructions}\label{sec:h1-constructs}

\ificalpfullversion
	We first note the differences between the graph constructions we present here and the constructions presented in previous literature~\cite{AS15,ACKKPT16,ABP17,DFKP15}. Firstly, many of the constructions presented in previous work feature a single target node. This is reasonable in the context of memory-hard functions since both the honest party and the adversary must compute the hash function dynamically (obtaining a single label as the output of the function) on each input. However, in our context of static-memory-hard functions, single-target-node constructions do not make sense. Secondly, our constructions differ from even the multiple target node constructions presented in the literature (specifically, the constructions of~\cite{DFKP15}) since prior constructions mainly focused on finding graphs that have large memory vs. time tradeoffs.
\fi

Our constructions are designed with the goal that any adversary that does not store almost all the target labels must dynamically use \emph{the same amount of space as needed to store all the labels} to compute the hash function (while \emph{still incurring a cost in runtime}). Moreover, our constructions based on local hardness ensure a stronger guarantee than the constructions in~\cite{DFKP15}; in our case, one must use at least $S$ space (for some definition of $S$) to compute \emph{any} given subset of targets larger than one's current memory usage, whereas in their case, they use $S$ space to compute some subset of targets chosen uniformly at random. Therefore, our specifications are stronger in that we provide a space bound as well as a time bound for adversaries; and moreover, for \emph{honest} parties, the time cost is only a one-time setup cost. We prove our pebbling costs in terms of the black-magic pebble game (defined in Section~\ref{sec:pebbling-defs}) as opposed to the standard pebble game used in previous works. Most notably, this means that in all of our constructions, the pebbling number is upper bounded by the number of targets (since one can always just pebble the targets with magic pebbles).

\ificalpfullversion{
We begin with some simple and clean constructions of $\cH_1$ based on pebbling constructions that exist in the literature. 
We first prove a lemma regarding the minimum number of pebbles used in the PROM model and the minimum number of pebbles used in the sequential memory model.
This is useful in more than one way: (1) it tells us that parallelization does not save the adversary in space so honest parties (who can only compute a constant number of labels at a time) and adversaries (who can compute an arbitrary number of labels at the same time) operate under the same space constraints and (2) it allows us to directly compare sustained time complexities between adversaries and honest parties with respect to space usage 
.

\begin{lemma}[Standard Pebbling Sequential/Parallel Equivalence]\label{lem:sec-par-equiv}
Given a DAG $G = (V, E)$, $\p(G, T) = \p^{\parallel}(G, T)$ where $\p(G, T)$ is defined to be the minimum standard pebbling space complexity in the sequential model, and we define $\p^{\parallel}(G, T)$ to be the minimum standard pebbling space complexity in the parallel model.
\end{lemma}

\ificalpfullversion
\begin{proof}
Any sequential pebbling strategy, $\s$ can be simulated by a parallel pebbling strategy, $\ps$ since $\ps$ can choose to place one pebble at a time. Therefore, $\p^{\parallel}(G, T) \leq \p(G, T)$. We now show that there exists a sequential pebbling strategy, $\s$, that uses the same number of pebbles to pebble a graph as a parallel strategy $\ps$. Suppose that at time $i$, a set of pebbles are added to nodes in $P_i$ in $G$ under algorithm $\ps$. Then, $pred(P_i)$ must be pebbled at time $i-1$. $\s$ can thus spend $|P_i\backslash P_{i-1}|$ pebbling steps to pebble the graph sequentially by adding pebbles on all vertices $v \in P_i \backslash P_{i-1}$ sequentially until the state of the graph is the same as the state of the graph at time $i$ under strategy $\ps$. Similarly, if a set of pebbles $D_i$ are deleted from the graph at time $i$, then $\s$ can choose to spend at most $|D_i|$ sequential pebbling steps to delete $|D_i|$ pebbles. If both strategies start on identical graphs with the same starting configuration $P_0$, then we have shown that $\p^{\parallel}(G, T) \geq \p(G, T)$. Thus, $\p^{\parallel}(G, T) = \p(G, T)$. 
\end{proof}
\fi

\ificalpfullversion
	We use Lemma~\ref{lem:sec-par-equiv} to prove an equivalent lemma for the black-magic pebble game below.
\fi

\begin{lemma}[Black-Magic Pebbling Sequential/Parallel Equivalence]\label{lem:magic-sec-par-equiv}
Given a DAG $G = (V, E)$, $\scostm(G, |T|, T) = \scostm^{\parallel}(G, |T|, T)$ where $\scostm(G, |T|, T)$ was defined to be the minimum black-magic pebbling space complexity in the sequential model, and we define $\scostm^{\parallel}(G, |T|, T) $ to be the minimum black-magic pebbling space complexity in the parallel model.
\end{lemma}
\fi

\ificalpfullversion
\begin{proof}
Any placement of black pebbles can be translated from the sequential to the parallel pebbling strategy and vice versa using the techniques stated in the proof of Lemma~\ref{lem:sec-par-equiv}. Any sequential pebbling placement of magic pebbles can be simulated trivially by a parallel pebbling strategy. Any parallel pebbling placement of $M$ magic pebbles can be simulated via a sequential pebbling strategy using $M$ additional steps. Thus, $\scostm(G, |T|, T)  = \scostm^{\parallel}(G, |T|, T) $.
\end{proof}

Now, we jump into our constructions. We first provide a simple construction and show why this construction is not optimal. In addition, we define some subgraph components in the pebbling literature that are important subcomponents of our constructions.

\subsubsection{A failed attempt at $\cH_1$} We first provide a failed attempt at constructing $\cH_1$ due to the large amount of time
that is needed to compute the function (for the sequential honest party) with respect to the amount of memory needed to store the output of the function. 
In other words, this construction is problematic in the sense that an exponential number of steps is necessary to compute the stored results of the function from scratch for
the honest party but the adversary with parallel processing time can compute the function from scratch in linear time. Although the honest party
could obtain the results of the preprocessing (i.e.\ the static part of the hash function) from elsewhere, we must ensure that they
can still feasibly compute $\cH_1$ themselves in the event that they do not trust any of the sources from which they can obtain
the static data.

Intuitively, our failed attempt at constructing $\cH_1$ is a series of binary search trees. From here onwards, we describe
all constructions of $\cH_1$ as a directed acyclic graph with $n$ nodes and later use our theorems above to 
prove static memory hardness from our constructed DAGs. 

\begin{gconstruction}[Composite Binary Tree DAG]\label{def:binary-search-tree-compression}
Let $B^C_h$ be a composite binary tree DAG with height $h$ constructed in the following way where $T$ is the number of targets of our DAG. Let $s \defeq |T|$. In our intended construction $h = s$.

\begin{enumerate}
\item Let the set of nodes be $V$. Let the set of edges be $E$.
\item Create $(s+1)2^{h-1} + s$ nodes. 
\item Create $s+1$ binary search trees using $(s+1)2^{h-1}$ nodes in total where edges are directed from children to parents in each binary tree. Let $r_i$ for $i \in [1, s+1]$ be the roots of these binary search trees.
\item Order the remaining nodes in some arbitrary order, let $s_j$ be the $j$th node in this order for $j \in [1, s]$. 
\item Create directed edges $(r_i, s_i)$ and $(r_{i+1\Mod s}, s_{i})$ for all $i \in [1, s]$.
\end{enumerate}
\end{gconstruction}

Given any binary search tree with height $h$, the minimum number of pebbles necessary to pebble the tree is $h$ (assuming a `tree' with one node has height $1$) using the rules of the standard pebble game. Therefore, to ensure that the apex of the tree is pebbled and that both the honest party and the adversary both use $h$ space to pebble the apex, the number of leaves necessary at the base of the tree is $2^{h-1}$. If we suppose that
the computationally weak honest party (who does not build special circuits) can only evaluate a constant number of random oracle calls at a time (place a constant number of pebbles), the number of sequential evaluations necessary for the honest party is $\geq \Omega(2^h)$ which is infeasible to accomplish. In constrast, the adversary only has to make $O(h)$ parallel random oracle calls, an exponential factor difference between the honest party and the adversary! Such a construction fails since it is clearly infeasible for the honest party since they would never be able to compute all target values of $\cH_1$ from scratch (since this computation requires exponential time for the honest party). 
Thus, we would like a construction that has the same minimum space requirement but also small sequential evaluation time. We prove a better (but also simply defined) construction below.
\fi

\ificalpfullversion
    \subsubsection{Cylinder construction}

\fi
\ificalpshortversion
    \subh{Cylinder construction}
\fi
We make use of what is defined in the pebbling literature as a \emph{pyramid graph}~\cite{GLT79} in constructing our \emph{cylinder graph}. The key characteristic of the pyramid graph we use is that the number of pebbles that is required to pebble the apex of the pyramid is equal to the height of the pyramid~\cite{GLT79} using the rules of the standard pebble game. Note that a pyramid by itself is not useful for our purposes since the black-magic pebbling space complexity of a pyramid with one apex is $1$. Therefore, we need to be able to use the pyramid in a different construction that uses superconstant number of pebbles in the magic pebble game in order to successfully pebble all target nodes.

\begin{gconstruction}[Illustrated in Fig.~\ref{fig:wraparound}]
\label{def:wrappyramid}
Let $\Pi^C_h$ be a \emph{cylinder} graph with height $h$. We define $\Pi^C_h$ as follows:

\begin{enumerate}
\item Create $2h^2$ nodes. Let this set of $2h^2$ nodes be $V$.
\item Arrange the nodes in $V$ into $2h$ levels of $h$ nodes each, ranging from level $0$ to level $2h-1$. Let the $j$-th node in level $i$ be $v_i^j$. Create directed edges $(v_{i}^{j\Mod{h}}, v_{i+1}^{j \Mod{h}})$ and $(v_{i}^{j\Mod{h}}, v_{i+1}^{(j+1)\Mod{h}})$ for all $i \in [0, 2h-2]$. Let this set of edges be $E$.
\end{enumerate}

\end{gconstruction}

 	\begin{figure}[h!]
        \centering
        \includegraphics[width=0.5\textwidth, angle=0]{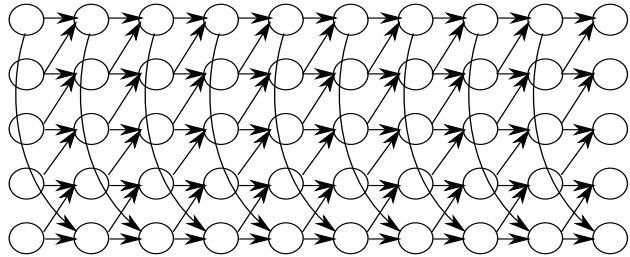}
        \caption{Cylinder construction (Def.~\ref{def:wrappyramid}) for $h=5$.}
		\label{fig:wraparound}
    \end{figure}

\ificalpfullversion
\begin{lemma}\label{lem:pyramids-tradeoff-proof}
Given a cylinder graph with height $h$, $\pyr$, $\p(\pyr, T) \geq h$. 
\end{lemma}

\begin{proof}
Let $T$ be the target nodes of $\pyr$. Each target node is connected to a pyramid of height $h$. Therefore, by the proofs of minimum pebbling cost of pyramids given in~\cite{GLT79}, the pyramid requires $h$ pebbles to pebble using the rules of the standard pebble game. Therefore, to pebble any one target node $t \in T$ requires $h$ pebbles, so pebbling all target nodes of $\pyr$, $T$, trivially requires $h$ pebbles.
\end{proof}

\begin{lemma}\label{thm:pyramds-ptime}
$\stime(\pyr, T) \geq 2h$.
\end{lemma}

\begin{proof}
The depth of $\pyr$ is $2h$ (i.e.\ the longest directed path in $\pyr$ has length $2h$). Thus, the minimum number of parallel steps necessary to pebble any $v \in T$ is $2h$. Let $L_i$ be the set of nodes at the $i$-th level of $\pyr$ where $T$ is at level $2h-1$ and $S$ is at level $0$. To pebble each target node requires that all vertices in $L_{h-1}$ ($v_{h-1}^i$ for all $i \in [1, h]$) be pebbled at some time step simultaneously\footnote{Whereby `simultaneously', we mean there exists some time $t'$ where all vertices in  $L_{h-1}$ are pebbled.}, $t \in [0, t_{\s}]$,  by normality of pebbing strategies (see the definition of \emph{frugal} and \emph{normal} strategies in Definitions~\ref{def:frugal-strategy} and~\ref{def:normal-strategy}~\cite{GLT79,DL17}).
given any normal strategy $\s$. Thus, at least $h$ parallel time steps where $h$ pebbles are on the graph simultaneously are necessary to pebble any target $v \in T$ because to pebble all nodes in $L_{h-1}$ at time $t$ requires $h$ parallel time steps where $h$ pebbles are used at each time step.  

Suppose for contradiction that $\stime(\pyr, T) < 2h$. We first prove that to pebble any $k$ targets (where $k \leq h$) simultaneously require at least $k$ time steps (where each time step is larger than $t$ defined above) where $h$ pebbles are on the graph simultaneously. Furthermore, there exists time steps $t_{l-1} > t_{l-2} > \cdots  > t_1> t$ where $h$ pebbles are on all vertices in $L_{h-1+j}$ ($v_{h - 1 + j}^i$ for all $i \in [1, h]$) at time $t_j$. We prove this by induction. Let the base case be $k = 1$. In order to pebble any target $v \in T$ using a normal strategy $\s$, there must be a time step $t_1 > t$ where $h$ pebbles are on all vertices in $L_{h}$ ($v_{h}^i$ for all $i \in [1, h]$) by normality of pebbling strategies (see Theorem~\ref{thm:normal-strategy} 
~\cite{GLT79}). We assume as our induction hypothesis that the statement is true for all $k \leq l-1$ where $l \leq h$. We now prove the statement for $k = l$. At time $t_{l-1}$, there exist $h$ pebbles on all vertices in $L_{h+l-2}$ by definition of $t_{l-1}$ and by our induction hypothesis. By inspection, to pebble any subset of $l$ targets requires all vertices in $L_{h+l-1}$ to be pebbled at some point in the execution of the pebbling strategy. Suppose there exists a strategy that pebbles $k$ targets using at most $k-1$ parallel moves where $h$ pebbles are on the graph during each of the $k-1$ parallel moves. By our induction hypothesis, pebbling any $k-1$ sized subset of the $k$ targets requires $k-1$ parallel moves where $h$ pebbles are on the graph and all nodes in $L_{h+k-1}$ for all $k < l$ are pebbled simultaneously at time $t_k$. If no more than $h-1$ pebbles can be on the vertices in $L_{h + l-1}$, this means that there exists a vertex in $L_{h+l-1}$ that must be pebbled with at least $l$ pebbles (given there exists a previous time step when $h$ pebbles are on all vertices in $L_{h+l-2}$ and no more than $h-1$ of these pebbles can be moved to the vertices in $L_{h+l-1}$). Let this vertex be $u$. If we continue strategy $\s$ without pebbling $u$, then there will exists a vertex at every level $h + l' -1$ (for all $l' \geq l$) where $l'$ pebbles are necessary to pebble the vertex. Thus, the lower bound on the minimum number of pebbles necessary to pebble $k$ targets using strategy $\s$ is $h -1 + l'$ at some time step $t_{l'} > t_{l-1}$, a contradiction since $l' \geq 1$\footnote{Note that a simpler proof can be shown to state that at least $h$ pebbles are needed to pebble $u$ at level $l'$ but we present the present proof to show that even for a cylinder with height $h$ (instead of $2h$) our proof here still holds--i.e.\ $h$ steps where $h$ pebbles are on the cylinder are necessary to pebble all targets $T$}.

Given that to pebble any $k$ targets requires at least $k$ time steps (inaddition to the $h$ timesteps necessary to pebble all nodes in $L_{h-1}$) where $h$ pebbles are on the graph simultaneously. Thus, pebbling all targets using any strategy that pebbles sequentially subsets of targets $S_1, \dots, S_d$ where $\bigcup_{i = 1}^d S = T$ results in $\sum_{i = 1}^d |S_i| \geq h$ steps where $h$ pebbles are on the graph simultaneously. 
In all cases, we reach a contradiction with $\stime(\pyr, T) < 2h$. Therefore, $\stime(\pyr, T) \geq 2h$. 
\end{proof}

\begin{theorem}\label{thm:sustained-wraparound}
Using the rules of the standard pebble game, $h$ pebbles are necessary for at least $h$ parallel steps to pebble \emph{any} target of a height $2h$ cylinder graph, $\pyr$. 
\end{theorem}

\begin{proof}
To pebble any target of $\pyr$ requires $h$ pebbles on all nodes in level $h$ by normality of pebbling strategies. Given at most $h$ pebbles, to pebble any subset $k$ of nodes in level $h$ (by the normality of pebbling strategies) require $h$ pebbles to be present on the graph for at least $k$ parallel time steps as proven in the proof for Lemma~\ref{thm:pyramds-ptime}. Thus, given a pebbling strategy that pebbles the following subsets of nodes in level $h$ sequentially, $S_1, \dots, S_d$ where $T = \bigcup_{i = 1}^d S_i$, the number of time steps where $h$ pebbles are on the graph is given by $\sum_{i = 1}^{d} |S_i| \geq h$. Therefore, $h$ pebbles are on the graph during at least $h$ time steps when pebbling any target of $\pyr$, proving our theorem.
\end{proof}
\fi

\ificalpshortversion
We first prove the black-magic pebbling space cost of pebbling \emph{all} the target nodes of the cylinder construction defined in Graph Construction~\ref{def:wrappyramid}.
\fi 

\begin{theorem}\label{thm:pyramid-locally-hard}
$\scost(\pyr, |T|, T) \geq h$ where $\pyr$ is defined as in Def.~\ref{def:wrappyramid} where $|S| = |T| = h$.
\end{theorem}

\begin{proof}

Assume for the sake of contradiction that $s < h$ pebbles can be used to pebble all target nodes in $T$. By the rules of the black-magic pebble game, we can choose to use either magic pebbles or black pebbles at each time step in a valid strategy. 

We first prove that given $s < |T|$ magic pebbles, one would choose to place the pebbles on $s$ target nodes as opposed to any number of intermediate nodes. Let $L_i$ be the set of nodes at the $i + h$-th level of $\pyr$ (for $0 \leq i \leq h-1$) where $T$ is at level $2h-1$ and $S$ is at level $0$. Given $s$ adjacent pebble placements on nodes in $L_i$, we can pebble at most $j \leq \max(0, s + i - h + 1)$ target nodes by construction of $\pyr$ without performing any repebbling of any nodes in $S$. (Note that we do not need to account for the case when $s < |T|$ pebbles are placed on levels $0$ to $h-1$ since no targets can be pebbled if that is the case.) If repebbling of any node in $L_i$ needs to be done (using black pebbles), then at least $h$ total pebbles are necessary to pebble $T$. We now show this is true. Suppose that in order to pebble a target node $v \in T$, there exist at most $h - i - 1$ magic pebbles on adjacent nodes in $L_i$. Then, at least $1$ additional pebble is necessary at some node in $L_i$ to pebble $v$. Let the node that needs to be pebble in $L_i$ be $w$. Suppose that we use a black pebble to pebble $w$ at level $i$ (i.e.\ we wouldn't choose to use magic pebbles to pebble the ancestors of $w$ since that would use more magic pebble than if we used a magic pebble to pebble $w$). Note that $w$ is the apex of a pyramid of height at least $i+1$. Therefore, at least $i+1$ black pebbles are necessary to pebble $w$ resulting in $i + 1 +h - i -1 = h$ total pebbles necessary to pebble $v$, which is greater than the initial $h-i-1$ magic pebbles in total pebble count for all $i \in [0, h-1]$ (our desired range of values of $i$). Note that this argument applies recursively to any number $i' \leq i$ missing pebbles at level $i$. 

Therefore, for any number of magic pebbles $s' \leq s$ that are \emph{not} on target nodes, we can obtain at most $s' - 1$ target values without performing repebbling of any nodes in $S$. It is then strictly more efficient to pebble $s'$ target nodes with magic pebbles instead of $s'$ non-target nodes. We can have a total of $s < h$ magic pebbles which is not enough pebbles to pebble all the target nodes. To pebble the target node that is not pebbled by a magic pebble, we require $h$ additional pebbles by pebbling price of pyramids~\cite{GLT79}, contradicting our assumption.
\end{proof}

\ificalpshortversion
As a simple extension of our theorem and proof above, we get Corollary~\ref{cor:cylindar-incrementally-hard}. Moreover, as an extension of the proof given for Theorem~\ref{thm:pyramid-locally-hard} that all magic pebbles are placed on targets and from lemmas regarding the sustained complexity of pebbling \emph{any one target} in the attached full version, we obtain Corollary~\ref{cor:cylinder-subset-sustained-time}.
\fi

\ificalpfullversion
As a simple extension of our theorem and proof above, we get Corollary~\ref{cor:cylindar-incrementally-hard}. Moreover, as an extension of the proof given for Theorem~\ref{thm:pyramid-locally-hard} that all magic pebbles are placed on targets and from Theorem~\ref{thm:sustained-wraparound}, we obtain Corollary~\ref{cor:cylinder-subset-sustained-time}.
\fi

\begin{corollary}\label{cor:cylindar-incrementally-hard}
Given a cylinder $G = (V, E)$ as constructed in Graph Construction~\ref{def:wrappyramid}, $G$ is incrementally hard: $\scostm(\g,|C| - 1, C) \geq |T|$ for any subset $C \subseteq T$.
\end{corollary}

\begin{corollary}\label{cor:cylinder-subset-sustained-time}
Given a cylinder $G = (V, E)$ as constructed in Graph Construction~\ref{def:wrappyramid}, $\stimem(G, |C| - 1, C) = \Theta(|T|)$ for all subsets of $C \subseteq T$. 
\end{corollary}

A logical question to ask after constructing our very simple hash function based on a cylinder graph is whether such a construction is optimal in terms of graph-optimal sustained complexity \emph{and} follows our requirements for a static-memory-hard hash function. As it turns out, the graph-optimal sustained complexity of a cylinder graph is optimal in the class of layered graphs. In other words, if we choose to use layered graphs in our constructions, then we cannot hope to get a better memory and time guarantee. From an implementation and practical standpoint, layered graphs are easier to implement and hence this result has potential practical applications (as more complicated constructions need to consider memory allocation factors in the real-life implementation, not considered in the theoretical model).

\begin{theorem}\label{thm:magic-pebble-game}
Given a \emph{layered graph}, $G = (V, E)$, if the number of target nodes is $|T| = s$ and $\scost(\g, s, T) \geq s$, then $|V| = \Omega(s^2)$. A \emph{layered graph} is one such that the vertices can be partitioned into layers and edges only go between vertices in consecutive layers.
\end{theorem}

\begin{proof}
In order to satisfy $\scostm(\g, s, T) \geq s$, the number of targets has to be at least $s$; if $|T| < s$, then $T$ can be completely pebbled with less than $s$ magic pebbles and $\scostm(\g, s, T) < s$. Suppose the sources (the first level) are at level $0$ and the targets (the last level) are at level $h-1$ where $h$ is the height of the layered graph. In any layered graph with in-degree $2$, the cost of pebbling a vertex $v_i$ in level $i$ is at most $i+1$~\cite{Nor15}. Therefore, the height of $G$ must be at least $s - 1$, in order for $\scost(\g, s, T) \geq s$. Let $h = s - 1$. In order for $\scostm(\g, s, T) \geq s$, the width of the layered graph in layer $j$ for all $j \in \left[\frac{h}{2}, h-1\right]$ must be at least $\frac{h}{2}$ (where by width, we mean the number of nodes in layer $j$). 

Suppose that a layer $j$ where $j \in \left[\frac{h}{2}, h-1\right]$ has width less than $\frac{h}{2}$. We can subsequently use less than $\frac{h}{2}$ magic pebbles to pebble layer $j$. Then, at most $\frac{h}{2}$ black pebbles are necessary to pebble all targets in $T$ resulting in $\scostm(\g, s, T) < h$ and $\scostm(\g, s, T) < s$ (by our definition of $h$), a contradiction. The total number of nodes in layers $[\frac{h}{2}, h-1]$ must then be at least $\frac{h^2}{4}$, and $|V| = \Omega(h^2) = \Omega(s^2)$. 
\end{proof}

\ificalpfullversion
\fi

Thus, our construction of the cylinder graph is optimal in terms of amount of memory used in the asymptotic sense for the class of layered graphs. An open question is whether this is also optimal when we consider the larger class of all DAGs.

\ificalpfullversion
    \begin{openquestion}
    Does Thm~\ref{thm:magic-pebble-game} also hold for general graphs with bounded in-degree $2$? 
    \end{openquestion}
\fi

Given the impossibility of providing a better space guarantee for layered graphs, we provide a general (non-layered) construction that transforms a graph from a certain class into another graph with the same space guarantee as in Theorem~\ref{thm:magic-pebble-game}. Furthermore, we provide an example below that has the same space guarantees but a better time guarantee.

\ificalpfullversion
    \subsubsection{Layering \emph{shortcut-free} graphs}

\fi
\ificalpshortversion
    \subh{Layering \emph{shortcut-free} graphs}
\fi
We now show how to convert any \emph{shortcut-free} DAG, $G = (V, E)$, with $\scost(G, T) = s$ and one target node (i.e. $|T| = 1$) into a DAG, $G' = (V', E')$, with $|T'| = s$ targets and $\scostm(G', s, |T'|) = s$.

\begin{definition}[Shortcut-Free Graphs]\label{def:shortcut-free-graph}
 Let $G = (V, E)$ be a DAG where $\scost(G, T) \geq s$. Let $t^{\s}_s$ be the last time step that exactly $s$ pebbles must be on $G$ during any normal and regular pebbling strategy, $\s$, 
(see Thms~\ref{thm:normal-strategy} and~\ref{thm:regular-strategy}, \cite{GLT79,DL17}) that uses $s$ pebbles. More specifically, let 
\ificalpshortversion
    $t^{\s}_s = \arg\max_{t' \in [t_{\s}]} \{|P_{t'}|: |P_{t'}| \geq \scost(\g, T)\}$
    where $t_{\s} = |\s|$ and $P_{t'} \in \s$ for all $t' \in [t_{\s}]$.
\fi
Let $X$ be the union of the set of nodes that are pebbled at $t^{\s}_s$ for all normal and regular strategies 
\ificalpfullversion
\fi
$\s$: $X = \bigcup\limits_{\s \in \mathbb{P}} P_{t^{\s}_s}$. Let $D$ be the set of descendants of nodes of $X$. A DAG is \emph{shortcut-free} if $|X| \leq s$ and given $s_1 < s$ pebbles placed on any subset $X_1 \subset X$, no normal and regular strategy uses less than $s - s_1$ pebbles to pebble $D \cup (X \backslash X_1)$.
\end{definition}

\ificalpfullversion
\begin{gconstruction}\label{def:layering-any-graph}
Given a shortcut-free DAG, $G = (V, E)$, with $\scost(G, T) = s$ and $|T| = 1$, we create a DAG, $G' = (V', E')$, with the following vertices and edges and with the set of targets $T'$ where $|T'| = s$. Let $X$ be defined as in Definition~\ref{def:shortcut-free-graph}.

\begin{enumerate}
\item $V'$ is composed of the nodes in $V$ and $s-1$ copies of $X \cup D$. Let the $i$-th copy of $X$ be $X_i$ (the original is $X_0$) and let the $i$-th copy of $x \in X_i$ be $x_i$. 
\item $E'$ is composed of the edges in $E$ and the following directed edges. If $(v, w) \in E$ and $v, w \in X$, then create edges $(v_i, w_i) \in E'$ for all $i \in [1, s-1]$. Create edges $(u, v_i) \in E'$ if $(u, v) \in E$ and $u \in V \backslash X, D$. 
\item The set of targets $T'$ is the union of the set of targets of the different copies: $T' = \bigcup_{i = 0}^{s-1} T_i$. 
\end{enumerate} 

Using the above construction, we have created a graph $G' = (V', E')$ where $|V'| = |V| + (s-1)(|D| + |X|)$ and $|T'| = s$.
\end{gconstruction}

\begin{theorem}\label{thm:general-layering-transformation}
Given a shortcut-free DAG $G=(V, E)$ with $\scost(G, T) = s$ and $|T| = 1$, the construction produced by Graph Construction~\ref{def:layering-any-graph} produces a DAG $G' = (V', E')$ such that $\scostm(G', s, |T|) = s$.
\end{theorem}

\begin{proof}
We first prove that $\scostm(G', s, |T|) \leq s$. Since there are $s$ different targets, $\scostm(G', s, |T|) \leq s$ trivially. 

We now prove that $\scostm(G', s, |T|) \geq s$. If only black pebbles are used to pebble the targets in $T'$, then $s$ black pebbles must trivially be used provided $\scost(G, T) = s$. Suppose some number of magic pebbles are used. Using the magic pebbles on any node in a copy of $D$ (defined in Def.~\ref{def:layering-any-graph}) that is not a target in $T'$ is strictly worse than using a magic pebble on a target. Suppose the total number of pebbles used is less than $s$. We first prove that no magic pebbles are used on copies of $D$. If the total number of pebbles used is less than $s$, then not all of the $s$ targets can be pebbled using magic pebbles. The remaining target that is not pebbled must be pebbled using $s$ black pebbles since $\scost(G, T) = s$ by definition. By the same logic, no magic pebbles are used on the nodes in the copies of $X$. 

Therefore, if less than $s$ magic pebbles are used to pebble the graph, all magic pebbles should be used to pebble the predecessors of $X$. No magic pebble can be removed and repebbled since such a magic pebble must be placed $s$ times (once for each copy of $X$ and $D$), exceeding the maximum number of magic pebbles we can have. Given that we can use a total of less than $s$ magic pebbles to pebble the predecessors of $X$, suppose some $s' < s$ pebbles are used, then less than $s -s'$ pebbles are left to pebble each copy of $X$ and $D$; by incremental hardness, less than $s-s'$ cannot be used to pebble each copy of $X$ and $D$. At least one magic pebble is used on the predecessors of $X$; by our definition of incremental hardness, less than $s-1$ pebbles cannot be used to pebble $X$ and $D$, a contradiction. Thus, $\scostm(G', s, T) \geq s$. 
\end{proof}
\fi

If $D = \Theta(s)$ and $s = O(\sqrt{|V|})$, then $|V'| = \Theta(s^2 + |V|)$ which has a better sustained time guarantee than our cylinder construction.

We first note that the sustained memory graphs presented in~\cite{ABP17} \emph{do not} achieve optimal local memory hardness because $X \cup D$ 
\ificalpfullversion
(as defined in Definition~\ref{def:layering-any-graph})
\fi 
\ificalpshortversion
(as defined in the layering graph construction given in the attached full version)
\fi
is $\Theta(n)$ (since the sources are the ones that remain pebbled in their construction). Thus, we would like to provide a construction of a shortcut-free DAG where $|X \cup D| = \Theta(s)$. Note that the size of $X \cup D$ will always be $\Omega(s)$, trivially.
\ificalpfullversion
We now provide a definition of a shortcut-free graph class $G$ that can be transformed using Definition~\ref{def:layering-any-graph}.
\fi
\ificalpshortversion
We now provide a definition of a shortcut-free graph class $G$ that can be transformed using the layering technique given in more detail in the full version.
\fi

\begin{gconstruction}[Illustrated in Fig.~\ref{fig:time-optimal}]\label{def:longer-time}
Let $G = (V, E)$ be a graph defined by parameter $s$ and in-degree $2$ with the following set of vertices and edges:

\begin{enumerate}
\item Create a height $s$ pyramid. Let $r_i$ be the root of a \emph{subpyramid} (i.e. a pyramid that lies in the original height $s$ pyramid) with height $i \in [2, s]$. One can pick any set of these subpyramids.
\item Topologically sort the vertices in each level and create a path through the vertices in each level (see Fig.~\ref{fig:time-optimal}). Replace any in-degree-$3$ nodes with a pyramid of height $3$, with a $6$-factor increase in the number of vertices.
\item Create $c_1s$ additional nodes for some constant $c_1 \geq 2$ (in Fig.~\ref{fig:time-optimal}, $c_1 = 6$). Label these nodes $v_j$ for all $j \in [1, c_1s]$.
\item Create directed edges $(r_s, v_1)$ and $(r_i, v_{k(i - 1)})$ for all $k \in [1, s]$. 
\item Create $s-1$ additional nodes. Let these nodes be $w_l$ for all $l \in [1, s-1]$.
\item Create directed edges $(v_{c_1 s}, w_1)$ and $(r_i, w_{i-1})$ for all $i \in [2, s]$. 
\item The target node is $w_{s-1}$. 
\end{enumerate}
\end{gconstruction}

\begin{figure}[h]
\centering
\includegraphics[width=0.5\textwidth, angle=270]{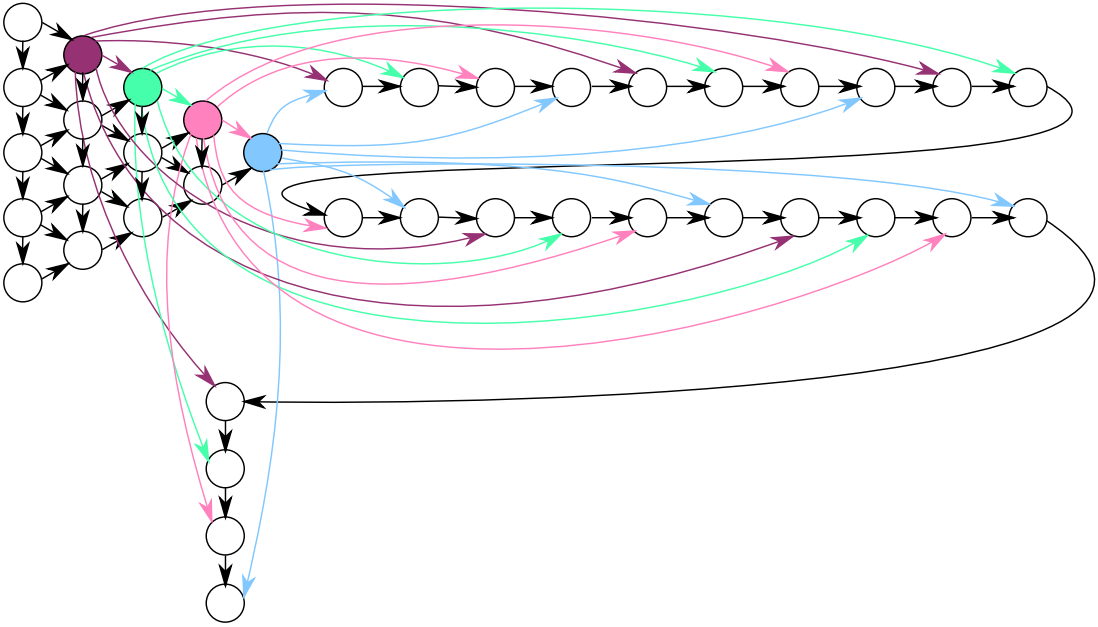}
\caption{Example of a time optimal graph family construction as defined in Def.~\ref{def:longer-time}. Here, $s=5$.}
\label{fig:time-optimal} 
\end{figure}

\ificalpfullversion
\begin{lemma}\label{lem:longer-time-scost}
Given a DAG $G = (V, E)$ and a parameter $s$ where $G$ is defined by Definition~\ref{def:longer-time}, $\scost(G, T) = s$. 
\end{lemma}

\begin{proof}
In order to pebble the apex of the pyramid of height $s$, we must use at least $s$ pebbles as proven in the proof for black pebbling cost of pyramids~\cite{GLT79}.
\end{proof}

Before we prove that $G=(V, E)$ created by Definition~\ref{def:longer-time} with parameter $s$ is shortcut-free, we first prove the following stronger lemma which will help us prove that $G$ is shortcut-free.

\begin{lemma}\label{lem:shortcut-free-prelim}
Let $G=(V, E)$ be a graph created using Definition~\ref{def:longer-time} with parameter $s$. Given a normal strategy $\s$ to pebble $G$, when $v_q$ for $q \in [1, c_1 s]$ is pebbled at some time step, black pebbles are present on all nodes in $[r_{i}, r_s]$ where $i = (q \Mod s-1) + 1$ from the time when $v_1$ is pebbled to when $v_q$ is pebbled.
\end{lemma}

\begin{proof}
We prove this lemma via induction.

In our base case when $i = s$, when the corresponding $v_q$ is pebbled, a black pebble must be on $r_s$ and $v_{q-1}$ in the previous time step. Thus, a black pebble remains on $r_s$ from the time $v_1$ is pebbled till the time that $v_q$ is pebbled or a set of $s-j$ black pebbles remain on the $j$-th level of the pyramid for some $j \in [0, s-1]$ (in which case we can charge one of these pebbles to be ``present on $r_s$''). Suppose neither of these conditions are met. Then, by the pebbling number of pyramids (see Thm~\ref{thm:pyramid-space},~\cite{Nor15}),
at least $s$ pebbles must be used to pebble $r_s$, contradicting the frugality of $\s$ (since at most $s$ pebbles are used to pebble $G$). In general, we make the observation that if there $s-j$ pebbles on some level $j \in [0, s-1]$, then we can charge these $s-j$ pebbles to be ``on all nodes in $[r_j, r_s]$''. 

For our induction hypothesis, we assume that the theorem is true for $j$ and prove the stratement for $i = j - 1$. When $i = j-1$ and the corresponding $v_q$ is pebbled, we assume by our induction hypothesis that there are $s - j$ black pebbles present on $[r_{j}, r_s]$ (or charged to be on $[r_j, r_s]$) from when $v_1$ is pebbled to when $v_q$ is pebbled. In order to pebble $v_q$, there must be black pebbles on $r_i$ and $v_{q-1}$. If there does not exist a black pebble on $r_i$ (or on the predecessors of $r_i$) from when $v_1$ is pebbled to when $v_q$ is pebbled, then at least one pebble must be removed from some $r \in [r_j, r_s]$ or from $v_{q-1}$ since at least $j-1$ pebbles are necessary to pebble $r_{j-1}$ ($s - j + 2$ pebbles are currently in use--leaving not enough pebbles to pebble $r_{j-1}$ unless a pebble is removed). If the black pebble is removed from $v_{q-1}$, the frugality of $\s$ is contradicted. If the black pebble is removed from some $r \in [r_j, r_s]$, then by observation, $r_s$ will need to be repebbled sometime in the future, also a contradiction to the frugality of $\s$. Thus, we prove our statement.
\end{proof}

\begin{lemma}\label{lem:longer-time-incremental}
Given a DAG $G = (V, E)$ and a parameter $s$ where $G$ is defined by Definition~\ref{def:longer-time}, $G$ is shortcut-free.
\end{lemma}

\begin{proof}
We first prove that any normal standard pebbling strategy $\s$ that pebbles $G$ must contain pebbles on all $r_i$ and $v_{c_1 s}$ at some time (say, $t_X$) during the execution of $\s$.

Let $X$ be the set of vertices containing black pebbles when $v_{c_1 s}$ is pebbled. Thus, a total of $s$ pebbles must be on the graph (specifically on all nodes in $X$) at this time in any normal strategy by proof of Lemma~\ref{lem:shortcut-free-prelim}. We now prove the incremental hardness of $G$. Let $s' < s$ pebbles be on $X$ at time $t_X$. We prove that we cannot pebble $X \cup D$ using less than $s - s'$ pebbles. 

Suppose for the purposes of contradiction, given $s' < s$, assume that $s'$ pebbles are placed on $X$ and less than $s-s'$ pebbles can be used to pebble $X\backslash X' \cup D$. Supose that $X \backslash X'$ includes either:

\begin{enumerate}
\item $v_{c_1 s}$ and some $s - s' - 1$ subset of vertices in $[r_2, r_s]$, or
\item some $s-s'$ subset of vertices in $[r_2, r_s]$. 
\end{enumerate}

In the first case, if no pebbles are on $v_i$ for $i \in [1, c_1 s]$, then at least one pebble needs to be used to pebble $v_i$ for $i \in [1, c_1 s]$. If $X\backslash X'$ includes some $s - s' - 1$ subset of vertices in $[r_2, r_s]$, then at least $s - s'$ pebbles are needed to pebble the vertices missing the pebbles. 

In the second case, if some subset $s - s'$ of vertices in $[r_2, r_s]$ are in $X \backslash X'$, then at least $s-s' + 1$ pebbles are necessary to pebble the nodes missing pebbles in order to be able to pebble $w_l$ for $l \in [1, s-1]$. 

In either case, at least $s-s'$ pebbles are necessary to pebble $X \backslash X' \cup D$, thus, this construction is shortcut-free.
\end{proof}

\begin{theorem}\label{thm:sustained-hardness-longer-time}
$s$ pebbles are necessary for at least $\Theta(s^2)$ parallel steps to pebble any target of $G'$. 
\end{theorem}

\begin{proof}
To pebble $v_j$ for all $j \in [1, c_1 s]$, we require pebbles on all $r_i$ for $i \in [2, s]$ and one pebble on the path from $v_1$ to $v_{c_1 s}$; otherwise, the entire pyramid must be rebuilt, resulting in repebbling all nodes in the graph as we showed in the proof of Lemma~\ref{lem:shortcut-free-prelim}. To pebble the pyramid requires $s$ pebbles on the pyramid at all times and takes $\Theta(s^2)$. We show this is true.

Suppose that at some point before pebbling the apex of the pyramid that a pebble is removed from the graph, then, by our requirement that $s-1$ pebbles must remain on $r_i$ for $i \in [2, s]$ and that a pebble must be on the path from $v_1$ to $v_{c_1s}$, the removed pebble cannot be used for either of these tasks. Thus, the entire pyramid must be rebuilt, contradicting the frugality of the strategy.

Thus, $s$ nodes must remain on the graph for $\Theta(s^2 + c_1 s) = \Theta(s^2)$ parallel time steps, proving our theorem.
\end{proof}
\fi

We create $G' = (V', E')$ from $G$ (as constructed using Definition~\ref{def:longer-time}) using 
\ificalpfullversion
Definition~\ref{def:layering-any-graph}
\fi
\ificalpshortversion
the layering construction given in the full version
\fi
, resulting in a graph with $\Theta(s^2)$ total  nodes.

\begin{theorem}\label{thm:longer-time-magic-number}
$\scostm(G', s, T) = s$.
\end{theorem}

\begin{proof}
By Lemma~\ref{lem:longer-time-incremental} the graph is shortcut-free and by Lemma~\ref{lem:longer-time-scost} $\scost(G, T) = s$, therefore, we use Theorem~\ref{thm:general-layering-transformation} to prove that $\scostm(G', s, T) = s$.
\end{proof}

\ificalpfullversion
    By the proof that $G'$ is shortcut-free, we obtain the following corollary that $G'$ is also incrementally hard. Moreover, Corollary~\ref{cor:notcylinder-subset-sustained-time} follows directly from the proof of Theorem~\ref{thm:general-layering-transformation}.
\fi
\ificalpshortversion
    From the proof that $G'$ is shortcut-free, we obtain Corollary~\ref{cor:notcylindar-incrementally-hard}. Moreover, Corollary~\ref{cor:notcylinder-subset-sustained-time} follows directly from the proof of the sustained complexity of the general layering transformation given in the full version.
\fi

\begin{corollary}\label{cor:notcylindar-incrementally-hard}
Given a graph $G = (V, E)$ as constructed in Graph Construction~\ref{def:longer-time}, $G$ is incrementally hard: $\scostm(\g,|C| - 1, C) \geq |T|$ for any subset $C \subseteq T$.
\end{corollary}

\ificalpfullversion
    The following corollary about the graph-optimal sustained time complexity is proven directly from the proof of Lemma~\ref{lem:shortcut-free-prelim} and Theorem~\ref{thm:sustained-hardness-longer-time} that if less than $\frac{s}{2}$ magic pebbles are on the pyramid, then half the pyramid must be rebuilt resulting in $\Theta(s^2)$ time-steps in which $s$ pebbles are on the graph; thus proving for the cases when $|C| - 1 < \frac{s}{2}$. We now prove the case when $|C|-1 \geq \frac{s}{2}$.
\fi

\begin{corollary}\label{cor:notcylinder-subset-sustained-time}
Given a graph $G=(V, E)$ as constructed in Graph Construction~\ref{def:longer-time}, $\stimem(G, |C| - 1, C) = \Theta(|V|)$ for all subsets of $C \subseteq T$. 
\end{corollary}

\begin{proof}
If $|C| - 1 \geq \frac{s}{2}$ magic pebbles are not placed on $r_i$ for all $i \in [2, s]$, then we have to rebuild at least half the pyramid, resulting in $\Theta(s^2) = \Theta(|V|)$ time being used. Thus, some $s' \geq \frac{s}{2}$ magic pebbles must be used on $r_i$ for all $i \in [2, s]$. Then, to pebble all $|C| \geq \frac{s}{2}$ targets requires $\Theta(s^2) = \Theta(|V|)$ time using another black pebble since $s' \geq \frac{s}{2}$ pebbles are used on the pyramid.
\end{proof}

\subsection{$\cH_2$ construction}

Our construction of $\cH_2$ is presented in Algorithm \ref{algo:basic_H2}.

\begin{algorithm}[ht!]\caption{$\cH_2$}\label{algo:basic_H2}
On input $(1^\sec,x)$ and given oracle access to $\Seek_{R}$ 
(where $R$ is the string outputted by $\cH_1$):
\begin{enumerate}
\item Let $\wsize{R}=|R|/w$ be the length of $R$ in words.
\item\label{step:rhos} Query the random oracle to obtain $\rho_0=\RO(x)$ and $\rho_1=\RO(x+1)$.
\item\label{step:sample-loc} Use $\rho_0$ to sample a random $\iota\in[\wsize{R}]$.
\item\label{step:target-label} Query the $\Seek_R$ oracle to obtain $y'=\Seek_R(\iota)$.
\item\label{step:h2-output} Output $y'\oplus\rho_1$.
\end{enumerate}
\end{algorithm}

\begin{lemma}\label{lem:RO}
    For any $R$, the output distribution of $\cH_2$ is 
    uniform over
    \ificalpfullversion the choice of random oracle \fi
    $\RO\gets\OSet$.
\end{lemma}
\begin{proof}
    Over the choice of random oracle, the value $\rho_1$ computed in Step~\ref{step:rhos} is truly random, and $y'$ is independent of $\rho_1$ by construction, so the output $y'\oplus\rho_1$ is also truly random.
\end{proof}

\begin{remark}
Lemma~\ref{lem:RO} is important as an indication that
our SHF construction ``behaves like a random oracle.''
The memory-hardness guarantee alone does not assure
that the hash function is suitable for cryptographic hashing:
e.g., a modified version of $\cH_2$ which directly outputted $y'$
instead of $y'\oplus\rho_1$ would still satisfy memory-hardness,
but would be an awful hash function (with polynomial size 
codomain). 
The inadequacy of existing memory-hardness definitions for assuring that
a function ``behaves like a hash function'' is discussed
by \cite{AT17}.
\end{remark}

\subsection{Proofs of hardness of SHF Constructions}

\ificalpfullversion
    We now prove the hardness of our graph constructions given earlier in Section~\ref{sec:constructions}.  

    We begin by stating two supporting lemmata.
    The first is due to Erd\H{o}s and R\'{e}nyi \cite{ER61}, on the topic of the Coupon Collector's Problem.

    \begin{lemma}[\cite{ER61}]\label{lem:coupon}
        Let $Z_n$ be a random variable denoting the number of
        samples required, when drawing uniformly from a set of $n$ distinct objects
        with replacement, to draw each object at least once.
        Then for any $c$, 
        $\lim_{n\rightarrow\infty}\Pr[Z_n<n\log{n}+cn]=e^{-e^{-c}}$.
    \end{lemma}

    \begin{corollary}\label{cor:coupon}
        Let $Z_{n,k}$ be a random variable denoting the number of samples required, when drawing uniformly from a set of $n$ distinct objects with replacement, to have drawn at least $k\in[n]$ distinct objects.
        Let $q\in\omega(k\log{k})$.
        Then $\Pr[Z_{n,k}<q]$ is overwhelming (in $k$).
    \end{corollary}
    \begin{proof}
        For $m\in\NN$ and $i\in[m-1]$, let $\cE_{i,m}$
        denote the event that after $i$ elements out of a set of $m$ elements have already been sampled uniformly with replacement, the $(i+1)$th sample will coincide with one of the elements already drawn. For any $i\leq k\leq n$, it holds that $\Pr[\cE_{i,n}]\geq\Pr[\cE_{i,k}]$. The desired event of drawing $k$ distinct objects corresponds exactly to the conjunction of $\cE_{i,m}$ for $i\in[k]$.
        Therefore, for all $k\in[n]$ and any $c'$,
        \begin{equation}\label{eqn:coupon-is-enough}
            \Pr[Z_{n,k}<c'] \geq \Pr[Z_k<c']\ .
        \end{equation}

        Hence, it suffices for our purposes to bound $\Pr[Z_k]$.
        From Lemma~\ref{lem:coupon},
        \begin{align*}
            & \lim_{k\rightarrow\infty}\Pr[Z_k<k\log{k}+ck] 
                =\lim_{k\rightarrow\infty}e^{-e^{-c}}.
        \end{align*}
        Applying a Taylor expansion,  we get
        $\Pr[Z_k<k\log{k}+ck]\in O(1-e^{-c})$.
        This probability is overwhelming in $k$ (i.e., $e^{-c}$ is negligible) whenever $c\in\omega(\log(\sec))$.
    \end{proof}

    Theorems~\ref{thm:wraparound}--\ref{thm:highest-memory2} state the static-memory-hardness of our SHF constructions based on Graph Constructions~\ref{def:wrappyramid} and \ref{def:longer-time}.
\fi

\begin{theorem}\label{thm:wraparound}
Define a static-memory hash function family $(\cH_1,\cH_2)$ as follows: let $\cH_1$ be the graph function family $\cF_{\pyr}$ (Graph Construction~\ref{def:wrappyramid}), and let $\cH_2$ be as defined in Algorithm \ref{algo:basic_H2}.
Let $\cH=\{h_\sec\}_{\sec\in\NN}$ be the static-memory hash function family
described by $(\cH_1,\cH_2)$.
For any $\lambda\geq 0$ and $\Lambda\in\Theta(\sqrt{n})$, 
$(\cH_1,\cH_2)$ is $(\Lambda-\lambda,\Theta(\sqrt{n}),q)$-hard 
where $q\in\omega(\Lambda\log\Lambda)$.
\end{theorem}
\begin{proof}
    Fix any $\lambda\geq 0$ and let $\Lambda\in\Theta(\sqrt{n})-\lambda/\sec$.
    Suppose, for contradiction, that the theorem does not hold. 
    Then by Definition~\ref{def:hardness},
    there exist: $\sec\in\NN$, $\Delta\in\omega(\log(\sec))$, 
    a string $R\in\zo^{\Lambda-\Delta}$, an algorithm $\cA$,
    and a set $X=\{x_1,\dots,x_q\}$ such that 
    the following probability is non-negligible:
    \begin{equation}
        \Pr_{\RO,\rho}\Big[
        \big\{(x_1,h_\sec(x_1)),\dots,(x_q,h_\sec(x_q))\big\}=\cA(1^\sec,R;\rho)
        \wedge
        \sattime_{\OSet}(\Lambda,\cA,R,\rho)\geq \tau\Big]\ .
    \end{equation}
    That is, the probability that an execution of $\cA$ on input $R$ outputs a set of $q$ correct input-output pairs of $h_\sec$, without using $\Lambda$ space for at least $\tau$ time-steps, is non-negligible. We denote by $\cE_\rho$ the event that
    \begin{equation*}
        (x_1,h_\sec(x_1)),\dots,(x_q,h_\sec(x_q))\big\}=\cA(1^\sec,R;\rho)
        \wedge
        \sattime_{\OSet}(\Lambda,\cA,R,\rho)\geq \tau\ .
    \end{equation*}

    Given a correct evaluation $y=h_\sec(x)$ of $\cH_2$ 
    on a given input $x$,
    one can easily compute $\rho_0,\rho_1$ by evaluating $\RO$ on $x,x+1$ respectively, and demask $y$ to obtain the value $y'(x)=y\oplus\rho_1$ of the target label computed in Step~\ref{step:target-label} of Algorithm~\ref{algo:basic_H2}. Moreover, the index $\iota$ computed in Step~\ref{step:sample-loc} can be computed as a deterministic function of $\rho_0$.
    Define $\cB'$ to be the deterministic algorithm that on input $(x,y)$ computes $\rho_0,\rho_1$ and $y'$ as described above, and outputs $(\iota,y')$.

    Next, define $\cB$ to be the algorithm that runs $\cA$ and then applies $\cB'$ on each pair $(x_i,y_i)$ outputted by $\cA$,
    and outputs the resulting set $J=\{(\iota_1,y'_1),\dots,(\iota_q,y'_q)\}$ where each $(\iota_i,y'_i)=\cB'(x_i,y_i)$.
    By construction, if $y_i=h_\sec(x_i)$, each $y'_i$ is the correct label of $\iota_i$th target node of the cylinder graph.
    Notice that this means that for each value of $\iota$,
    there is a unique value of $y'$ such that
    $(\iota,y')=\cB'(x,h_\sec(x))$ for any $x$.
    
    Let $I$ denote $|\{\iota_i\}_{x_i\in X}|$.
    Since the set $X$ is fixed before the random oracle, the locations $I$ are distributed uniformly and independently with replacement in a support of size at least $|R|$. Then by Corollary~\ref{cor:coupon}, the number of distinct locations $|I|$ is at least $\Lambda$ with overwhelming probability. That is, there is a negligible function $\eps'$ such that
    $\Pr\left[|I|\geq\Lambda\right]\geq1-\eps'$.
    Conditioned on $\cE_\rho$, all pairs $(x_i,y_i)$ outputted by $\cA$ are such that $y_i=h_\sec(x_i)$, and we have already observed that each value of $\iota$ induces a unique value of $y'$ outputted by $\cB'$ on input pairs of the form $(x_i,h_\sec(x_i))$.
    It follows that $\Pr[|J|\geq\Lambda ~|~ \cE_\rho]\geq1-\eps$. 

    Now consider the ex-post-facto magic pebbling strategy $\s$ induced by $\cB$.
    By Lemma~\ref{lem:epfm-legal}, with overwhelming probability over the choice of random oracle and the coins of $\cB$, $\s$ is legal and uses no more than $\wsize{R}$ magic pebbles; call this event $\cE'_\rho$ (where $\rho$ denotes the randomness of $\cB$).
    By Lemma~\ref{lem:epfm-space}, with overwhelming probability over the same,
    \begin{equation}\label{eqn:pebbling-space-bounded-per-step}
    \forall i\in[t], ~ |P_i^\RO|\leq\wsize{\sigma_i}+\lambda\ ,
    \end{equation}
    where $t$ is the length of $\s$, $P_i$ is the $i$th configuration of $\s$, and $\sigma_i$ is the $i$th state of the execution of $\cB$. We denote by $\cE''_\rho$ the event that \eqref{eqn:pebbling-space-bounded-per-step} is satisfied (where $\rho$ denotes the randomness of $\cB$).

    Finally, we observe that conditioned on $\cE_\rho$, since we established above that $\cB$ outputs a set of at least $\Lambda$ correct target labels, the strategy $\s$ must successfully pebble the corresponding $\Lambda$ target nodes.
    Since $\Pr[\cE_\rho]$ is non-negligible and $\Pr[\cE']$ and $\Pr[\cE'']$ are overwhelming, $\Pr[\cE'\wedge\cE'' | \cE]$ must be negligibly close to $\Pr[\cE]$ (and thus, non-negligible).
    The occurrence of $\cE\wedge\cE'\wedge\cE''$ implies the existence of a pebbling strategy $\s$ that is legal, uses no more than $\wsize{R}$ magic pebbles, and uses fewer than $\sqrt{n}$ pebbles overall.
    This contradicts Corollary~\ref{cor:cylinder-subset-sustained-time}.
\end{proof}

\begin{theorem}\label{thm:highest-memory}
Define a static-memory hash function family $(\cH_1,\cH_2)$ as follows: let $\cH_1$ be the graph function family $\cF_{\sg{2}}$ (Graph Construction~\ref{def:longer-time}), and let $\cH_2$ be as defined in Algorithm \ref{algo:basic_H2}.
Then, for any $\lambda\geq 0$ and $\Lambda\in\Theta(\sqrt{n})$  
where $n$ is the number of nodes in the graph,
$(\cH_1,\cH_2)$ is $(\Lambda-\lambda,\Theta(n),q)$-hard,
where $q\in\omega(\Lambda\log\Lambda)$.
\end{theorem}
\begin{proof}[Proof sketch]
    Identical proof structure to the proof of Theorem~\ref{thm:wraparound}, except instead of invoking Corollary~\ref{cor:cylinder-subset-sustained-time} at the end, we derive a contradiction to Corollary~\ref{cor:notcylinder-subset-sustained-time}.
\end{proof}

The parameter $q$ is suboptimal in Theorems~\ref{thm:wraparound} and \ref{thm:highest-memory}. We can achieve optimality (i.e., $q=\wsize{|R|}$) by the following alternative construction of $\cH_2$: make $q'=\omega(\log(\sec))$ random calls instead of just one call to the $\Seek$ oracle in Step~\ref{step:target-label}.
To preserve the output size of $h_\sec$,
it may be useful to reduce the size of node labels by a corresponding factor of $q'$. This can be achieved by truncating the random oracle outputs used to compute labels in Definition~\ref{def:labeling}. 
\ificalpfullversion
    The description of this altered $\cH_2^{q'}$ and the definition of graph function family $\cF{q'}_{G}$ with shorter labels are given in Appendix~\ref{appx:h2-alternative}.
\fi

\begin{theorem}\label{thm:wraparound2}
Define a static-memory hash function family $(\cH_1,\cH_2)$ as follows: let $\cH_1$ be the graph function family $\cF^{\sec/q'}_{\pyr}$ (Graph Construction~\ref{def:wrappyramid}), and let $\cH_2$ be $\cH_2^{q'}$ as defined in 
\ificalpshortversion the attached full version \fi
\ificalpfullversion Algorithm \ref{algo:fancy_H2} \fi
for some $q'\in\omega(\log\Lambda)$.
Let $\cH=\{h_\sec\}_{\sec\in\NN}$ be the static-memory hash function family
described by $(\cH_1,\cH'_2)$.
Then, for any $\lambda\geq 0$ and $\Lambda\in\Theta(\sqrt{n})$, 
$(\cH_1,\cH_2)$ is $(\Lambda-\lambda,\Theta(\sqrt{n}),q)$-hard 
where $q=\wsize{\Lambda}$.
\end{theorem}
\begin{proof}[Proof sketch]
    Identical proof structure to the proof of Theorem~\ref{thm:wraparound}, except that when invoking Corollary~\ref{cor:coupon}, due to the design of $\cH'_2$ which calls $\Seek$ more times than $\cH_2$, we obtain the stronger statement that an adversary that successfully outputs $q$ pairs $((x_1,h_\sec(x_1)),\dots,(x_q,h_\sec(x_q)))$ must correctly guess $q$ target labels of the graph.
\end{proof}

\begin{theorem}\label{thm:highest-memory2}
Define a static-memory hash function family $(\cH_1,\cH_2)$ as follows: let $\cH_1$ be the graph function family $\cF^{\sec/q'}_{\sg{2}}$ (Graph Construction~\ref{def:longer-time}), and let $\cH_2$ be $\cH_2^{q'}$ as defined in 
\ificalpshortversion the full version \fi
\ificalpfullversion Algorithm \ref{algo:fancy_H2} \fi
for some $q'\in\omega(\log\Lambda)$.
Then, for any $\lambda\geq 0$ and $\Lambda\in\Theta(\sqrt{n})$  
where $n$ is the number of nodes in the graph,
$(\cH_1,\cH_2)$ is $(\Lambda-\lambda,\Theta(n),q)$-hard,
where $q\in\wsize{\Lambda}$.
\end{theorem}
\begin{proof}[Proof sketch]
    Identical proof structure to the proof of Theorem~\ref{thm:wraparound2}, except instead of invoking Corollary~\ref{cor:cylinder-subset-sustained-time} at the end, we derive a contradiction to Corollary~\ref{cor:notcylinder-subset-sustained-time}.
\end{proof}


\iffullversion
\fi


\section{Upper bounds and motivation for $\cc^{\alpha}$}\label{sec:cc-alpha}

Next, we motivate our notion of $\cca$ (defined in Definition~\ref{def:cca}). We show that both the honest party and the adversary may choose to use different pebbling strategies given different values of $\alpha$ even when $\alpha$ is constant. Furthermore, we show that both of our pebbling constructions of $\cH_1$ (given in Section~\ref{sec:constructions}) have the desirable feature that the honest party \emph{and} the adversary use the same strategy regardless of the size of $\alpha$.

\ificalpfullversion
\subsection{$\cc$ and $\cca$ consider cumulative cost of \emph{different strategies}}

	We present a graph family with in-degree-$2$ where the strategy that an adversary chooses to pebble an instance $G$ in the graph family differs depending on the $\alpha$ parameter of the $\cc^{\alpha}$ complexity measure. We show that in our case, for certain $\alpha$, we would choose to use constant space, whereas for other $\alpha$, using superconstant space is the preferred option. We define our graph family as follows:
\fi

\begin{gconstruction}\label{def:cc-alpha-diff}
We define a graph family $\sgf$ with bounded degree $2$ and arbitrary $n \in \mathbb{N}$ nodes such that the time-space tradeoff of a graph with $n$ nodes in the family is $T(S) \geq (\frac{n^c}{n^a}) (n^a - (S - 2)) (n^b) + n$ (where $S$ is the number of pebbles used to pebble the graph) where $0 \leq a, b, c < 1$, $b + c > a + 1$, $a < b, c$, and $n^c \approx n - n^{a+b}$. 

\begin{itemize}
\item Given a graph $\sg{2} = (V, E)$ with $n$ vertices, partition the set of vertices, $V$, into $2$ sets, $A$ and $C$ where $|A| = n^{a+b}$ and $|C| = n^c$ (since we know $n^c \approx n-n^{a+b}$, $n^c + n^{a+b} \approx n$). 
\item We arbitrarily order all vertices in $C$ in some order, $[v_i, \dots, v_n]$ and create edges $(v_j, v_{j+1}) \in E$ for all $j \in [i, n-1]$.
\item We arbitrarily order all vertices in $A$ in some order, $[v_1, \dots, v_{i-1}]$ and create edges $(v_{j}, v_{j+1}) \in E$ for all $j \in [1, i-2]$. 
\item We create edge $(v_{i-1}, v_{i})$.
\item Create edges $(v_{k}, v_l) \in E$ ($v_k \in A$ and $v_l \in C$) where $k \Mod n^b = 0$ and $l = n^{a + b} + \left(\frac{k}{n^b}\right) + (q -1)n^a$ for all integers $q \in \Big[1, \frac{n^c}{n^a}\Big]$.
\end{itemize}
\end{gconstruction}

\ificalpfullversion
	Fig.~\ref{fig:cc-alpha-diff} illustrates Graph Construction~\ref{def:cc-alpha-diff}.
\fi

\ificalpfullversion
	\begin{figure}[h]
	\centering
	\includegraphics[width=0.5\textwidth]{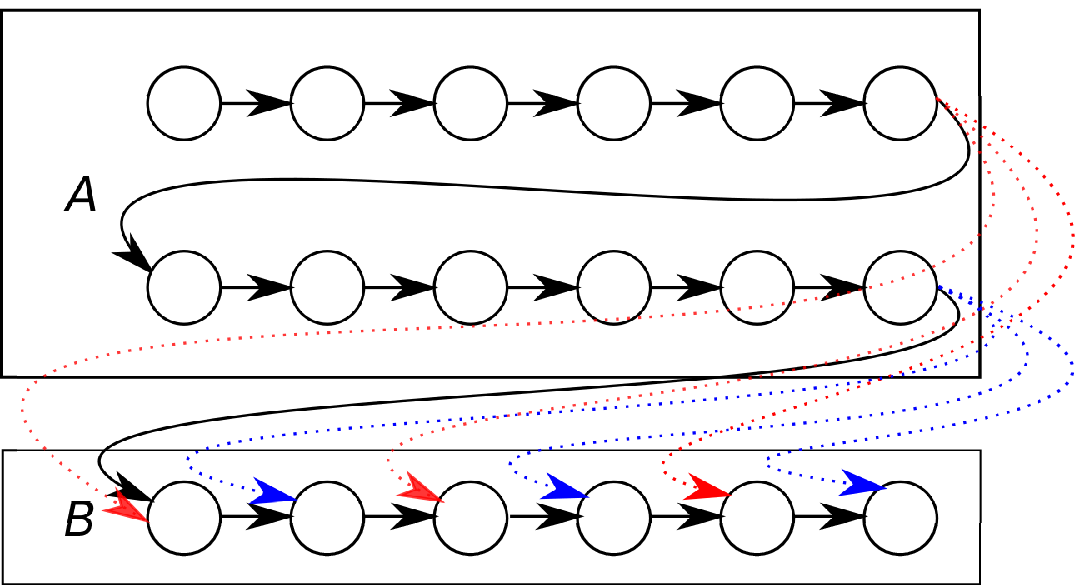}
	\caption{Graph Construction~\ref{def:cc-alpha-diff} with $n = 16$, $a = \frac{1}{4}$, $b = \frac{2}{3}$, $c = \frac{2}{3}$. For clarity, we depict $n^a = 2$, $n^b \approx 6$ and $n^c \approx 6$.}\label{fig:cc-alpha-diff}
	\end{figure}
\fi

We show that there are at least two pebbling strategies, $\s_1$ and $\s_2$, where an adversary would differ in his preferred strategy depending on $\alpha$ when using the $\cca$ complexity measure when $\alpha > \alpha'$ where $\alpha'$ is calculated with respect to the parameters of the graph family constructed from Graph Construction~\ref{def:cc-alpha-diff}.

\ificalpshortversion
	\begin{wrapfigure}{r}{0.4\textwidth}
	  	\centering
		\includegraphics[width=0.5\textwidth, angle=-90]{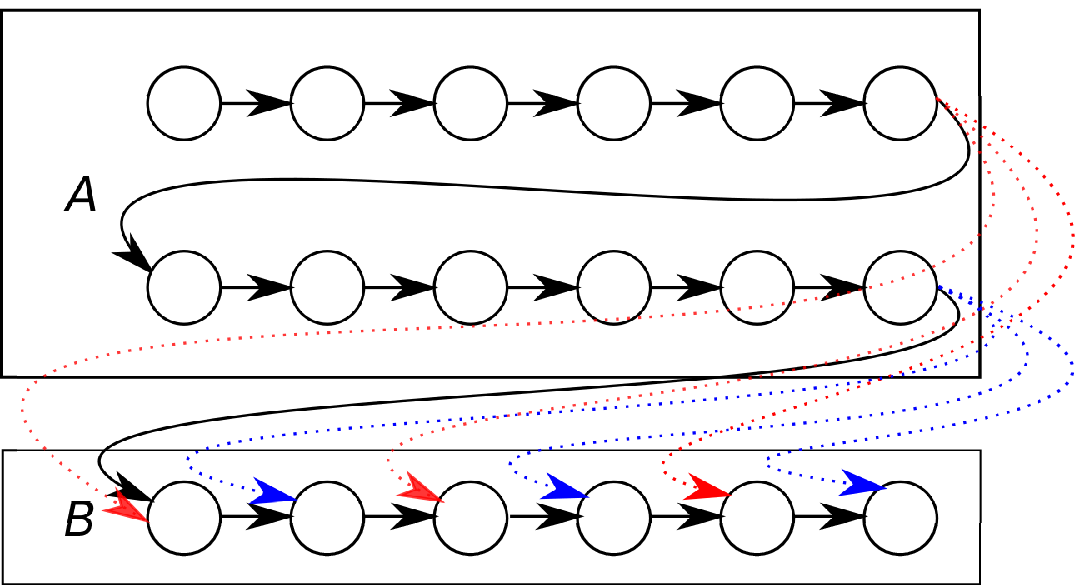}
		\caption{Graph Construction~\ref{def:cc-alpha-diff} with $n = 16$, $a = \frac{1}{4}$, $b = \frac{2}{3}$, $c = \frac{2}{3}$. For clarity, we depict $n^a = 2$, $n^b \approx 6$ and $n^c \approx 6$.}\label{fig:cc-alpha-diff}
	\end{wrapfigure}
\fi

\ificalpfullversion
\begin{lemma}\label{lem:cc-tradeoff}
Given a pebbling strategy $\s_1$ that uses constant space $S_1$, $\ti(\s_1) = \Theta(n^{b + c})$ where $\sg{2} \in \sgf$ in defined by Graph Construction~\ref{def:cc-alpha-diff}. 
\end{lemma}

\begin{proof}
Suppose that a constant $S_1$ pebbles can be on the graph at any particular time, then at most $S_1 - 1$ of the vertices in $A \subseteq V$ can be pebbled. It does not help to pebble the vertices in $C$ since all vertices in $C$ needs to be pebbled only once regardless of the pebbling strategy used. Since only the vertices $v_j \in A$ where $j \mod n^b = n^b - 1$ are connected to vertices in $C$, using the given $S_1$, the optimal placements are on vertices $v_j$ in order to minimize pebbling time since any extra space needs to be used to pebble $C$ and pebbling anywhere else results in greater pebbling time since the pebble needs to be moved to vertex $v_j$ by the pigeonhole principle. Given constant $S_1$ pebbles, there exist $v_j$ vertices that do not contain pebbles. Thus, each time one reaches a vertex in $C$ with predecessor $v_j \in A$ without a pebble, at least $n^b$ time must be spent to pebble it. Therefore, given $S_1$ pebbles, the total amount of time necessary to pebble $\sg{2}$ is $(\frac{n^c}{n^a})(n^a - S_1)(n^b) + n = \Theta(n^{b + c})$. 
\end{proof}

\begin{corollary}\label{cor:cc-tradeoff}
Given a pebbling strategy, $\s_1$, that uses constant space $S_1$, $\pcost(\s_1) = \Theta(n^{b + c})$ where $\sg{2} \in \sgf$ is constructed by Def.~\ref{def:cc-alpha-diff}. 
\end{corollary}

\begin{proof}
This follows immediately from Lemma~\ref{lem:cc-tradeoff} since constant space is used throughout the pebbling. 
\end{proof}

\begin{lemma}\label{lem:min-time}
Given a pebbling strategy $\s_2$ that uses space $S_2 = n^{a} + 1$, $\ti(\s_2) = \Theta(n)$ where $\sg{2} \in \sgf$ is constructed by Graph Construction~\ref{def:cc-alpha-diff}.
\end{lemma}

\begin{proof}
It is trivial to show that pebbling a line takes $\Omega(n)$ time since all nodes have to be pebbled at least once. We now show a strategy using $n^{a} + 1$ pebbles that uses $O(n)$ time. 

We start with the vertices in $A$ and pebble them in topological order, keeping pebbles on all $v_j \in A$ where $j \Mod n^b = n^b - 1$. There exists exactly $n^a$ vertices in $A$ by definition that are predecessors of vertices in $C$. Therefore, as we pebble the vertices in $A$ in topological order, we leave a pebble on each vertex $v_j$. When we pebble $C$ all predecessors of vertices in $C$ are either in $C$ or are pebbled in $A$. Therefore, we only need to pebble all vertices in $A$ and $C$ once, resulting in $\ti(\s_2) = \Theta(n)$.
\end{proof}

The following corollary is directly proven by the proof of Lemma~\ref{lem:cc-tradeoff}.

\begin{corollary}\label{cor:min-time-max-space}
Given a pebbling strategy, $\s_2$, that uses space $S_2 = n^{a} + 1$ and $\ti(\s_2) = \Theta(n)$, $\pcost(\s_2) = \Theta(n^{\alpha a + 1})$ where $\sg{2} \in \sgf$ is constructed by Graph Construction~\ref{def:cc-alpha-diff}.
\end{corollary}

\begin{lemma}\label{lem:cc-min-alpha-1}
When $\alpha = 1$, then $\cca(\sg{2}) = \Theta(n^{a + 1})$.  
\end{lemma}

\begin{proof}
Suppose in the case when $\alpha = 1$, we use a pebbling strategy, $\s$, that uses nonconstant space $s = o(n^a)$. Then, for each pebble, we pebble one of the vertices $v_j \in A$ where $j \Mod n^b = n^b - 1$. The resulting $\pcost(\s) = s(\frac{n^c}{n^a})(n^a - s)(n^b) + n$ which is minimized when $s = n^{a} + 1$ given $b + c > a + 1$ by definition of our graph family.
\end{proof}

\begin{lemma}\label{lem:cc-alpha-exists}
For all $a, b, c$, there exists an $\alpha'$ such that for all constant $\alpha > \alpha'$, $\cca(\sg{2}) = \Theta(n^{b + c})$.
\end{lemma}

\begin{proof}
Given a pebbling strategy, $\s$, that uses space $s = \omega(1)$, the pebbling cost is then $\pcost(\s) = s^{\alpha}(\frac{n^c}{n^a})(n^a - s)(n^b) + n$. When $\alpha > 1$, $\pcost(\s) = \Theta(\min(s^{\alpha}n^{b+c}, n^{\alpha a + 1})) = \omega(n^{b + c})$ when $\alpha > \frac{b + c}{a}$ and $s = \omega(1)$. Therefore, only for $s = O(1)$, does the pebbling cost become $\pcost(\s) = \Theta(n^{b + c})$ when $\alpha' = \frac{b + c}{a}$. Since $a, b, c$ are constants, for all $\alpha > \alpha'$, $\cca(G) = \Theta\left(n^{b + c}\right)$.
\end{proof}

From the above two lemmas, we immediately get the following theorem regarding the $\cca$ of the constructions given different constant values of $\alpha$.
\fi

\begin{theorem}\label{thm:different-strategies}
Given a graph $G = (V, E)$ as constructed by Graph Construction~\ref{def:cc-alpha-diff}, when $\alpha = 1$, $\cca(\sg{2}) = \Theta(n^{a + 1})$ but when $\alpha > \alpha'$ for some constant $\alpha'$, $\cca(\sg{2}) = \Theta(n^{b+c})$.
\end{theorem}

As an immediate result of the above, there exists a point for constants $a, b, c$ that the adversary chooses a different strategy to pebble a graph for different constant values of $\alpha$ (we can pick values of $a, b,c$ such that $\alpha'$ can be reduced even down to $\alpha' \geq 3$). 

\ificalpfullversion
	\subsection{Upper bounds for $\cca$}

\fi
\ificalpshortversion
	\subh{Upper bounds for $\cca$}
\fi
We prove a tighter upper bound for $\cca$ when $\alpha$ is a constant than the trivial upper bound of $n^{\alpha + 1}$. We first note that $n^{\alpha + 1}$ is a trivial upper bound on the $\cca(\sg{2})$ of a graph, $\sg{2}$, since at any timestep $\p(\sg{2}, T) \leq n$ and the algorithm runs for $\ti(\sg{2}, |T|) \leq n$ given $n$ space is used throughout. Therefore, $\cca(\sg{2}) \leq n^{\alpha + 1}$ for all graphs $\sg{2}$. We now prove a tighter upper bound using the general pebbling algorithm described in~\cite{AB16} as $\mathsf{GenPeb}(G, S, g, d)$.

\ificalpfullversion
We formulate a simplified version of the $\mathsf{GenPeb}(G, S, g, d)$ procedure which we call the $\mathsf{GenPeb}(G)$ procedure. At a high-level the $\mathsf{GenPeb}(G)$ algorithm proceeds as follows (see~\cite{AB16} for more detail).

\begin{definition}[$\mathsf{GenPeb}(G)$:]

~

\begin{enumerate}
\item There exists a subset $S$ of $|S| \leq \frac{2\alpha n \log{\log{n}}}{\log{n}}$ vertices (for large enough $n$ where $2\alpha \log{\log{n}} \leq \log{n}$) such that $\mathsf{depth}(G-S) \leq \frac{n}{\log^{\alpha}{n}}$ (Lemma 6.1, 6.2 in~\cite{AB16},~\cite{Val77}). 
\item \textbf{Balloon Phase:} Pebble all nodes up to depth $\frac{n}{\log^{\alpha}{n}}$ (depth measured from the last light phase) until all immediate descendants lie in $S$.
\item \textbf{Light Phase:} When all immediate descendants lie in $S$, remove all pebbles from nodes not in $S$ and not on parents of the next nodes to be pebbled. Continue in the light phase until a node not in $S$ must be pebbled. 
\item Repeat the above until no more nodes need to be pebbled. 
\end{enumerate}
\end{definition}
\fi

\ificalpfullversion
\begin{lemma}\label{lem:total-balloon-pebbles}
Let $s_{\Sigma}$ be the total number of pebbles used in the balloon phase (the sum of the number of pebbles used in all balloon phases) and $s_{\Sigma \backslash S}$ be the total number of pebbles used in the balloon phases on all nodes $v \not\in S$. Then, $s_{\Sigma \backslash S} \leq n$.  
\end{lemma}

\begin{proof} 
This proof is trivial since at most $n$ pebbles can be the graph at any time.
\end{proof}

\begin{lemma}\label{lem:upper-bound-heavy}
Let $\Sigma \backslash S$ be the subgraph of $\sg{2} = (V, E)$ which is pebbled during the balloon phase and whose vertices are not in $S$. Then, $\cca(\Sigma \backslash S) \leq \frac{n^{\alpha + 1}}{\log^{\alpha}{n}}$ . 
\end{lemma}

\begin{proof}
By Lemma~\ref{lem:total-balloon-pebbles}, the number of pebbles necessary to pebble $\Sigma \backslash S$ is at most $n$: $s_{\Sigma \backslash S} \leq n$. Therefore, we can compute $\cca(\Sigma \backslash S) \leq \sum_{B_i \in \mathcal{B}} (|B_i|)^{\alpha} \leq n^{\alpha}(\frac{n}{\log^{\alpha}{n}}) = \frac{n^{\alpha + 1}}{\log^{\alpha}{n}}$ given a series of balloon phase pebble configurations $\mathcal{B}$ where $\sum_{B_i \in \mathcal{B}} |B_i| = n$ and $B_0 \bigcupdot \cdots \bigcupdot B_{|\mathcal{B}|} = V$. 
\end{proof}

\begin{lemma}\label{lem:upper-bound-light}
Let $\cca(S)$ be the cost of pebbling $S$ in both the light and the balloon phases. The $\cca(S)$ of the light and balloon phases is at most $O\left(\frac{n^{\alpha + 1} (\log{\log{n}})^{\alpha}}{\log^{\alpha}{n}}\right)$. 
\end{lemma}

\begin{proof}
The total amount of time that light and balloon phases last in which nodes in $S$ are pebbled is at most $n$ timesteps since a number greater than $n$ implies that $|S| \geq n$ which is impossible since the number of nodes in the graph is $n$. In the light phases, at most $2|S| = \frac{4\alpha n \log{\log{n}}}{\log{n}}$ pebbles are kept on the graph since each node has bounded in-degree $2$. Therefore, $\cca{(\sg{2})} \leq \frac{4^{\alpha} \alpha^\alpha n^{\alpha + 1} (\log{\log{n}})^{\alpha}}{\log^{\alpha}{n}}$. 
\end{proof}
\fi

\ificalpfullversion
	\begin{theorem}\label{lem:cc-bound}
	For any bounded in-degree-$2$ graph, $\cca{(\sg{2})} = O\left(\frac{n^{\alpha + 1}(\log{\log{n}})^{\alpha}}{\log^{\alpha}{n}}\right)$ for constant $\alpha \geq 1$. 
	\end{theorem}
\fi
\ificalpshortversion
\begin{theorem}\label{lem:cc-bound}
	For any in-degree-$2$ graph and $\alpha\geq1$, $\cca{(\sg{2})} = O\left(\frac{n^{\alpha + 1}(\log{\log{n}})^{\alpha}}{\log^{\alpha}{n}}\right)$.
	\end{theorem}
\fi

\begin{proof}
This follows directly from Lemmas~\ref{lem:upper-bound-heavy} and~\ref{lem:upper-bound-light}.
\end{proof}
\ificalpfullversion
    \subsection{Asymptotically tight sequential lower bound for $\alpha = 1$}

\fi
\ificalpshortversion
    \subh{Asymptotically tight lower bound for $\alpha = 1$}
\fi
\ificalpcorrectedversion
    \subh{Asymptotically tight lower bound for $\alpha = 1$}
\fi
We give an explicit construction of a graph that achieves asymptotically tight lower bound (up to $\log{\log{n}}$ factors) in $\cca$ that matches our upper bound provided in Section~\ref{sec:cc-alpha} for $\alpha = 1$ and in~\cite{AB16,ABP17sustained} \emph{when considering the sequential pebbling model}. Previous constructions~\cite{AB16, ABP17} ignored $\log{\log{n}}$ factors and were not tight up to such factors in the parallel model. Because we consider the sequential pebbling model (and not the parallel model) in proving our lowerbound below, our results are incomparable to these previous lower bound results in the parallel model. 
Our graph constructions are new, and their tightness in the parallel pebbling model is an open question.

In our construction, we make use of the stacked superconcentrators constructed in \cite[\S4]{LT82} except that the vertices are connected in some topological order (blowing up our graph by only a constant factor of $6$ if we replace all degree $3$ nodes with a height $3$ pyramid).

\begin{gconstruction}\label{def:cca-lower-bound-construct}
Let $C(n, k)$ be a stacked superconcentrator with $k$ layers where $C_i$ is the $i$-th linear superconcentrator. We create the following edges between nodes. Let $\mathrm{T}$ be a topological sort order of the vertices in $C(n, k)$. Create edges $(v_i, v_{i+1})$ where $v_i$ is the vertex immediately preceding $v_{i+1}$ in $\mathrm{T}$. Replace all degree $3$ nodes with pyramids of height $3$.
\end{gconstruction}

\ificalpfullversion
    It was proven in~\cite{LT82} (Theorem 4.2.6) that given $S \leq \frac{n}{20}$ pebbles, $k$ layers, and $n$ nodes in each linear superconcentrator per layer, the pebbling time, $T(n, k, S)$, of pebbling $C(n, k)$ is lower bounded by:

    \begin{align*}
    T(n, k, S) = n \Omega\left(\left(\frac{nk}{64S}\right)^k\right).
    \end{align*}

    In our construction defined by Def.~\ref{def:cca-lower-bound-construct}, we first let $S = c_1(N\log{\log{N}}/\log{N})$ (for some constant $c_1$), $n = 20S$, $k = \floor{N/S}$, and we get a graph $C(n, k)$ with $\Theta(N)$ vertices. Thus, we obtain the following tradeoff for this graph given $S$ pebbles:

    \begin{align*}
    T \geq S \Omega\left(\frac{N}{S}\right)^{\Omega(N/S)}
    \end{align*}

    for $S \leq c_2\left(\frac{N\log{\log{N}}}{\log{N}}\right)$ for some constant $c_2$ where $c_2 < c_1$.

    Thus, we notice two main characteristics of our graph. If $S \geq c_1\left(\frac{N\log{\log{N}}}{\log{N}}\right)$, then the time it takes to pebble the graph is $O(N)$ since the width of the graph is $\Theta\left(\frac{N\log{\log{N}}}{\log{N}}\right)$. Second, if $S \leq c_2 \left(\frac{N\log{\log{N}}}{\log{N}}\right)$ then $S$ pebbles are used to pebble the graph for $\omega(N)$ time by Theorem 4.2.6 of~\cite{LT82}. Note that if the tradeoff is sufficiently great, then we achieve our stated lower bound. To prove our stated lower bound, we modify the proof for Theorem 4.2.5 of~\cite{LT82} so that we account for $\cca$ instead of just the time-space tradeoff. Minimizing the equation for tradeoff in terms of $\alpha = 1$ and showing that the cost is greater than the cost of when $S \geq c_1 \left(\frac{N\log{\log{N}}}{\log{N}}\right)$ and the cumulative complexity for when $S \geq c_1\left(\frac{N\log{\log{N}}}{\log{N}}\right)$ is $\Theta\left(\frac{N^2 \log{\log{N}}}{\log{N}}\right)$ then provides us with the lower bound we want.

    We use the same notation as that used in the proof of Theorem 4.2.5 in~\cite{LT82}. Let $n$ be the number of outputs of the superconcentrator $C(n, k)$ and $k$ be the number of copies of the linear superconcentrators (number of levels in the stack of superconcentrators) in $C(n, k)$. We number the parts of $C(n, k)$ similarly to how they are numbered in the proof of Theorem 4.2.5, let $C_i$ be the $i$-th copy of the linear superconcentrators that composes $C(n, k)$. We consider the outputs of $C_k$ as numbered in the order in which they are first pebbled. Let $z_i$ be the time that output $i$ (where $1 \leq i \leq n$) is pebbled. Therefore, $z_0 = 0$ and $z_{n + 1} = \ti(C(n, k), S)$. 
Then, let $[z'_i, z''_i]$ be the interval of time starting with the $z'_i$-th move and ending with the $z''_i$-th move where $z_{i -1} \leq z'_i \leq z''_i \leq z_i$. Let $p_i$ be the minimum number of pebbles on $C_k$ in the interval $[z_{i-1}, z_{i}]$ for $1 \leq i \leq n$ and where $p_0 = 0$, $p_{n+1} = 0$, and $p_{i} \leq S$ for all $i$ in the valid range. 

    We first note that since we do not remove any vertices or edges (only add edges to the construction to maintain the topological order and to ensure that at most one additional pebble is added to the graph at each time step), all properties of the graph with respect to $n$ as proven in~\cite{LT82} still hold (i.e. adding edges does not change the linear superconcentrator properties of the graphs). Hence, we restate some of the key theorems and lemmas in~\cite{LT82} that will allow us to prove the lower bound in $\cca$ when $\alpha =1$ that we seek.

    We restate the definition of a good interval given in~\cite{LT82} below:

    \begin{definition}[Good Intervals~\cite{LT82}]\label{def:good}
    An interval $[i, j] \subset [1, n]$ is \emph{good} if it fulfills the following three requirements:

    \begin{align}
    p_i \leq \frac{j - i}{2},\\
    p_{j+1} \leq \frac{j - i}{2},\\
    p_k > \frac{j - i}{8} \text{ for } i < k \leq j.
    \end{align}
    \end{definition}

    We also restate one key lemma relating to good intervals below:

    \begin{lemma}[Lemma 4.2.3~\cite{LT82}]\label{lem:good-bound}
    During the good interval $[i, j]$ at least $n - 2S$ different outputs of $C_{k-1}$ are pebbled. Only $S - 1 - \floor{\frac{j-i}{8}}$ pebbles are available to pebble the $n-2S$ different outputs of $C_{k-1}$.
    \end{lemma}

    We also restate a combinatorial lemma proved in~\cite{LT82} that will allow us to prove a recursive relation on $\cca$ (which will subsequently allow us to provide a bound for our construction).

    \begin{lemma}[Lemma 4.2.4~\cite{LT82}]\label{lem:bound-on-good}
    Let $r \leq n$. We can find a set of disjoint good intervals in $[1, r]$ that covers at least $\frac{r}{4} - S - p_{r+1}$ elements of $[1, r]$.
    \end{lemma}

    Finally, we adapt a theorem based on a simple application of BLBA that provides a (not quite tight enough) lower bound on the time necessary to pebble our constructed graph given $S$ pebbles and provide a proof for our construction defined in Graph Construction~\ref{def:cca-lower-bound-construct}.

    \begin{theorem}[Theorem 4.2.1~\cite{LT82}]\label{thm:blba-simple}
    In order to pebble all outputs of $C(n,k)$ as defined in Graph Construction~\ref{def:cca-lower-bound-construct} using $S$ black pebbles, $2 \leq S \leq \frac{n-1}{4}$ (starting with any configuration of pebbles on the graph), we need $T$ placements where 

    \begin{align*}
    T \geq n\left(\frac{n}{10S}\right)^k.
    \end{align*} 
    \end{theorem}


    Using these lemmas, we now write our final recursive theorem for the $\cca$ of our construction. 

    \begin{theorem}\label{thm:cca-lower-1}
    Let $\cca(N, k, S)$ be the $\cca$ (when $\alpha = 1$) necessary to pebble all the outputs of $C(n, k)$ (recall that the topological sort of the vertices requires that for the last output to be pebbled, all other outputs must be pebbled) with $S \leq \frac{n}{20}$ pebbles. Then,

    \begin{align}
    T(n, 1, S) \geq \frac{n^2}{10S} \\
    T(n, k, S) \geq \min_{(x_1, \dots, x_m) \in D_k} \sum_{1 \leq i \leq m} T\left(n, k-1, S-1-\left\lfloor\frac{x_i - 1}{8}\right\rfloor\right) \text{ for } k > 1,\\
    \cca(N, k, S) \geq \min_{D_1, \dots, D_k} \sum_{1 \leq j \leq k}\sum_{(x_1, \dots, x_m) \in D_j} \left\lfloor\frac{x_i - 1}{8}\right\rfloor \left(T\left(n, j - 1, S-1-\left\lfloor\frac{x_i - 1}{8}\right\rfloor\right)\right)\\
    \geq \min_{D} \sum_{(x_1, \dots, x_m) \in D} \left\lfloor\frac{x_i - 1}{8}\right\rfloor T\left(n, k-1, S-1-\left\lfloor\frac{x_i - 1}{8}\right\rfloor\right).\label{eq:cca-lb}
    \end{align}

    where $D_i$ is an index set that contains all the ways in which we can select a large number of good intervals. Specifically,

    \begin{align*}
    D_i = \left\{(x_1, \dots, x_m) | m>\frac{n}{64S}, 1 \leq x_i \leq 8S - 6 \text{ for } 1 \leq i \leq m, \text{ and } \sum_{1 \leq i \leq m} x_i \geq \frac{n}{8}\right\}.
    \end{align*}
    \end{theorem}

    \begin{proof}
    The proof for the expression for $T(n, k, S)$ follows directly from Theorem 4.2.5 in~\cite{LT82}. 

    Now we prove the expression for $\cca$ of $C(n, k)$ for the case when $S \leq n/20$. For each good interval, at least $\left\lfloor \frac{x_i-1}{8}\right\rfloor$ pebbles must remain on $C_k$ while $C_1, \dots, C_{k-1}$ are pebbled with the remaining $S-1 - \left\lfloor\frac{i-1}{8}\right\rfloor$ pebbles. Therefore, the $\cca$ when $\alpha = 1$ of the good period with length $x$ is $\left\lfloor \frac{x_i-1}{8}\right\rfloor T\left(n, k-1, S-1-\left\lfloor\frac{x_i - 1}{8}\right\rfloor\right)$. By Lemma~\ref{lem:bound-on-good}, we have that the total length of the disjoint good intervals is at least $n/8$ (since $p_{r+1} \leq S$ and $n/4 - 2S \geq n/8$). Thus, summing over the $\cca$ for all good intervals and minimizing over all possible allocations of good intervals gives a lower bound on the $\cca$ for $C_k$ which is a lowerbound on the $\cca$ when $\alpha = 1$ of the entire graph. 
    %
    \end{proof}

    \begin{lemma}\label{lem:closed-form-lb}
    When $S = c_1\left(\frac{N\log\log{N}}{\log{N}}\right)$ for some constant $c_1$, $n = 20S$, $k = \left\lfloor N/S \right\rfloor$ and we create a graph according to Graph Construction~\ref{def:cca-lower-bound-construct}, $C(n, k)$ with $\Theta(N)$ vertices, 
    \begin{align}
    \cca(N, k, S) \geq \min_{D} \sum_{(x_1, \dots, x_m) \in D} \left\lfloor\frac{x_i - 1}{8}\right\rfloor \left(20S\left(\frac{20S(\floor{N/S}-1)}{c\left(S-1-\left\lfloor\frac{x_i - 1}{8}\right\rfloor\right)}\right)^{\floor{N/S}-1}\right)
    \end{align}

    for $S \leq c_2 \left(\frac{N\log\log{N}}{\log{N}}\right)$ for some constants $c$ (specified in the proof) and $c_2 < c_1$.
    \end{lemma}

    \begin{proof}
    We know from~\cite{LT82} that the expression for $T(n, k, S)$ is lower bounded by $T(n, k, S) \geq n\left(\frac{nk}{cS}\right)^k$ for some constant $c \geq 10$. Therefore, we can substitute this expression into our Eq.~\ref{eq:cca-lb} to obtain the following expression:

    \begin{align*}
    \cca(N, k, S) 
    \geq \min_{D} \sum_{(x_1, \dots, x_m) \in D} \left\lfloor\frac{x_i - 1}{8}\right\rfloor \left(n\left(\frac{n(k-1)}{c(S-1-\left\lfloor\frac{x_i - 1}{8}\right\rfloor)}\right)^{k-1}\right).
    \end{align*}

    Substituting our values as stated above then gives 

    \begin{align}\label{eq:summation-lower-bound}
    \cca(N, k, S) \geq \min_{D} \sum_{(x_1, \dots, x_m) \in D} \left\lfloor\frac{x_i - 1}{8}\right\rfloor \left(20S\left(\frac{20S(\floor{N/S}-1)}{c\left(S-1-\left\lfloor\frac{x_i - 1}{8}\right\rfloor\right)}\right)^{\floor{N/S}-1}\right)
    \end{align}

    for some number of pebbles used that is less than $n/20$; or in other words, for some constant $c_2$, $S \leq c_2 \left(\frac{N\log\log{N}}{\log{N}}\right)$ where we determine the exact values of $c_1$ and $c_2$ later on (since the exact values of $c_1$ and $c_2$ also depend on the types of linear superconcentrators used in each of the $k$ layers of our construction).
    \end{proof}

    \begin{lemma}\label{lem:closed-form-lb-exp}
    Given $S \leq c_2\left(\frac{N\log{\log{N}}}{\log{N}}\right)$ for some constant $c_2$ where $c_2 \left(\frac{N\log{\log{N}}}{\log{N}}\right) < n/20$,

    \begin{align}
    \cca(N, k, S) \geq \frac{c_2}{8}\left(\frac{N\log{\log{N}}}{\log{N}}\right)\left(20S \left(\frac{20S \left(\left\lfloor \frac{N}{S}\right\rfloor - 1\right)}{c(S-1)}\right)^{\left\lfloor\frac{N}{S}\right\rfloor -1}\right).\label{eq:closed-form-lb-exp}
    \end{align}
    \end{lemma}

    \begin{proof}
    We assume for the sake of contradiction that there exists a closed formed lowerbound for the equation where some $x_i > 1$. Suppose there exists some good period with length $x_i > 1$, then the term 

    \begin{align*}
    \frac{x_i}{8}\left(20S \left(\frac{20S \left(\left\lfloor \frac{N}{S}\right\rfloor - 1\right)}{c(S-1-\left\lfloor \frac{x_i - 1}{8}\right\rfloor)}\right)^{\left\lfloor\frac{N}{S}\right\rfloor -1}\right)
    \end{align*}

    is in the summation of the calculation of $\cca(N, k, S)$ (see Eq.~\ref{eq:summation-lower-bound}). We can replace the term with the following:

    \begin{align*}
    x_i\left(\frac{1}{8} \left(20S \left(\frac{20S \left(\left\lfloor \frac{N}{S}\right\rfloor - 1\right)}{c(S-1)}\right)^{\left\lfloor\frac{N}{S}\right\rfloor -1}\right)\right)
    \end{align*}

    which results in a smaller $\cca(N, k, S)$ a contradiction, therefore no values of $x_i$ are greater than $1$ and the closed form lower bound is that as stated in Eq.~\ref{eq:closed-form-lb-exp}.
    \end{proof}

    \begin{lemma}\label{lem:cca-lb-lessS}
    Given $S \leq c_2 \left(\frac{N\log{\log{N}}}{\log{N}}\right)$ for some constant $c_2$ where $c_2 \left(\frac{N\log{\log{N}}}{\log{N}}\right) < n/20$, $\cca$ when $\alpha = 1$ is $\omega\left(\frac{N^2\log{\log{N}}}{\log{N}}\right)$.
    \end{lemma}

    \begin{proof}
    From Lemma~\ref{lem:closed-form-lb-exp}, the $\cca$ when less than $c_2 \left(\frac{N\log{\log{N}}}{\log{N}}\right)$ pebbles are used is lower bounded by the closed form expression, 

    \begin{align}
    \cca(N, k, S) \geq \frac{c_2}{64}\left(\frac{N\log{\log{N}}}{\log{N}}\right)\left(20S \left(\frac{20S \left(\left\lfloor \frac{N}{S}\right\rfloor - 1\right)}{c(S-1)}\right)^{\left\lfloor\frac{N}{S}\right\rfloor -1}\right).\label{eq:lb-cca}
    \end{align} 

    We know that the lower bound given in Eq.~\ref{eq:lb-cca} is $\Theta\left(\frac{N\log{\log{N}}}{\log{N}}\left(S\left(\frac{N}{S}\right)^{\frac{N}{S} -1}\right)\right)$. 


    Given $S \leq \frac{N\log{\log{N}}}{\log{N}}$ pebbles, we now prove that the $\cca$ of our construction for $\alpha = 1$ is $\omega\left(\frac{N^2 \log{\log{N}}}{\log{N}}\right)$. We know that $S\left(\frac{N}{S}\right)^{\frac{N}{S} - 1} = \omega(N)$ for all $S \leq c_2\left(\frac{N\log{\log{N}}}{\log{N}}\right)$. Therefore, $\cca(N, k, S) = \omega\left(\frac{N^2\log{\log{N}}}{\log{N}}\right)$.
\end{proof}
\fi

\ificalpshortversion
We use the proofs presented in~\cite{LT82} of the closed-form time-space trade-offs of the superconcentrator construction used above. For more detail, refer to the full version attached.
\fi

\ificalpfullversion
    \begin{theorem}\label{thm:cca-lb}
    Given $S > c_2 \left(\frac{N\log{\log{N}}}{\log{N}}\right)$, $\cca$ when $\alpha = 1$ is $\Omega\left(\frac{N^2\log{\log{N}}}{\log{N}}\right)$. Therefore, $\cca(G) = \Theta\left(\frac{N^2\log{\log{N}}}{\log{N}}\right)$ in the sequential\footnote{Erratum: An earlier version of this paper stated the theorem for general pebbling strategies, not just sequential ones. The proof herein is unchanged from that earlier version, and proves the theorem only for sequential strategies.} pebbling model where $G$ is given by our Graph Construction~\ref{def:cca-lower-bound-construct} above.
    \end{theorem}
\fi
\ificalpshortversion
\begin{theorem}\label{thm:cca-lb}
    Given $S > c_2 \left(\frac{N\log{\log{N}}}{\log{N}}\right)$, $\cca$ when $\alpha = 1$ is $\Omega\left(\frac{N^2\log{\log{N}}}{\log{N}}\right)$. Therefore, $\cca(G) = \Theta\left(\frac{N^2\log{\log{N}}}{\log{N}}\right)$ in the sequential pebbling model where $G$ is given by our Graph Construction~\ref{def:cca-lower-bound-construct} above.
    \end{theorem}
\fi

\begin{proof}
Let $S$ be large enough that a single linear superconcentrator with $n$ output nodes can be pebbled in almost linear time. In this case, we use the simple BLBA argument presented in Theorem 4.2.1 of~\cite{LT82} to prove that in this case, $\cca(N, k, S) = \Omega\left(\frac{N^2\log{\log{N}}}{\log{N}}\right)$ since each $C_i$ in the construction of $C(n, k)$ as defined in Graph Construction~\ref{def:cca-lower-bound-construct} along with the edges joining $C_{k-1}$ with $C_k$ is an $n$-superconcentrator. 

The BLBA theorem as proven in~\cite{LT82} proves a tradeoff in time with respect to the number of pebbles in the starting and ending configuration of the graph. Let $S_a$ be the starting number of pebbles on the graph and $S_b$ be the ending number of pebbles on the graph. Suppose that $S_b = 0$ for the sake of lowerbounding our cumulative complexity. Then $c_2 \left(\frac{N\log{\log{N}}}{\log{N}}\right) < \min_{1\leq i \leq k}\left(S^i_a\right) \leq S$ by our theorem statement where $S^i_a$ is the starting pebble configuration for level $i$. Suppose that $S^{c_i}_a \leq c_2 \left(\frac{N\log{\log{N}}}{\log{N}}\right)$ for $L$ levels (i.e. for some set of levels in $[c_1, \dots, c_L]$), then $\cca(n, i, S)$ is given by Lemma~\ref{lem:cca-lb-lessS} for the $L$ values. Using Lemma~\ref{lem:cca-lb-lessS}, we see that in order for the bound from Lemma~\ref{lem:cca-lb-lessS} to not hold, we must have $L = o(N/S)$. But, then, $N/S - o(N/S) = \Theta(N/S)$ layers are pebbled with $S^i_a > c_2\left(\frac{N\log{\log{N}}}{\log{N}}\right)$ pebbles. Therefore, we achieve the same asymptotic bound by considering $c_2 \left(\frac{N\log{\log{N}}}{\log{N}}\right) < \min_{1\leq i \leq k}\left(S^i_a\right) \leq S$.


Thus, by BLBA, we know that 

\begin{align}
T(n, 1, S) \geq \max\left(1, \frac{n-2S}{2S + 1}\right)\\
T(n, i, S) \geq n\left(\max\left(1, \frac{n}{10S}\right)\right)^i\\
\cca(N, k, S) \geq \sum_{1\leq i \leq k: S^i_a} S^i_a T(n, i - 1, S - S^i_a) \max\left(1, \left(\frac{n-2S^i_a}{2S^i_a + 1}\right)\right)\\
\geq n \min_{1 \leq i \leq k: S^i_a}\left(S^i_a\right)\max\left(1, \frac{n-2S}{2S + 1}\right)(k-1)\label{eq:blba-lb-eq}
\end{align}

We can simplify in the last step since $T(n, i - 1, S - S^i_a) \geq n$ for all $1 \leq i\leq k$. Furthermore, by our argument above, we know that $\min_{1 \leq i \leq k: S^i_a}\left(S^i_a\right) = \Theta(S)$.


When $n = c_1 \left(\frac{N\log{\log{N}}}{\log{N}}\right)$, $S > c_2 \left(\frac{N\log{\log{N}}}{\log{N}}\right)$, and $k = \frac{\log{N}}{\log{\log{N}}}$, then Eq.~\ref{eq:blba-lb-eq} simplifies to $\Omega\left(\frac{N^2 \log{\log{N}}}{\log{N}}\right)$ for some predefined $c_2$ and $c_1$. Otherwise, the time of pebbling is $N$ using $c_1\left(\frac{N\log{\log{N}}}{\log{N}}\right)$ pebbles resulting in $\cca$ when $\alpha = 1$ to be $\Theta\left(\frac{N\log{\log{N}}}{\log{N}}\right)$.
\end{proof}

\paragraph{Case of $\alpha = 2$} We briefly note that the above construction does not asymptotically achieve tightness for $\alpha = 2$ by our current analysis. This is due to the fact that when $\alpha = 2$, Lemma~\ref{lem:cca-lb-lessS} no longer holds due to the fact that $\left(\frac{N\log\log{N}}{\log{N}}\right)\cdot \left(\frac{\log{N}}{\log{\log{N}}}\right)^{\frac{\log{N}}{\log{\log{N}}}} = o(N^2)$.

\begin{openquestion}
Does there exist a bounded in-degree graph family that has $\cca$ for $\alpha \geq 2$ that meets the upper bound?
\end{openquestion}

\ificalpfullversion
  \section{Implementation}\label{sec:implementation}
We implemented our \emph{cylinder} construction defined in Def.~\ref{def:wrappyramid}. We choose to implement this construction because it is simplest of the constructions we present for $\cH_1$, yet achieves memory and time bounds comparable to our more complicated construction. In implementing the pebbling construction, we seek to minimize the runtime of $\cH_1$ while maximizing its output size.  This leads to some interesting tradeoffs as well as an observation about static-memory-hardness and the random oracle model in general.

\subh{Overview of implementation}
First, we map an entire row of labels (i.e., labels in a particular layer of our construction) in our cylinder construction defined in Definition~\ref{def:wrappyramid} to an array of bits in memory of length $l$; the computer is not aware of the label sizes nor the boundaries between labels.  We can implement a serialized pebbling algorithm by iteratively reading $n$ bits, starting at offset $f$, sending the read bits to a hash function, writing the $n$ bits returned from the hash function at offset $f$, and finally incrementing $f$ by additive $n$ bits for the next round. This process is repeated until the end of the string, which constitutes one row of the cylinder construction.  This procedure of processing the rows of the cylinder is repeated once for every row of the cylinder DAG.

\subh{Degree and label size in the implemented function}
In the code, there are no direct arguments for pebbling the graph in terms of number of pebbles, label size or degree; instead we parametrize by the following parameters (which we choose): the total size of the array $l$, the input size $i$ and the output size $n$ of the hash function.  

The label size is $n$ which is also the output length of the hash function.  Every hash produces one label. The number of labels per row is then $\frac{l}{n}$.  $l$ should be a multiple of $n$ so that there are no partial labels at the end of the array.

The indegree of the wraparound pyramid is $\frac{i}{n}$.  Here as well, $i$ should be a multiple
 of $n$ so that the degree is an integer and this maps cleanly to the pebbling model when we consider ingree $2n$ (ie constant indegree $2$ in the pebbling model).

The height needed for the wraparound pyramid is then the array size divided by the difference between input and output sizes, or $\frac{l}{i-n}$.  The input size must be greater than the output size for the height to be defined.  This corresponds with the requirement that the degree must be 2 or greater for the pyramid construction to provide memory guarantees.

\subh{Extra memory usage}
In order to implement the wraparound pyramid, memory usage needs to be greater than that stated in theoretical model due to necessary memory allocations in the hardware.  The leftmost bits of the array need to be copied and appended to the right side, so that the lower level input values are available to the final hash interations which consume the wraparound inputs.  This increases the memory needed by $i-n$.

\subh{Reducing number of hashes}
The runtime of evaluating $\cH_1$ is determined by the number of hash transformations called as very little other computation is done.  The number of hashes per row is $\frac{l}{n}$, and the number of rows is $\frac{l}{i-n}$, giving a total number of hash calls as

\[\left(\frac{l}{n}\right)\left(\frac{l}{i-n}\right) = \frac{l^2}{ni-n^2}\]

$l^2$ indicates the expected time requirement proportional to the square of the output size.  To optimize the time for a given $l$, we look at the denominator, $ni - n^2$, noting that $i > n$, keeping this positive.  To reduce the time taken, increase the input size $i$.  Graphically, this makes sense as descending from the top of the wraparound pyramid, the higher degree will quickly cover the entire width of a row.  However, in practice we cannot increase $i$ and maintain the memory-hard properties, which leads to an interesting divergence between the random oracle model and real-world hash functions.

\subh{Data busses to the random oracle}
One aspect which is rarely discussed in the Random Oracle Model is the exact process by which one makes a call to the oracle.  Does the query need to be sent to the oracle via a parallel bus, all bits at once, or is the query sent via a serial bus, one bit at a time?  If serially, can we send some of the bits, then wait a while, and send the rest?  We are not aware of literature dealing with these mechanics of data transmission to and from the oracle; however in our case it's quite relevant.  If serial transmission is allowed, $i$ can be made arbitrarily large without needing to store the whole row of the wraparound pyramid in memory.  For each bit of a label, as soon as it is computed it can be sent to all the oracles using that bit as an input, and promptly forgotten; the oracles act as a memory cache.  The memory-hardness proofs implicity assume an oracle model where the entire query is handed over simultaneously to the oracle, and as such, any query to the oracle must exist in its entirety in memory before the query is made.

In practice, real-world hash functions resemble a serial-bus oracle much more closely than a parallel bus oracle.  Whether we're referring to Merkle-Damgard, Sponge construction, or other methods, all widely used hash functions are built out of fixed length one-way functions.  The internal state of a hash function can thus act as a data cache for the purposes of the pebbling graph.  For a high-degree node, the left predecessors can be fed to the hash function and forgotten before the right predecessors are known.  Since the internal state of the hash function has a fixed size, this defeats the memory hardness promised by the pebbling construction.

\subh{Instantiating the random oracle}
We used blake2b, a fast and well-known hash function, for our implementation.  Blake2b has an internal state size of 1024 bits, so we were able to set $i$ to 1024 bits while keeping the memory hardness.  $n$ was set to 512 bits, giving a 2-degree pebbling graph.  Decreasing either $i$ or $n$ would lead to inefficient use of the function.  It would seem that hash functions with a larger internal state size, capable of supporting a larger $i$ would be faster for this usage, but it is not as clear as larger state sizes may correlate with slower evaluation of a single hash.

\begin{figure*}[t!]
	\centering
    \begin{subfigure}[b]{0.33\textwidth}
        \centering
		\includegraphics[width=0.99\textwidth]{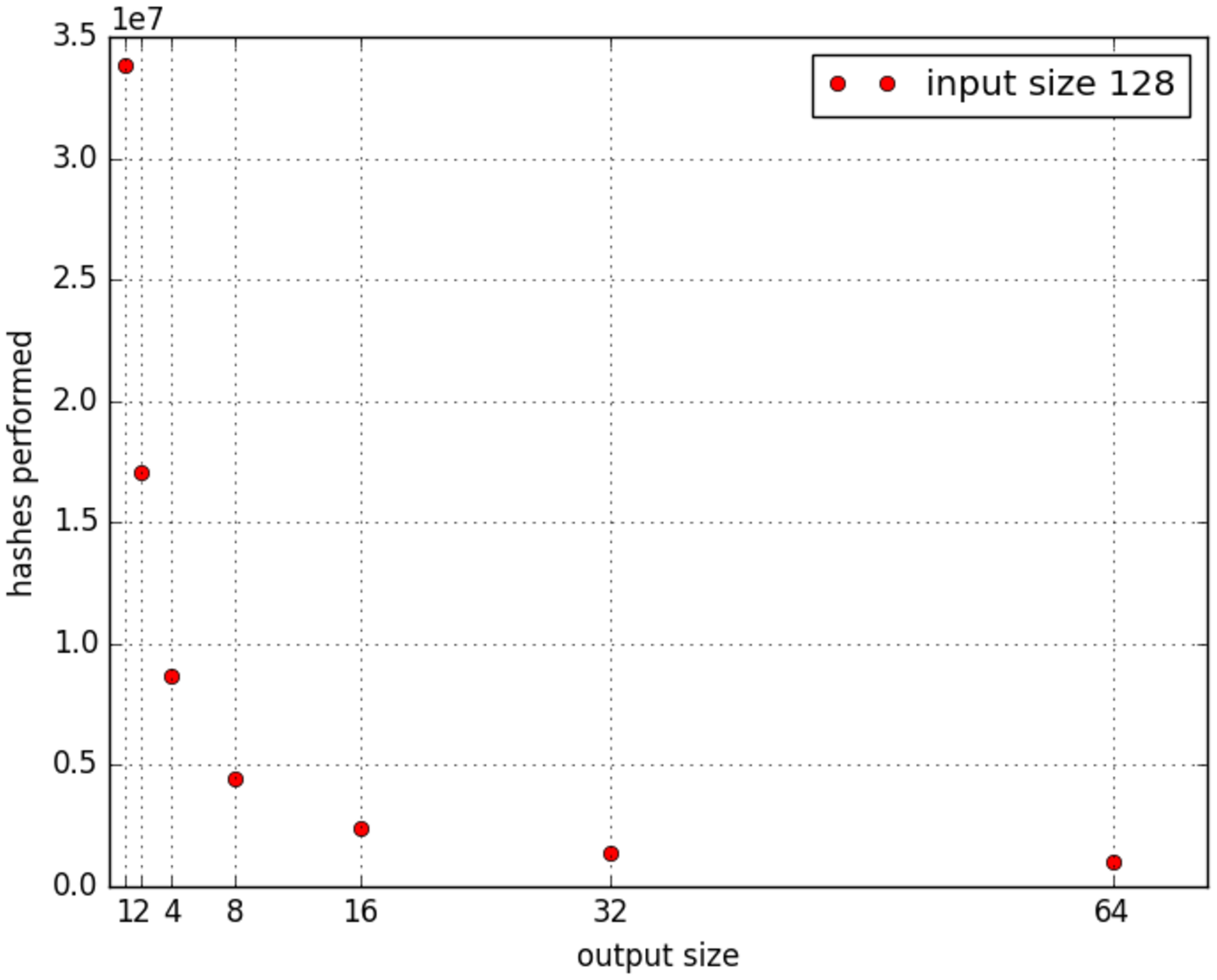}
		\caption{Runtime vs hash output size, 128 byte input, 65Kbyte row}
    \end{subfigure}%
    ~
    \begin{subfigure}[b]{0.33\textwidth}
        \centering
		\includegraphics[width=0.99\textwidth]{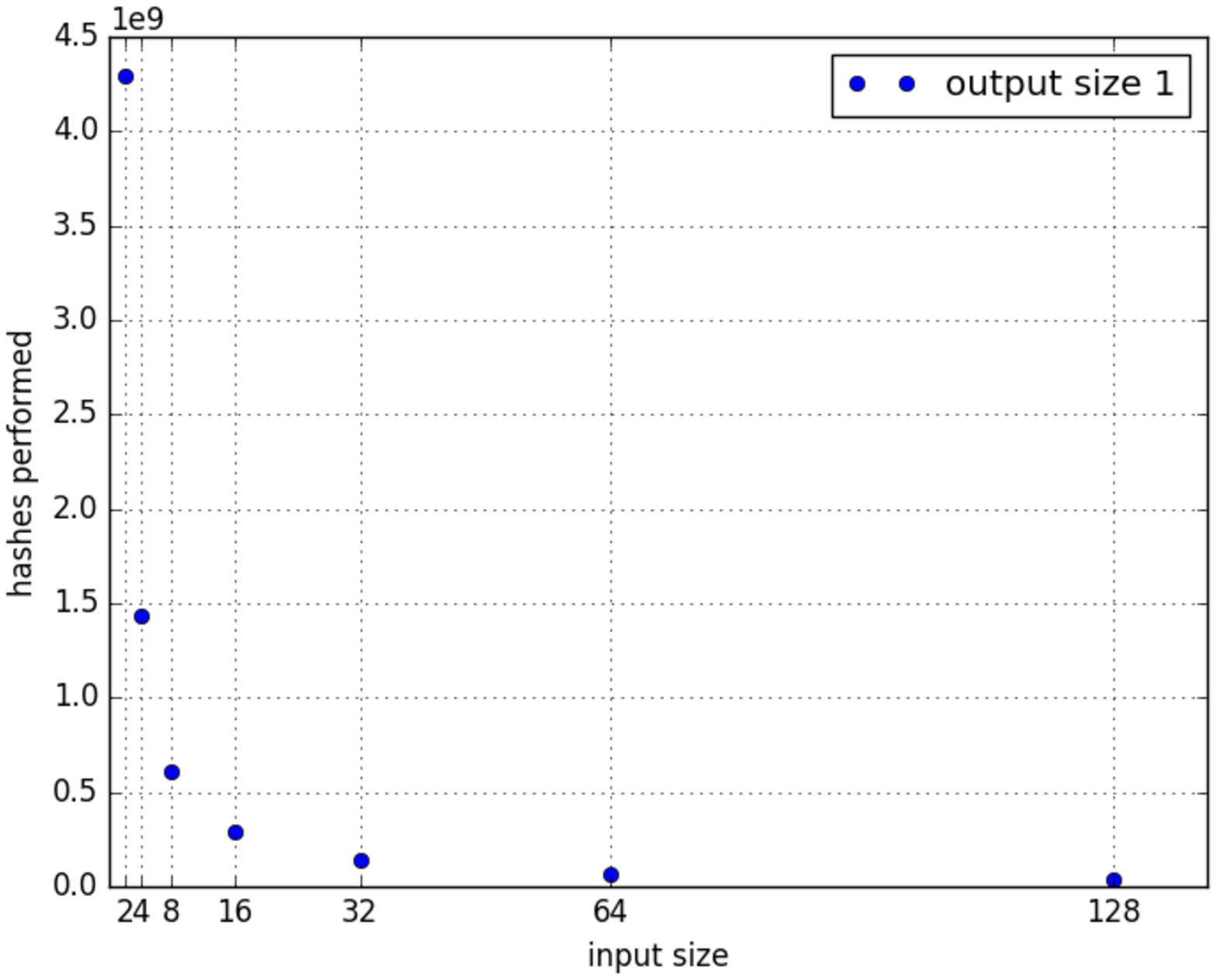}
		\caption{Runtime vs hash input size, 1 byte output, 65Kbyte row}
    \end{subfigure}%
    ~
    \begin{subfigure}[b]{0.33\textwidth}
        \centering
		\includegraphics[width=0.99\textwidth]{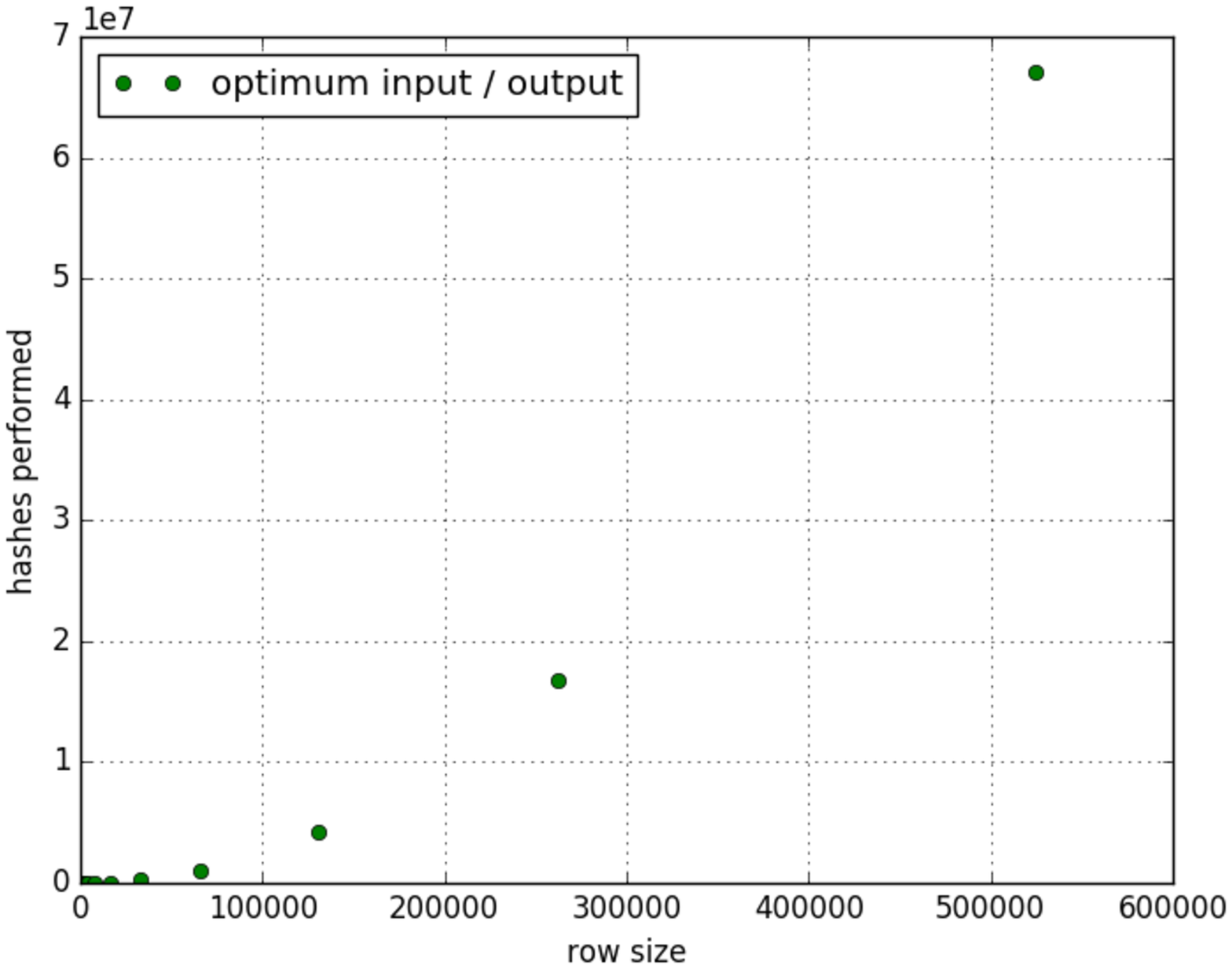}
		\caption{Runtime vs row size, 64 byte output, 128 byte input}
    \end{subfigure}

    \vspace{-5pt}
    
    \caption{Evaluation of cylinder implementation}
\end{figure*}

We implement a data-dependent memory access pattern for $\cH_2$. 
Other papers (eg catena~\cite{LW15} and balloonhash~\cite{BCS16}) have identified security vulnerabilities due to data-dependent memory access patterns which can leak information about the password to an attacker with incomplete access to the physical system evaluating the password hashing function.  These attacks occur because of the variable time taken to evaluate the function based on the input data, primarily due to the automatic caching of data inside the CPU.

Our $\cH_2$ function, while implementing a data-dependent memory access pattern, should not be vulnerable to cache timing attacks.  In practical usage, the output of $\cH_1$ will be very large with respect to the number of queries $\cH_2$ performs on the data; for any single evaluation of $\cH_2$, nearly all of the $\cH_1$ data will go unread.  Because of this sparse access, data will be read once and not used again before being evicted from the cache.  The probability that an input to $\cH_2$ results in a collision, and multiple reads from the same memory region, is thus modeled by $(\wsize{\Lambda})^{-Q}$
where $\Lambda$ is the size of the output of $\cH_1$, and $Q$ is the number of oracle calls made by $\cH_2$.

Inputs resulting in cache hits should be rare, and knowledge of a cache hit during $\cH_2$ evaluation give a bounded advantage to the attacker expressed by
\begin{align*}
\frac{\text{total number of access patterns with } n \text{ collisions}}{\text{total number of zero-collision access patterns}} \ .
\end{align*}
We observe that this could potentially be a large advantage in practice because attackers do not need to perform memory lookups into the set of $\cH_1$ outputs in order to detect collisions. That is, an attacker still has to do lots of hash operations, but their memory requirement could go down significantly.
\fi

\ificalpfullversion
  \section*{Acknowledgements}

\fi
\ificalpshortversion
  \subh{Acknowledgements}
\fi
We are grateful to Jeremiah Blocki for valuable feedback on an earlier version of this paper. We thank Ling Ren for helpful discussions. We also thank Erik D. Demaine and Shafi Goldwasser for their advice and discussions related to this paper. Finally, we thank our anonymous reviewers 
for insightful comments.

Sunoo's research is supported by the Center for Science of
  Information STC (CSoI),
  an NSF Science and Technology Center (grant agreement CCF-0939370),
  MACS project NSF grant CNS-1413920, and
  a Simons Investigator Award Agreement dated 2012-06-05.

\bibliographystyle{alpha}
\bibliography{ref}

\ificalpshortversion\else
  \appendix
  \iffullversion
    \input{pebble-sliding-rule}
  \fi
  \section{Details of SHF construction with short labels} 
\label{appx:h2-alternative}

\begin{algorithm}[ht!]\caption{$\cH^{q'}_2$}\label{algo:fancy_H2}
    On input $(1^\sec,x)$ and given oracle access to $\Seek_{R}$ 
    (where $R$ is the string outputted by $\cH_1$):
    \begin{enumerate}
    \item Let $\wsize{R}=|R|/w$ be the length of $R$ in words.
    \item Query the random oracle to obtain $\rho_0=\RO(x)$ and $\rho_1=\RO(x+1)$.
    \item Use $\rho_0$ to sample randomly $\iota_1,\dots,\iota_{q'}\in[\wsize{R}]$.
    \item Query the $\Seek_R$ oracle to obtain $\{y'_i=\Seek_R(\iota_i)\}_{i\in[q']}$.
    \item Output $(y'_1||\dots||y'_{q'})\oplus\rho_1$.
    \end{enumerate}
\end{algorithm}

\begin{definition}[$q''$-labeling]\label{def:labeling}
    Let $G=(V,E)$ be a DAG with maximum in-degree $\deg$, let
    $\LabelSet$ be an arbitrary ``label set,'' and define
    $\OSet(\deg,\LabelSet)=\left(V\times \bigcup_{\deg'=1}^{\deg}\LabelSet^{\deg'}\to \LabelSet\right)$.
    Let $\RO|_{q''}$ be the function that outputs the first $q'$ bits of the output of $\RO$.
    For any function $\RO\in\OSet(\deg,\LabelSet)$ and any label $\zeta\in\LabelSet$,
    the $(\RO,\zeta,q'')$-labeling of $G$ is a mapping 
    $\lab_{\RO,\zeta}:V\to\LabelSet$ defined recursively as follows.\footnote{%
    We abuse notation slightly and also invoke $\lab_{\RO,\zeta}$ on \emph{sets} of vertices, in which case the output is defined to be a tuple containing the labels of all the input vertices, arranged in lexicographic order of vertices.}
    $$\lab_{\RO,\zeta}(v)=\begin{cases}
    \RO|_{q''}(v,\zeta) & \mbox{ if } \indeg(v)=0 \\
    \RO|_{q''}(v,\lab_{\RO,\zeta}(\pred(v))) & \mbox{ if } \indeg(v)>0 
    \end{cases}\ .$$
    \end{definition}

    Then, we define our family of random oracle functions defined from our hard to pebble graph family constructions.
    
\begin{definition}[$q''$-graph function family]\label{def:gff}
    Let $n=n(\sec)$ and let $\gfsec=\{\glong\}_{\sec\in\NN}$ be a graph family. We write $\OSet_{\delta,\sec}$ to denote the set $\OSet(\deg,\zo^\sec)$ as defined in Definition \ref{def:labeling}.
    The \emph{$q''$-graph function family} of $\gf$ is the family of oracle functions 
    $\cF^{q''}_{\gf}=\{\ff_{\g}\}_{\sec\in\NN}$ where
    $\ff_{\g}=\{f^{\RO}_{\g}:\zo^\sec\to(\zo^\sec)^z\}_{\RO\in\OSet_{\delta,\sec}}$ and $z=z(\sec)$ is the number of sink nodes in $\g$. The output of $f^{\RO}_{\g}$ on input label $\zeta\in\zo^\sec$ is defined to be
    $$f^{\RO}_{\g}(\zeta)\defeq\lab_{\RO,\zeta}(\sink(\g))\ ,$$
    where $\sink(\g)$ is the set of sink nodes of $\g$.
\end{definition}
  \ificalpcorrectedversion\else
    \section{Regular and Normal Strategies} 

Here, we restate three theorems and prove briefly their equivalent formulation in the parallel model for the parallel model adapted from theorems in~\cite{GLT79,DL17} proven in the sequential model.

We first restate the definitions for normal and regular strategies:

\begin{definition}[Frugal Strategy~\cite{GLT79}]\label{def:frugal-strategy}
Given a DAG $G = (V, E)$, a \emph{frugal strategy} is a pebbling strategy with no unnecessary placements. In particular, the following are true of any frugal pebbling strategy:

\begin{enumerate}
\item At all times after the first placement on a vertex $v$, some path from $v$ to the goal vertex contains a pebble.
\item At all times after the last placement on a vertex $v$, all paths from $v$ to the goal vertex contain a pebble. 
\item The number of placements on a nongoal vertex is bounded by the total number of placements on its successors. 
\end{enumerate}
\end{definition}

\begin{definition}[Normal Strategy~\cite{GLT79}]\label{def:normal-strategy}
A \emph{normal strategy} is a standard pebbling strategy that is frugal and it pebbles each pyramid $P$ in $G$ as follows: after the first pebble is placed on $P$, no placement or removal of pebbles occurs outside $P$ until the apex of $P$ is pebbled and all other pebbles are removed from $P$. No new placement occurs on $P$ until after the pebble on the apex of $P$ is removed.
\end{definition}

\begin{theorem}[Normal Strategy Conversion~\cite{GLT79}]\label{thm:normal-strategy}
If the goal vertex is not inside a pyramid, any standard pebbling strategy can be transformed into a normal pebbling strategy without increasing the number of pebbles used in both the sequential and parallel pebbling models.  
\end{theorem}

\begin{proof}
The proof of this statement in the sequential model is given in~\cite{GLT79}. We now prove this statement in the parallel model. By our proof of Lemma~\ref{lem:sec-par-equiv}, any sequential strategy can be simulated trivially by a parallel strategy; therefore, if any pebbling strategy can be transformed into a sequential normal pebbling strategy, then any pebbling strategy can be transformed into a parallel normal pebbling strategy. 
\end{proof}

We now define regular pebbling strategies:

\begin{definition}[Regular Strategy~\cite{DL17}]\label{def:regular-strategy}
Given a DAG $G = (V, E)$, a \emph{regular strategy} is a standard pebbling strategy that is frugal and after the first pebble is placed on any road graph $R_w \in G$, no placements of pebbles occurs outisde $R_w$ until the set of desired outputs of $R_w$ all contain pebbles and all other pebbles are removed from $R_w$.
\end{definition}

By the same argument as given for the proof of Theorem~\ref{thm:normal-strategy}, we can prove the equivalent for parallel regular pebbling strategies.

\begin{theorem}[Regular Strategy Conversion~\cite{DL17}]\label{thm:regular-strategy}
Given a DAG $G = (V, E)$, if each input, $i_j \in \left\{i_1, \dots, i_w\right\}$, to a road graph has at most $1$ predecessor, any standard pebbling strategy that pebbles a set of desired outputs, $O \subseteq \left\{o_1, \dots, o_w\right\}$, at the same tiime can be transformed into a \emph{regular strategy} without increasing the number of pebbles used.
\end{theorem}

In addition, we prove this stronger theorem about the pebbling space complexity of pyramid graphs below than the theorems provided in~\cite{GLT79,Nor15} that will be useful for determining the pebbling space complexity of pyramids in the magic pebble game.

\begin{theorem}\label{thm:pyramid-space}
Given a pyramid graph $\Pi_{h}$ with $h$ levels where level $1$ has $h$ nodes and level $h$ has $1$ node. Given $S$ pebbles and if all $S$ pebbles are placed on level $i$ of the pyramid and $S < h + 1 - i$, then the apex of the pyramid cannot be pebbled using the rules of the standard pebble game.
\end{theorem}

\begin{proof}
Given $S < h+1-i$ pebbles on the $i$-th layer of a height $h$ pyramid, we know that the $i$-th level of the pyramid forms a height $h + 1 - i$ height pyramid with the apex. Thus, by the pebbling space complexity of pyramids, $h + 1 - i$ pebbles are necessary on level $h+1-i$ in order to pebble the apex. 
\end{proof}

  \fi
  \iffullversion
    \input{dropoff}
    \input{ro-mem-req}
    \input{decremental}
  \fi
  \fi

\end{document}